\documentclass[11pt]{article}

\usepackage[margin=1in]{geometry}
\usepackage[numbers,sort&compress]{natbib}
\usepackage{amsmath}
\usepackage{amssymb}
\usepackage{amsthm}
\usepackage{algorithm}
\usepackage{algorithmic}
\usepackage{booktabs}
\usepackage{enumitem}
\usepackage{multicol}
\usepackage{xcolor}
\usepackage{soul}
\usepackage{graphicx}
\usepackage[small]{caption}
\usepackage{microtype}
\usepackage{hyperref}
\usepackage{doi}
\usepackage{tikz}
\usepackage{eufrak,bbold}
\hypersetup{
    colorlinks=true,
    linkcolor=black,
    citecolor=black,
    filecolor=black,
    urlcolor=[HTML]{0088CC},
    pdftitle={},
}

\theoremstyle{definition}

\newtheorem{thm}{Theorem}
\newtheorem{que}{Question}
\newtheorem{prop}{Proposition}[section]
\newtheorem{lthm}[prop]{Theorem}
\newtheorem{lem}[prop]{Lemma}
\newtheorem{cor}[prop]{Corollary}

\newtheorem{rem}[prop]{Remark}

\newtheorem{eg}[prop]{Example}
\newtheorem{dfn}[prop]{Definition}

\newtheorem{obs}[prop]{Observation}
\newtheorem{ntn}[prop]{Notation}

\newcommand{\dft}[1]{\textbf{\textit{#1}}}

\newcommand{\abs}[1]{\left|#1\right|}

\newcommand{\set}[1]{\left\{#1\right\}}
\newcommand{\sucht}{\,\middle|\,}

\newcommand{\N}{\mathbf{N}}
\newcommand{\R}{\mathbf{R}}
\newcommand{\Z}{\mathbf{Z}}

\newcommand{\NP}{\mathrm{NP}}
\newcommand{\Poly}{\mathrm{P}}
\newcommand{\shBIS}{\mathrm{\#BIS}}
\newcommand{\shDOWN}{\mathrm{\#DOWN}}
\newcommand{\shSM}{\mathrm{\#SM}}
\newcommand{\shP}{\mathrm{\#P}}

\newcommand{\ConstructPathDecomposition}{\mathsf{ConstructPathDecomposition}}
\newcommand{\ConstructInstance}{\mathsf{ConstructInstance}}

\DeclareMathOperator{\pw}{pw}

\DeclareMathOperator{\down}{down}
\DeclareMathOperator{\width}{width}
\DeclareMathOperator{\Anc}{Anc}

\renewcommand{\th}{{}^{\mathrm{th}}}

\newcommand{\calI}{\mathcal{I}}
\newcommand{\calM}{\mathcal{M}}
\newcommand{\calP}{\mathcal{P}}
\newcommand{\calR}{\mathcal{R}}
\newcommand{\calX}{\mathcal{X}}
\newcommand{\calY}{\mathcal{Y}}

\newcommand{\myvec}[1]{\vec{#1}}
\newcommand{\veca}{\myvec{a}}
\newcommand{\vecb}{\myvec{b}}
\newcommand{\vecc}{\myvec{c}}
\newcommand{\vecm}{\myvec{m}}
\newcommand{\vecv}{\myvec{v}}
\newcommand{\vecw}{\myvec{w}}
\newcommand{\vecx}{\myvec{x}}
\newcommand{\vecy}{\myvec{y}}
\newcommand{\vecalpha}{\myvec{\alpha}}
\newcommand{\vecbeta}{\myvec{\beta}}

\DeclareMathOperator{\bound}{Bound}
\DeclareMathOperator{\attr}{Attr}
\DeclareMathOperator{\range}{Range}
\DeclareMathOperator{\rank}{rank}
\DeclareMathOperator{\klist}{List}

\DeclareMathOperator{\orank}{\mathbf{f}}
\newcommand{\rmax}{\orank_{\mathrm{max}}}
\newcommand{\rmin}{\orank_{\mathrm{min}}}
\DeclareMathOperator{\ext}{ext}

\title{Stable Matchings with Restricted Preferences: \\ Structure and Complexity}

\author{
  Christine T.\ Cheng\\
  University of Wisconsin, Milwaukee\\
  ccheng@uwm.edu
  \and
  Will Rosenbaum\\
  Amherst College\\
  wrosenbaum@amherst.edu
}      

\date{\today}

\begin{document}

\maketitle
\thispagestyle{empty}

\begin{abstract}
It is well known that every stable matching instance $I$ has a {\it rotation poset} $\calR(I)$ that can be computed efficiently and the downsets of $\calR(I)$ are in one-to-one correspondence with the stable matchings of $I$.  Furthermore, for every poset $\calP$, an instance $I(\calP)$ can be constructed efficiently so that the rotation poset of $I(\calP)$ is isomorphic to $\calP$.  In this case, we say that $I(\calP)$ {\it realizes} $\calP$.   Many researchers exploit the rotation poset of an instance to develop fast algorithms or to establish the hardness of stable matching problems.  

In order to gain a parameterized understanding  of the complexity of sampling stable matchings, Bhatnagar et al.~\cite{Bhatnagar2008-sampling} introduced stable matching instances whose preference lists are restricted but nevertheless model situations that arise in practice.  In this paper, we study four such parameterized restrictions; our goal is to characterize the rotation posets that arise from these models:  

  \begin{enumerate}
  \item \emph{$k$-bounded}, where each agent has at most $k$ acceptable partners;
  \item \emph{$k$-attribute}, where each agent $a$ has an associated vector $\veca \in \R^k$, and $a$ ranks agents in decreasing order according to a linear function $\varphi_a : \R^k \to \R$;
  \item \emph{$(k_1, k_2)$-list}, where the men and women can be partitioned into $k_1$ and $k_2$ sets (respectively) such that within each set all agents have identical preferences;
  \item \emph{$k$-range}, where for each woman $w$, the men's rankings of $w$ differ by at most $k-1$ (and symmetrically for the women's rankings of each man $m$).
  \end{enumerate}
  We prove that there is a constant $k$ so that {\it every} rotation poset is realized by some instance in models 1--3  ($k \geq 3$ for $k$-bounded, $k \geq 6$ for $k$-attribute, and $k_1 \geq 2, k_2 = \infty$ for $(k_1, k_2)$-list, respectively).  We describe efficient algorithms for constructing such instances given the Hasse diagram of a poset. As a consequence, the fundamental problem of counting stable matchings remains $\shBIS$-complete even for these restricted instances.

 For $k$-range preferences, we show that a poset $\calP$ is realizable if and only if the Hasse diagram of $\calP$ has pathwidth bounded by functions of $k$. Using this characterization, we show that the following problems are fixed parameter tractable when parameterized by the range of the instance: exactly counting and uniformly sampling stable matchings, finding median, sex-equal, and balanced stable matchings.
\end{abstract}

\clearpage

\setcounter{page}{1}
\setcounter{tocdepth}{1}

\tableofcontents

\section{Introduction}
\label{sec:introduction}

In the last 60 years, the stable marriage problem and its variants have emerged as central topics in economics, computer science, and mathematics. The basic problem is phrased as follows. Let $M$ and $W$ be two disjoint sets of $n$ agents, traditionally referred to as men and women, respectively. Each agent has a preference list that ranks (all or some) members of the opposite sex.   The goal is to find a matching $\mu$ between $M$ and $W$ such that no $m \in M$ and $w \in W$ mutually prefer each other to their partners in $\mu$. Such a matching is called \emph{stable}. In their seminal work, Gale and Shapley~\cite{Gale1962-college} proved that a stable matching $\mu$ always exists and provided an $O(n^2)$ time algorithm for finding one. This run-time has been shown to be optimal in a variety of computational models~\cite{Gonczarowski2019-stable, Ng1990-lower}.

While a stable matching instance of size $n$ may have $2^{\Theta(n)}$ stable matchings~\cite{Knuth1976-marriage, Karlin2018-simply}, the Gale-Shapley algorithm outputs only two kinds of stable matchings---the man-optimal/woman-pessimal and the woman-optimal/man-pessimal stable matchings.  That is,  these stable matchings are the best for one group of agents but the worst for the other group. To address this inequity, researchers  investigated various notions of ``fair''  stable matchings~\cite{Knuth1976-marriage, Irving1987-efficient, Gusfield1989-stable, Kato1993-complexity, Feder1995-stable, Teo1998-geometry, Cheng2008-generalized, Cheng2010-understanding, Manlove2013-algorithmics, McDermid2014-sex, Cheng2016-eccentricity, Gupta2017-treewidth} and studied the problem of counting and sampling stable matchings uniformly at random~\cite{Irving1986-complexity, Gusfield1987-three, Chebolu2012-complexity, Bhatnagar2008-sampling}. For all of these problems, crucial insights are gained by understanding the combinatorial structure of the set of stable matchings.
 
Irving and Leather~\cite{Irving1986-complexity} showed that every stable matching instance $I$ has an associated rotation poset $\calR(I)$ that determines the structure of the set of stable matchings for $I$. In particular, they proved that the stable matchings of $I$ are in one-to-one correspondence with the downsets of $\calR(I)$ and that $\calR(I)$ can be computed efficiently. Furthermore, given any finite poset $\calP$, there is a stable matching instance $I = I(\calP)$ whose rotation poset is isomorphic to $\calP$, and $I$ can be computed efficiently as well. We say that such an instance $I$ \dft{realizes} $\calP$.  

Many problems related to finding a fair stable matching are NP-hard. They include computing a median stable matching~\cite{Teo1998-geometry, Cheng2008-generalized, Cheng2010-understanding}, a balanced stable matching~\cite{Feder1995-stable} or a sex-equal stable matching~\cite{Kato1993-complexity, McDermid2014-sex}.   A natural approach then is to consider stable matching instances whose preference lists are restricted but  nevertheless model situations that arise in practice.   Below, we consider three such models from the work of Bhatnagar et al.~\cite{Bhatnagar2008-sampling}.  

\begin{itemize}
\item In the \emph{$k$-attribute model}, each agent has an associated $k$-dimensional profile and their preference lists are determined by applying a linear function to the profiles of their potential partners.  The model is motivated by online dating sites.  In this context, participants are frequently asked an extensive set of questions. Some  answers are  used to create a participant's profile while other answers are used to formulate a function that ranks possible dates according to the others' profiles.  

\item In the \emph{$k$-range model}, there is an objective ranking for each set of agents.  Each person ranks agents from the other group to within $k$ of their objective ranks. This model captures the scenario when participants make use of ``official rankings" to create their preference lists.  Students, for instance, might use their state's ranking of nearby high schools to guide their choices, while the schools might base their preferences according to the students' test scores. The participants do not have to copy the official rankings exactly; but they cannot deviate that much from the objective ranks when they make their choices.  

\item In the \emph{$k$-list model}, each set of agents can be partitioned into $k$ groups such that agents within each group have identical preference lists.  The model applies to situations when students who aspire to be engineers all have the same preferences, students who plan to pursue a business-related major have the same preferences, etc.  Similarly, high schools with a STEM focus may rank students the same way while high schools with a focus on the arts may have their own rankings. 
\end{itemize}

Bhatnagar et al.~\cite{Bhatnagar2008-sampling} studied the problem of sampling of stable matchings using the Markov Chain Monte Carlo (MCMC) method. For general instances, however, the method can take exponential time to converge to equilibrium. Such worst-case behavior occurs, for example, when the rotation poset of the instance  has $d$ bottom elements $b_1, b_2, \hdots, b_d$, $d$ top elements $t_1, t_2, \hdots, t_d$  and a single element that $c$ such that $b_i \prec c \prec t_j$ for all $i, j \in \{1, 2, \hdots, d\}$ with $d = \theta(n)$. We refer to this poset as the star poset $\calP_d$.

When the preference lists of the agents are restricted, we expect the rotation posets of the instances to be restricted as well.  Yet Bhatnagar et al.\ showed that there are instances from each of the three models described above that have $\calP_d$ as a rotation poset  even when the parameter $k$ is a small constant (i.e., $k = 2$ in the $k$-attribute model, $k = 5$ in the $k$-range model and $k = 4$ in the $k$-list model).  That is, even under narrow conditions, MCMC is not a feasible method for sampling stable matchings.  

Following~\cite{Bhatnagar2008-sampling}, Chebolu et al.~\cite{Chebolu2012-complexity} considered the problem of counting the number of stable matchings ($\shSM$) for instances in the $k$-attribute model and a related $k$-Euclidean model (introduced by Bogomolnaia and Laslier~\cite{Bogomolnaia2007-euclidean}). They showed that $\shSM$ is $\shBIS$-complete even when the preference lists are restricted to the $3$-attribute and $2$-Euclidean models. Thus, under a widely believed conjecture of Dyer et al.~\cite{Dyer2004-relative}, $\shSM$ does not admit a fully polynomial randomized approximation scheme (FPRAS). Chebolu et al. established their result for the $3$-attribute model by showing that for any bipartite graph $G$, there is a $3$-attribute stable matching instance whose rotation poset's Hasse diagram is isomorphic to $G$.\footnote{The intent is for the independent sets of $G$ to be in one-to-one correspondence with the downsets of the rotation poset.}  

Bhatnagar et al.~\cite{Bhatnagar2008-sampling} and Chebolu et al.~\cite{Chebolu2012-complexity} established their hardness results by constructing stable matching instances with restricted preference lists that realized \emph{certain} families of rotation posets.  Our goal is to go a step further: to \emph{characterize} the family of rotation posets realized by restricted families of stable matching instances.

\subsection{Our Contributions}

%% \textcolor{red}{Should we say something about the fact that in all of these models we are not fixing $n$?  In the Bhatnager et al. paper, they say that in the 2-attribute model, there are rotations that cannot arise but this is from restricting $n$.  Maybe not in this section but later subsections?} \\

In this paper, we study the expressive power of restricted preference list models in terms of the rotation posets realized by instances of the models. We characterize the rotation posets for the $k$-attribute and $k$-range models, as well as the following preference list models:
\begin{description}
\item[$k$-bounded model:] Each agent's preference list has length at most $k$. 
\item[$(k_1, k_2)$-list model:] The men and women can be partitioned into $k_1$ and $k_2$ groups respectively such that within each group all agents have identical preference lists.   When we allow the women (or the men) to have any kind of preference lists, we refer to the model as the $(k_1, \infty)$-list (or the $(\infty, k_2)$-list) model. Notice that the $(k_1, k_2)$-list model generalizes Bhatnagar et al.'s $k$-list model---their $k$-list is precisely our $(k_1, k_2)$-list model with $k_1 = k_2 = k$.  
\end{description}

%\textcolor{red}{I'm on the fence  with this.  Should we stick with the $(k_1, k_2)$-list model or with the $(k_1, \infty)$-list model?  We never really use the former although it's a generalization of the $k$-list model.  We use the latter but it's a more restrictive version of the $k$-list model.} 

%Notice that the $(k_1, k_2)$-list model generalizes Bhatnagar et al.'s $k$-list model -- their $k$-list is precisely our $(k_1, k_2)$-list model with $k_1 = k_2 = k$.  

Except for the $k$-bounded model, we shall assume that the preference lists of all agents in these models are complete.\footnote{The case for incomplete preference lists will be addressed in the discussion section of the paper.} We now describe our main contributions.

%emphThe $k$-bounded preference model is equivalent to SM instances on (bipartite) graphs with maximum degree at most $k$, where each vertex ranks only its adjacent nodes (or equivalently, adjacent edges). Our $(k_1, k_2)$-list model generalizes Bhatnagar et al.'s $k$-list model, as their $k$-list is precisely our $(k_1, k_2)$-list model with $k_1 = k_2 = k$.

\paragraph{Generic instance construction.}  First, we present  a generic construction algorithm that given $H(\calP)$, the Hasse diagram of  poset $\calP$,  and an edge coloring $\phi$ of $H(\calP)$ returns a stable matching instance $I$ whose rotation poset is isomorphic to $\calP$.  The number of agents of $I$ and the running time of the algorithm is $O(p+q)$,  where $p$ and $q$ are the number of vertices and edges of $H(\calP)$ respectively (cf.\ Theorem~\ref{thm:generic-construct}). 

More importantly for us, the algorithm is fairly simple and yet remarkably versatile. For a fixed poset $\calP$, we can produce stable matching instances with different properties that still realize $\calP$ by changing the input edge coloring and tweaking some parts of the algorithm.  For example, when $\phi$ assigns the same color to all the edges of $H(\calP)$, the algorithm produces an instance with $4p$ agents---an improvement to Irving and Leather's construction~\cite{Irving1986-complexity} when $H(\calP)$ is a dense graph because their instance's number of agents depends on both $p$ and $q$. In order to obtain instances $I$ realizable in the $k$-bounded model, we take $\phi$ to be a proper edge coloring of $H(\calP)$, while for the $(k_1, k_2)$-list model, $\phi$ assigns the color $v$ to all edges leaving node $v$. 

\paragraph{Constructing instances in restricted models.} Using the generic construction algorithm, we establish that for a given $\calP$, stable matching instances that realize $\calP$ can be constructed efficiently for  the following preference list models:
\begin{itemize}
\item $k$-bounded for any $k \geq 3$ (cf.\ Theorem~\ref{thm:k-bound}),
\item $k$-attribute for any $k \geq 6$ (cf.\ Theorem~\ref{thm:k-attr}),
\item $(k, \infty)$-list or $(\infty, k)$-list for any $k \geq 2$ (cf.\ Theorem~\ref{thm:k-list}).
\end{itemize}
In other words, the above  models are ``rich'' enough that {\it every} poset can be realized by some instance in that model.

\begin{rem}
    Irving and Leather's construction~\cite{Irving1986-complexity} produces incomplete preference lists that realize an arbitrary poset where the men's preference lists have lengths at most 3. However, the women's preference lists can be arbitrarily long in their construction (for instance if the poset is a chain of length $p$, some woman will have a preference list of length $p$). Thus our result for $k$-bounded preferences---in which all agents have bounded preference lists---is not implied by previous work.
\end{rem}

\begin{rem}
    In~\cite[Theorem~3.1]{Bhatnagar2008-sampling}, Bhatnagar et al.\ proved that for any fixed instance size $n > 2$, there is a poset $\calP$ realized by a (general) SM instance $I$ of size $n$ that cannot be realized by any $k$-attribute instance for any $k < n / 2$. On the other hand, our result for the $k$-attribute model implies that there is some $6$-attribute instance that realizes $\calP$, albeit with an instance size $n' > n$. (Some $n' = O(n^4)$ suffices.)
\end{rem}

Our results for the $k$-bounded and $(k, \infty)$-list (or $(\infty, k)$-list) models are tight.  The rotation posets of instances in the $1$-bounded model can only be the empty poset because all the instances have only one stable matching while those in the $2$-bounded model can only be antichains.  Thus, it is surprising that in the $3$-bounded model {\it all} posets can be a rotation poset of some instance.   

In the $(\infty, k)$-list model, when $k = 1$, all the women have the same preference list.  This model is equivalent to the one studied by Irving et al.~\cite{Irving2008-stable} where they show that all instances from this model have a unique stable matching. Once again, these instances have the empty poset as a rotation poset. But when $k = 2$, the women can have two different preference lists. According to our result, a drastic change occurs and {\it any} poset can be the rotation poset of some instance. 

As for the $k$-attribute model, Chebolu et al.~\cite{Chebolu2012-complexity} proved that instances in the $1$-attribute model have  paths as rotation posets.  Bhatnagar et al.~\cite{Bhatnagar2008-sampling} showed that instances in the $2$-attribute model can realize arbitrary star posets as their rotation posets. We suspect that our result can be improved to show that instances in the $k$-attribute model, with some $k < 6$, can have an arbitrary rotation poset.

Our structural results for the $k$-bounded, $k$-attribute, and $(k, \infty)$- and $(\infty, k)$-list models have the following implications: any ``structural'' stable matching problem---i.e., any problem whose solution depends only on the rotation poset of an instance---is as hard in these restricted models as the general case. We state two explicit consequences:
\begin{itemize}
\item In the $k$-bounded ($k \geq 3$), $k$-attribute ($k \geq 6$), and $(k, \infty)$- and $(\infty, k)$-list ($k \geq 2$) models, $\shSM$ is $\shBIS$-complete.
\item In the models above, it is $\shP$-hard to find generalized median stable matchings.
\end{itemize}

\paragraph{Characterization of $k$-range posets.} For the $k$-range model, we prove that there is no fixed constant $k^*$ such that every poset can be realized by some instance in the $k^*$-range model (cf.\ Corollary~\ref{cor:no-constant}). Instead, we show that the ``range'' of an stable matching instance $I$ and the pathwidth of a poset are connected in the following sense:
\begin{itemize}
  \item When $H(\calP)$ has a path decomposition of width $k$, there is an instance in the $O(k)$-range model that realizes $\calP$.  The instance has $O(kp)$ agents and can be constructed in $O(k^2p^2)$ time where $p$ is the number of elements in $\calP$ (cf.\ Theorem~\ref{thm:k-range-construct}).
 
  \item On the other hand, suppose $I$ is a $k$-range instance.   Then the Hasse diagram $H$ of $\calR(I)$ has a pathwidth of $O(k^2)$ (cf.\ Theorem~\ref{thm:k-range-upper}). Moreover, a path decomposition of $H$ can be computed in time polynomial in the instance size, independent of $k$.
\end{itemize}

Using this characterization of rotation posets realized by $k$-range instances, we show that many problems that are computationally hard for general instances are fixed-parameter tractable (FPT) in the $k$-range model. Specifically, we show that the following problems admit FPT algorithms parameterized by the range of the instance\footnote{Every stable matching instance $I$ is a $k$-range instance for some value of $k$. The \emph{range} of $I$ is the minimum $k$ for which $I$ is a $k$-range instance.}:
\begin{itemize}
\item exactly counting stable matchings,
\item sampling stable matchings exactly uniformly,
\item computing generalized median stable matchings,
\item finding balanced stable matchings,
\item finding sex-equal stable matchings.
\end{itemize}
Since every stable matching instance is $k$-range for some value of $k$, these results show that the range of an instance is a valuable parameter through which to study the parameterized complexity of stable marriage problems.

\subsection{Related Work}

Gale and Shapley~\cite{Gale1962-college} first proved that all stable marriage problem (SMP) instances admit stable matchings, and described an efficient algorithm for finding one. Knuth posed the question of characterizing the structure of the set of stable matchings as Research Problem~6 in~\cite{Knuth1976-marriage} (also in English translation~\cite{Knuth1997-stable}). Irving and Leather's work~\cite{Irving1986-complexity} introduced the notion of the rotation poset and showed how it determines the structure of the set of stable matchings. Using this structural characterization of the set of stable matchings, Irving and Leather showed that $\shSM$ is $\shP$-complete, implying that $\shSM$ cannot be solved in polynomial time unless $\Poly = \NP$~\cite{Valiant1979-complexity}. Specifically, Irving and Leather demonstrated that counting stable matchings is equivalent to $\shDOWN$ (counting downsets in posets), a problem which was shown to be $\shP$-complete by Provan and Ball~\cite{Provan1983-complexity}. $\shDOWN$ was also shown to be $\shBIS$-complete by Dyer et al.~\cite{Dyer2004-relative}, hence Irving and Leather's proof also implies the $\shBIS$-completeness of $\shSM$.\footnote{The class $\shBIS$ is the class of counting problems reducible to counting independent sets in bipartite graphs under polynomial-time approximation-preserving transformations. In~\cite{Dyer2004-relative}, Dyer et al.\ call this class $\#\mathrm{RH\Pi}_1$. However we adopt the convention, common in recent works, to associate the class $\#\mathrm{RH\Pi}_1$ with its canonical complete problem, $\shBIS$. Dyer et al.\ conjecture that no $\shBIS$-complete problem has a fully polynomial randomized approximation scheme (FPRAS), and prove that $\shDOWN$ is $\shBIS$-complete.} 

The structural properties of stable matchings revealed by Irving and Leather's work have been exploited in numerous subsequent works to understand the complexities of various tasks related to the SMP. The classical book of Gusfield and Irving~\cite{Gusfield1989-stable} (and references therein) describes some early applications, while the more recent book of Manlove~\cite{Manlove2013-algorithmics} gives an expansive overview of the SMP and its variants. Recently, Karlin et al.~\cite{Karlin2018-simply} analyzed novel features of the rotation poset in order to show that any SM instance of size $n$ has at most $c^{n}$ stable matchings for some constant $c$, thereby making significant progress towards another of Knuth's longstanding open problems.

Bhatnagar et al.~\cite{Bhatnagar2008-sampling} introduced the $k$-attribute and $k$-range preference models we study, as well as the $k$-list model which our $(k_1, k_2)$-list model generalizes. When $k_1 = 1$ or $k_2 = 1$, our $(k_1, k_2)$-model corresponds to ``master preference lists'' (cf.\ Irving et al.~\cite{Irving2008-stable}). Chebolu et al.~\cite{Chebolu2012-complexity} show that counting stable matchings is $\shBIS$ complete in the $k$-attribute model for any $k \geq 3$, as well as the $k$-Euclidean model of Bogomolnaia and Laslier~\cite{Bogomolnaia2007-euclidean}. Thus our Theorem~\ref{thm:k-attr} gives an alternative (arguably simpler) proof of Chebolu et al.'s $\shBIS$-hardness result, albeit for the weaker condition $k \geq 6$.

%% \wnote{This isn't correct... \cite{McDermid2014-sex}'s result shows that when men's prefs are 2-bounded sex-equal stable matchings can be found efficiently, not for the $(2, \infty)$-list model. Do we just cut all of the following text? On the other hand, they show that finding sex equal stable matchings is hard in the 3-bounded model. We should mention this.}

%% \updated{The case where $k_1 = 2$ was studied by McDermid and Irving~\cite{McDermid2014-sex}, who showed that ``sex-equal'' stable matchings can be found in polynomial time in (what we call) the $(2, \infty)$-list model (while finding sex-equal stable matchings is is $\NP$-hard in general~\cite{Kato1993-complexity}). Our Theorem~\ref{thm:k-list}  implies that the structure of $(2, \infty)$-list SM instances is already as expressive as the general case, which makes McDermid and Irving's result all the more surprising.\footnote{Computing a sex-equal stable matching can be defined in terms of finding a zero weight downset in the associated rotation poset, where weights are assigned according to the actual preference lists. Theorem~\ref{thm:k-list} implies that rotation posets of $(2, \infty)$-list instances are arbitrary, but the restriction on the preferences limits the flexibility of the associated weights of elements in the rotation poset. McDermid and Irving's result implies that the limitation on the weights is sufficient to admit a polynomial-time algorithm.}}

Finding (and verifying) stable matchings with restricted preferences has also been studied in the centralized~\cite{Kunnemann2019-subquadratic} and distributed~\cite{Khanchandani2017-distributed} settings. K\"unnemann et al.~\cite{Kunnemann2019-subquadratic} showed that for some instances of $k$-attribute, $k$-list, $k$-Euclidean preference models, stable matchings can be computed in $o(n^2)$ time when $k = O(1)$. For $k = \omega(\log n)$, however, $k$-attribute and $k$-Euclidean preferences require $\Omega(n^2)$ time, assuming the strong exponential time hypothesis. In the distributed setting, Khanchandani and Wattenhofer~\cite{Khanchandani2017-distributed} study a preference model equivalent to the $k$-range model. They show that in this model, a stable matching can be computed in $\Theta(k \cdot n)$ distributed rounds (with each node sending or receiving a single $O(\log n)$ bit message each round). 

The parameterized complexity of two $\NP$-hard variants of the SMP were studied by Marx and Schlotter~\cite{Marx2010-parametrized, Marx2011-stable}. The first paper~\cite{Marx2010-parametrized} studies the SMP with incomplete preferences and ties, a problem for which finding a maximum size stable matching is $\NP$-hard~\cite{Manlove2002-hard}. The second paper~\cite{Marx2011-stable} analyzes the Hospital/Residents problem with couples, for which it is $\NP$-hard to determine if a stable matching exists. Marx and Schlotter show that the two problems are fixed-parameter tractable (parameterized by the number of ties, the maximum length, or overall length of ties for the former problem, and the number of couples in the latter problem). That is, these problems can be solved in polynomial time whenever the relevant parameter is a fixed constant (independent of $n$). For an overview of parameterized complexity, see~\cite{Cygan2015-parametrized}. 
 
Recently, Gupta et al.~\cite{Gupta2017-treewidth} considered the parameterized complexity of several other hard variants of the SMP, parameterized by the treewidth of the ``primal graph'' (i.e., the graph of acceptable partners), and the treewidth of the Hasse diagram of the rotation poset of an instance. In particular, they give FPT algorithms for computing sex-equal and balanced stable matchings parameterized by treewidth of the Hasse diagram. Combined with our characterization of the rotation posets arising from $k$-range preferences in Theorem~\ref{thm:k-range-upper}, Gupta et al.'s results immediately imply FPT algorithms for finding sex-equal and balanced stable matchings parameterized by the range of an instance (cf.\ Corollary~\ref{cor:sex-equal-balanced}). In~\cite{Gupta2019-balanced}, Gupta et al.\ studied the parameterized complexity of finding a balanced stable matching, parameterized by the maximum balance that can be achieved.

\subsection{Paper Overview}

Here we give a high level overview of the remainder of the paper. In Section~\ref{sec:background}, we give the necessary background and notation to understand the technical portion of the paper. In order to make our paper as self-contained as possible, Section~\ref{sec:background} contains substantial background on the SMP, especially the structural results of Irving and Leather. We formally introduce the four models of restricted preference lists studied later in the paper, and give a brief overview of the results for pathwidth needed in our analysis of the $k$-range model.

In Section~\ref{sec:generic-construction}, we present the generic construction algorithm $\ConstructInstance$. The construction takes a finite poset $\calP$---or more specifically, the Hasse diagram of $\calP$, $H(\calP)$---as well as an edge coloring of $H(\calP)$, and constructs a stable matching instance whose rotation poset is isomorphic to $\calP$. Our generic construction is conceptually different from the construction of Irving and Leather~\cite{Irving1986-complexity}, though our analysis is similar to the one described in Gusfield and Irving's~\cite[Section~3.8]{Gusfield1989-stable} book for Irving and Leather's construction. 

Section~\ref{sec:k-bounded} contains our main results for the $k$-bounded and $k$-attribute models. The argument for the $k$-bounded model is straightforward and makes use of results from Section~\ref{sec:generic-construction}.  The geometric approach we take for the $k$-attribute model is similar in spirit, but more general than Bhatnagar et al.'s construction. Our key technical tool is Gale's construction of ``neighborly polytopes''~\cite{Gale1963-neighborly}. 

In Section~\ref{sec:k-list}, we develop our main result for the $(k,\infty)$-list model by first constructing an instance with incomplete preference lists using $\ConstructInstance$.  We then embed the men's lists into two complete preference lists (while completing the women's preference lists arbitrarily)  and show that the rotation poset of the new instance is the same as the original instance.

Sections~\ref{sec:k-range} and \ref{sec:k-range-algo} present our main results for the $k$-range model. In Section~\ref{sec:k-range},  we begin with a width-$k$ path decomposition of $H(\calP)$, the Hasse diagram of poset $\calP$. We tweak $\ConstructInstance$ so that it incorporates the path decomposition into the construction of a stable matching instance. We then expand each agents' preference list into a complete one, making sure that the instance has range $O(k)$.  In Section~\ref{sec:k-range-algo}, we study instances in the $k$-range model and show that a path decomposition of width $O(k^2)$ can be efficiently computed for their rotation posets.  In Section~\ref{sec:k-range-consequences}, we show how the path decompositions can be used to obtain FPT algorithms for several computationally hard problems: counting and uniformly sampling stable matchings, and finding a median, sex-equal, and balanced stable matchings.

Finally, the paper concludes with Section~\ref{sec:conclusions} which discusses related questions and directions for future work.

\section{Background and Preliminaries}
\label{sec:background}

In the classical version of the stable marriage problem, a \dft{stable matching instance} $I = (M, W, P)$ consists of a set of $n$ men $M$ and a set of $n$ women $W$ together with their  preferences $P$.  For each $m \in M$, $P$ contains a \dft{preference list} $P_m$ that is a total ordering of $W$, and symmetrically each $w \in W$ has a corresponding preference list $P_w$ that is a total ordering of $M$ in $P$.   For each  $a \in M \cup W$, we refer to $P_a(1)$ as $a$'s most preferred partner, $P_a(2)$ as $a$'s second most preferred partner etc.  If $b = P_a(i)$, we say that $a$ assigned $b$ a \dft{rank} of $i$ and write $P_a(b) = i$.  Throughout this paper, we assume that $n = \abs{M} = \abs{W}$ is the \dft{size} of the instance.

A \dft{matching} $\mu$ of $I$ is a set of $n$ man-woman pairs such that every agent is part of exactly one pair.   Suppose $(m, w'), (m', w) \in \mu$ with $m \neq m'$ and $w \neq w'$.  We call $(m, w)$ a \dft{blocking pair} of $\mu$  if $m$ and $w$ mutually prefer each other to their partners in $\mu$; that is,  $P_m(w) < P_m(w')$ and $P_w(m) < P_w(m')$. We say that $\mu$ is a \dft{stable matching} if it has no blocking pairs.  Gale and Shapley's seminal paper proved that {\it every} stable matching instance $I$  has a stable matching~\cite{Gale1962-college}.  Moreover, they presented an algorithm that finds a stable matching of $I$ in $O(n^2)$ time (assuming preferences can be accessed and compared in unit time).   

A common extension of the stable marriage problem allows for the number of men and women to be different and for each agent $a$ to have a preference list that ranks only a subset of \dft{acceptable} partners from the opposite group (i.e., the agents may have \dft{incomplete preference lists}).  The agents $m$ and $w$ form an \dft{ acceptable pair} if the two agents are acceptable to each other.  For such an instance, a \dft{matching} $\mu$ is now a set of acceptable pairs such that each agent  is part of at most one pair.  We call $(m,w)$ \dft{a blocking pair} of $\mu$ if (i) it is an acceptable pair,  (ii) $m$ is either unmatched or prefers $w$ to his partner in $\mu$ and (iii) $w$ is either unmatched or prefers $m$ to his partner in $\mu$.   Once again, a matching is \dft{stable} if it has no blocking pairs. 
Like the classical case,  every instance in this setting has a stable matching and the Gale-Shapley algorithm can easily be modified to find one.  Furthermore, a well-known result by Gale and Sotomayor~\cite{Gale1985-remarks} (cf.~\cite{Roth1986-allocation}) states that \emph{every} stable matching of the instance has the same number of pairs and matches exactly the same set of agents. 
  
To streamline our discussion, we shall use \dft{SM instances} to refer to instances in both the classical version and its extension to incomplete preference lists. When we consider instances in the classical case only, we shall say that the SM instances have \dft{complete preference lists}.   On the other hand, when we consider instances with incomplete preference lists, we shall assume that the preference lists of all agents are \dft{consistent} (that is, $a$ is in $b$'s preference list if and only if $b$ is in $a$'s preference list) and all agents appear in every stable matching of $I$.\footnote{Making preference lists consistent and removing participants that are never part of a stable matching do not affect the set of stable matchings of $I$.} Hence, we can again define the \dft{size} of $I$ as $\abs{M} = \abs{W}$.  

In general, an SM instance  can have many stable matchings. We say that a pair $(m,w)$ is a \dft{stable pair} if it is part of  some stable matching of the instance.   We also say that $m$ and $w$ are each other's \dft{stable partners}.  The Gale-Shapley algorithm can produce only two types of stable matchings.  In its man-oriented version, the output is the \dft{man-optimal stable matching}, which is also the \dft{woman-pessimal stable matching}.  That is,  every man is matched to his best stable partner while simultaneously every woman is matched to her worst stable partner.  In the woman-oriented version, the output is the \dft{woman-optimal stable matching}, which is also the \dft{man-pessimal stable matching} and defined accordingly.

\subsection{Posets and DAGs}

A \dft{partially ordered set} or \dft{poset}  $\calP = (X, \prec)$ consists of a set $X$ and a binary relation $\prec$ on the elements of $X$ that is antisymmetric and transitive.\footnote{Throughout the paper, we will assume that all posets have a finite number of elements.} Two posets $\calP = (X, \prec)$ and $\calP' = (Y, \prec')$ are \dft{isomorphic} if there exists bijection $f \colon X \to Y$ satisfying $x \prec y \iff f(x) \prec' f(y)$.  The \dft{Hasse diagram} of $\calP$ is the directed acyclic graph $H(\calP) = (X,E)$ such that 
\[
E = \set{(x, y) \in X \times X \sucht x \prec y \text{ and no } z \text{ satisfies } x \prec z \prec y}.
\]
A subset $Z \subseteq X$  is a \dft{downset} (also a \dft{closed subset} or an \dft{order ideal})  of $\calP$ if  for every  $z \in Z$ and $y \in X$ such that  $y \prec z$ then $y \in Z$.  That is, if $z \in Z$ then every \dft{predecessor} $y$ of $z$  is also in $Z$. 
 
We extend the above terminology to directed acyclic graphs (DAGs).  Let $G = (V, E)$ be a DAG. We call a subset $Z \subseteq V$ a \dft{downset} of $G$ if for every $z \in Z$ and $y \in V$ such that there is a path from $y$ to $z$ (i.e., $y$ is an \dft{ancestor} of $z$) then $y \in Z$.  We denote the number of downsets of $G$ by $\down(G)$.   
   
   The \dft{transitive closure} of DAG $G=(V,E)$ is another DAG $G' = (V, E')$ such that $(u,v) \in E'$ if and only if there is a path from $u$ to $v$ in $G$.  In the later parts of the paper,  instead of the Hasse diagram of a poset $\calP = (X, \prec)$, we will work  with a DAG ${H} = (X,E)$ so that the transitive closures of ${H}$ and $H(\calP)$  are exactly the same. That is, with slight abuse of notation,  their transitive closure is $\calP$.   We note the following.
  
\begin{obs}
Let $H$ and $H'$ be two DAGs  whose transitive closures are exactly the same.  Then $Z$ is a downset of $H$ if and only if $Z$ is a downset of $H'$. 
\end{obs}

In the problem $\shDOWN$, we are given a poset $\calP$ and the goal is to count the number of distinct downsets of $\calP$.  The above observation implies that if $H$ is a DAG whose transitive closure is $\calP$,  then $\down(H)$ is exactly the number of downsets of $\calP$.

\subsection{The Rotation Poset}

To explore all the stable matchings of an instance $I$, we need to consider its \emph{rotation poset}, described below.

\begin{dfn}
  \label{dfn:rotation}
  Let  $\mu$ be a stable matching of SM instance $I$. Let $\rho$ be a circular list of man-woman pairs from $\mu$:
  \[
  \rho = (m_1, w_1), (m_2, w_2), \ldots, (m_\ell, w_\ell). 
  \]
  We say that $\rho$ is a \dft{rotation exposed in $\mu$} if for all $i = 1, 2, \ldots, \ell$, $w_{i+1}$ is the first woman on $m_i$'s preference list after $w_i$ who prefers $m_i$ to her partner $m_{i+1}$ in $\mu$ (where by convention $m_{\ell+1} = m_1$). That is, for all $i$ we have:
  \begin{itemize}
  \item $w_{i+1}$ prefers $m_i$ to $m_{i+1}$, and
  \item there is no $w$ with $P_{m_i}(w_i) < P_{m_i}(w) < P_{m_{i}}(w_{i+1})$ such that $w$ prefers $m_i$ to her partner $m$ in $\mu$.
  \end{itemize}
\end{dfn}

If $\rho = (m_1, w_1), \ldots, (m_\ell, w_\ell)$ is a rotation exposed in $\mu$, then we can form the \dft{elimination of $\rho$}, denoted $\mu \setminus \rho$, as follows:
\[
\mu \setminus \rho = \mu \setminus \set{(m_i, w_i) \sucht i = 1, 2, \ldots, \ell} \cup \set{(m_i, w_{i+1}) \sucht i = 1, 2, \ldots, \ell}.
\]
The matching $\mu \setminus \rho$ is another stable matching. In fact, Irving and Leather showed that every stable matching $\mu$ can be obtained by starting from the man-optimal stable matching $\mu_0$ and eliminating a sequence of rotations. More formally, they proved the following.

\begin{lem}[Irving \& Leather~\cite{Irving1986-complexity}]
  \label{lem:rotation-elimination}
  Let $\mu$ be a stable matching of SM instance $I$. Then there exists a sequence of rotations $\rho_0, \rho_1, \ldots, \rho_{k-1}$ and stable matchings $\mu_0, \mu_1, \ldots, \mu_k = \mu$ such that $\mu_0$ is the man-optimal stable matching and,  for each $i$, $\rho_i$ is a rotation exposed in $\mu_i$, and $\mu_{i+1} = \mu_i \setminus \rho_i$. Moreover, the set $\set{\rho_0, \rho_1, \ldots, \rho_{k-1}}$ uniquely specifies $\mu$. 
\end{lem}

We denote as $R(I)$ the set of all rotations exposed in any stable matching of $I$.  Lemma~\ref{lem:rotation-elimination} shows that every stable matching $\mu$  of $I$ can be associated with a unique subset $R_\mu \subseteq R(I)$.  For example,  the man-optimal stable matching, $\mu_0$, corresponds to the empty set  while the woman-optimal stable matching, $\mu_z$, corresponds to $R(I)$.   Not all subsets of $R(I)$, however, correspond to some stable matching of $I$. To characterize \emph{which} ones do, Irving and Leather defined a poset structure on $R(I)$ as follows.

\begin{dfn}
  \label{dfn:rotation-poset}
  Suppose $\rho, \rho' \in R(I)$. We say $\rho$ \dft{precedes} $\rho'$  and write $\rho \prec \rho'$ if for every stable matching $\mu$ in which $\rho'$ is exposed, we have $\rho \in R_\mu$. That is, $\rho \prec \rho'$ if $\rho$ was eliminated in every stable matching in which $\rho'$ is exposed.
\end{dfn}

\begin{lthm}[Irving \& Leather~\cite{Irving1986-complexity}]
  \label{thm:stable-closed}
  Let $I$ be an SM instance. Then $\calR(I) = (R(I), \prec)$ is a poset.  Moreover, there is a one-to-one correspondence between the stable matchings of $I$ and the downsets of $\calR(I)$.
\end{lthm}

We refer to $\calR(I)$ as the \dft{rotation poset of $I$}. In~\cite{Irving1987-efficient}, Irving et al.\ showed that a DAG whose transitive closure is the rotation poset $\calR(I)$ can be computed in time $O(n^3)$. This runtime was improved to $O(n^2)$ by Gusfield in~\cite{Gusfield1987-three}.

\begin{lthm}[Gusfield~\cite{Gusfield1987-three}]
  \label{thm:compute-dag}
  Given any SM instance $I$ of size $n$,  a DAG $G(I)$ whose transitive closure is $\calR(I)$ can be computed in $O(n^2)$ time. 
\end{lthm}

We shall call $G(I)$ the \dft{rotation digraph} of $I$.  Since the transitive closure of $G(I)$ is $\calR(I)$, we note that the Hasse diagram of $\calR(I)$ is in fact a subgraph of $G(I)$. 

\begin{eg}
  \label{eg:rotation-poset}
  Consider the SM instance shown below.\\
  \begin{centering}
    \begin{minipage}{0.45\textwidth}
      \[
      \begin{array}{rllll}
        m_1: & w_1 & w_2 & w_3 & w_4 \\
        m_2: & w_2 & w_4 &  w_1 & w_3 \\
        m_3: & w_3 & w_4 & w_2 & w_1 \\
        m_4: & w_4 & w_2 & w_3 & w_1 \\
      \end{array}
      \]
    \end{minipage}
    \hfill
    \begin{minipage}{0.45\textwidth}
      \[
      \begin{array}{rllll}
        w_1: & m_2 & m_1 & m_3 & m_4 \\
        w_2: & m_3 & m_1 & m_4 & m_2 \\
        w_3: & m_4 & m_1 & m_2 & m_3 \\
        w_4: & m_1 & m_3 & m_4 & m_2 \\     
      \end{array}  
      \]
    \end{minipage}
  \end{centering}\\
  It has four stable matchings:
  \[
  \begin{array}{rl} 
    \mu_0 = & \set{ (m_1, w_1), (m_2, w_2), (m_3, w_3), (m_4, w_4) } \\
    \mu_1 = & \set{ (m_1, w_2), (m_2, w_1), (m_3, w_3), (m_4, w_4) } \\
    \mu_2 = & \set{ (m_1, w_2), (m_2, w_1), (m_3, w_4), (m_4, w_3) } \\
    \mu_3 = & \set{ (m_1, w_4), (m_2, w_1), (m_3, w_2), (m_4, w_3) } \\
  \end{array}
  \]
  where $\mu_0$ is the man-optimal stable matching and $\mu_3$ is the woman-optimal stable matching. It has three rotations:
  \[
  \begin{array}{rl}
    \rho_1 = & (m_1, w_1), (m_2, w_2) \\
    \rho_2 = & (m_3, w_3), (m_4, w_4) \\
    \rho_3 = & (m_1, w_2), (m_3, w_4). \\
  \end{array}
  \]
  The rotation poset $\calR(I)$ and digraph $G(I)$ constructed by Gusfield's algorithm are depicted in Figure~\ref{fig:rotation-poset}.

  \begin{figure}[h!]
    \begin{center}
      \includegraphics{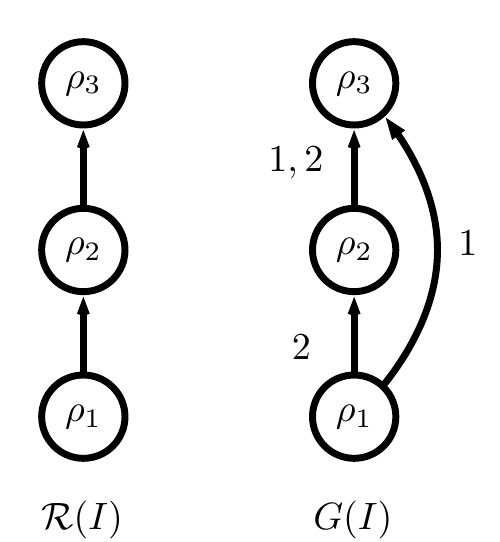}
    \end{center}
    \caption{The rotation poset (left), $\calR(I)$ and the DAG $G(I)$ (right) computed by Gusfield's algorithm. The edge labels in $G(I)$ indicate the ``rule'' associated with the edge---see Remark~\ref{rem:dag-structure}}
    \label{fig:rotation-poset}
  \end{figure}

  %% \begin{center}
  %%   \begin{tikzpicture}
  %%     \begin{scope}[every node/.style={circle,thick,draw}]
  %%       \node(1) at (0,0) {$\rho_1$};
  %%       \node(2) at (0,1.5) {$\rho_2$};
  %%       \node(3) at (0,3) {$\rho_3$};
  %%     \end{scope}
  %%     \path[->, draw, thick](1)--(2);
  %%     \path[->, draw, thick](2)--(3);
  %%     \node[draw=none] at (0,-1) {$\calR(I)$};
  %%   \end{tikzpicture}
  %%   \hspace*{1in}
  %%   \begin{tikzpicture}
  %%     \begin{scope}[every node/.style={circle,thick,draw}]
  %%       \node(1) at (0,0) {$\rho_1$};
  %%       \node(2) at (0,1.5) {$\rho_2$};
  %%       \node(3) at (0,3) {$\rho_3$};
  %%     \end{scope}
  %%     \draw[->, thick] (1) to node[left] {2} (2);
  %%     \draw[->, thick] (2) to node[left] {1,2} (3);
  %%     %\path[->, draw, thick](1)--(2);
  %%     %\path[->, draw, thick](2)--(3);
  %%     \draw[bend right, ->, thick] (1) to node[right] {1} (3); 
  %%     \node[draw=none] at (0,-1) {$G(I)$};
  %%   \end{tikzpicture}
  %%   %\bigskip
  %% \end{center}

  The correspondence between the stable matchings of $I$ and the downsets of $\calR(I)$ are as follows:  $\mu_0$ with $\varnothing$, $\mu_1$ with $\{\rho_1\}$, $\mu_2$ with $\{\rho_1, \rho_2\}$, and $\mu_3$ with $\{\rho_1, \rho_2, \rho_3\}$.  
\end{eg}

Let $ \rho = (m_1, w_1), (m_2, w_2), \ldots, (m_\ell, w_\ell)$ be a rotation of SM instance $I$.   When $\rho$ is eliminated from a stable matching, each $m_i$ is matched to $w_{i+1}$ or, equivalently, each $w_i$ is matched to $m_{i-1}$.  We say that \dft{$\rho$ moves $m_i$ down to $w_{i+1}$} because $m_i$ prefers $w_i$ to $w_{i+1}$.  If $w$ is strictly between $w_i$ and $w_{i+1}$ in $m_i$'s preference list, we also say that \dft{$\rho$ moves $m_i$ below $w$.} Similarly, \dft{$\rho$ moves $w_i$ up to $m_{i-1}$} because $w_i$ prefers $m_{i-1}$ to $m_i$. If $m$ is strictly between $m_{i-1}$ and $m_i$ in $w_i$'s preference list, \dft{$\rho$ moves $w_i$ above $m$.}

\begin{rem}
  \label{rem:dag-structure}
  The rotation digraph $G(I)$ is defined by adding edges according to the following rules~\cite[Section~3.2]{Gusfield1989-stable}. (We have labeled the edges of $G(I)$ in Figure~\ref{fig:rotation-poset} with the rules that were used to create the edges.)
  \begin{description}
  \item[Rule 1.]  Suppose $(m,w)$ is in rotation $\rho$.  If $\rho'$ is the (unique) rotation that moves $m$ to $w$, then there is a directed edge from $\rho'$ to $\rho$.   
    
    In the previous example,  there is an edge from $\rho_1$ to $\rho_3$ because of $(m_1, w_2)$ and an edge from $\rho_2$ to $\rho_3$ because of $(m_3, w_4)$.  
    
  \item[Rule 2.]  Suppose $(m,w)$ is {\it not} in rotation $\rho$. If $\rho'$ is the (unique) rotation that moves $w$ above $m$,  $\rho$ is the (unique) rotation that moves $m$ below $w$ and $\rho' \neq \rho$, then there is a directed edge from $\rho'$ to $\rho$. Note that for these steps to happen, $w$ is in $\rho'$ while $m$ is in $\rho$.

    In the previous example, there is an edge from $\rho_1$ to $\rho_2$ because of $(m_4, w_2)$.  Rotation $\rho_2$ was not exposed in $\mu_0$ because $m_4$ prefers $w_2$ to $w_4$ and $w_2$ prefers $m_4$ to $m_2$.   Rotation $\rho_1$ moves $w_2$ above $m_4$.  Only when $\rho_1$ is eliminated is $\rho_2$ exposed and $\rho_2$ moves $m_4$ below $w_2$ to $w_3$.    There is also an edge from $\rho_2$ to $\rho_3$ because of $(m_1, w_3)$.  Rotation $\rho_2$ moves $w_3$ above $m_1$ while rotation $\rho_3$ moves $m_1$ below $w_3$.  
  \end{description}
\end{rem}

The two rules above can be thought of informally as the reasons for the precedence relation among the rotations of $I$. In Section~\ref{sec:generic-construction}, we shall use Rule~2 extensively to create SM instances (whereas Irving and Leather's construction forces precedence between rotations by appealing to Rule~1).

%\textcolor{red}{Should we present an example where the two rules are used?}

\subsection{Realizing Posets as Rotation Posets}

Given a poset $\calP$, we say that an SM instance $I$ \dft{realizes} $\calP$ if the rotation poset of $I$, $\calR(I)$, is isomorphic to $\calP$.  Irving and Leather showed the following fundamental result.
 % regarding the structure of rotation posets.

\begin{thm}[Irving \& Leather~\cite{Irving1986-complexity}]
  \label{thm:rotation-posets-general}
  Let $\calP$ be a poset whose Hasse diagram has $p$ vertices and $q$ edges.  There is an SM instance $I$ that realizes $\calP$.  Moreover, the number of agents of $I$ is $O(p+q)$ and the instance $I$ can be constructed in $O(p+q)$ time. 
\end{thm}

Together, Theorems~\ref{thm:stable-closed} and~\ref{thm:rotation-posets-general} show that counting stable matchings and the downsets of a (finite) poset are essentially equivalent: any instance of one problem can be efficiently reduced to an equivalent instance of the other. %We formalize this connection and state some consequences in the following section.

In order to show that an SM instance has a particular rotation poset, it will be useful to have a criterion under which two SM instances $I$ and $I'$ have isomorphic rotation posets. We describe a sufficient condition below that is partly based on Irving and Leather's notion of ``shortlists"~\cite{Irving1986-complexity}.

\begin{dfn}
 \label{dfn:shortlists}
 Let  $I = (M, W, P)$ be an SM instance. Let $a \in M \cup W$.  The \dft{symmetric shortlist} of $a$, $S_a$, is  
 the sublist of $P_a$ consisting of agents $b$ such that 
 \begin{itemize}
\item $b$ is the man-optimal or woman-optimal stable partner of $a$ or $b$ lies between these two partners in $P_a$, and 
\item $a$ is  the man-optimal or woman-optimal stable partner of $b$ or $a$ lies between these two partners in $P_b$. 
\end{itemize}
   We shall use $S(I)$ to denote the set containing  all agents' symmetric shortlists.  
\end{dfn}

Here's a simple fact about symmetric shortlists.

\begin{prop}
\label{prop:shortlists}
Let $I=(M,W,P)$ be an SM instance.  For any agent $a$, $a$ is part of $S(b)$ if and only if $b$ is part of  $S(a)$. 
\end{prop}

\begin{eg}[Example~\ref{eg:rotation-poset} continued]
  The shortlists of the instance are shown below. \\
  \begin{centering}
    \begin{minipage}{0.45\textwidth}
      \[
      \begin{array}{rllll}
        m_1: & w_1 & w_2 & w_3 & w_4 \\
        m_2: & w_2  &  w_1  \\
        m_3: & w_3 & w_4 & w_2  \\
        m_4: & w_4 & w_2 & w_3  \\
      \end{array}
      \]
    \end{minipage}
    \hfill
    \begin{minipage}{0.45\textwidth}
      \[
      \begin{array}{rllll}
        w_1: & m_2 & m_1  \\
        w_2: & m_3 & m_1 & m_4 & m_2 \\
        w_3: & m_4 & m_1 & m_3 \\
        w_4: & m_1 & m_3 & m_4  \\     
      \end{array}  
      \]
    \end{minipage}
  \end{centering}
\end{eg}

We will rely heavily upon the following result, a restatement of~\cite[Theorem~1.2.5]{Gusfield1989-stable}.

\begin{lem}[Cf.~{\cite[Theorem~1.2.5]{Gusfield1989-stable}}]
  \label{lem:shortlists}
  Let $I = (M, W, P)$.  Let $I'$ be obtained from $I$ be replacing each agent's preference list $P_a$ with their symmetric shortlist $S_a$; i.e.,  $I' = (M,W, S(I))$.   Then $I$ and $I'$ have identical rotation posets and  stable matchings. 
  % and $I' = (M, W, P')$ be two SM instances with the same sets of men and women. Suppose $S(I) = S(I')$. Then $I$ and $I'$ have isomorphic rotation posets.
\end{lem}

%% \begin{proof} According to Lemma~\ref{lem:rotation-elimination}, the starting point for computing all rotations of $I$ is the man-optimal stable matching $\mu_0$.  Once all the rotations have been eliminated, the woman-optimal stable matching is obtained.  Thus, for any woman $w \in W$, men that she ranked before her woman-optimal stable partner or after her man-optimal stable partner will not affect the rotation poset of $I$.  In particular, for such a man $m$, the pair $(m,w)$ will never be part of Rules 1 or 2 in Remark~\ref{rem:dag-structure}.   So these men can be removed from $w$'s preference list,   $w$ can in turn be removed from their preference lists and the resulting instance has exactly the same rotation poset and the same set of stable matchings as $I$.  By a similar argument, for each man $m \in M$, all the women he ranked before his man-optimal stable partner or after his woman-optimal stable partner can be removed,  $m$ in turn can be removed from their preference lists and the instance's rotation poset and set of stable matchings will not be affected.    But these actions imply that we have reduced each agent's preference list $P_a$ to their symmetric shortlist $S_a$.  The lemma follows. 
%% \end{proof}

\begin{cor}
\label{cor:shortlists}
Let $I = (M,W,P)$ and $I' = (M, W, P')$ such that $S(I) = S(I')$.  Then $I$ and $I'$ have the identical rotation posets and stable matchings. \end{cor}

\begin{cor}
\label{cor:shortlists-2}
Let $I = (M,W,P)$ be an SM instance with incomplete preference lists.\footnote{Recall that we assume without loss of generality that every agent is matched in every stable matching.}  Let $I' = (M, W, P')$ be obtained from $I$ by adding to the end of each incomplete preference list its missing agents.  Then $I$ and $I'$ have identical rotation posets and stable matchings. 
\end{cor}

\begin{proof}
First, we note that running the man-oriented Gale-Shapley algorithm on $I$ and $I'$ returns the same man-optimal stable matching $\mu_0$.  Similarly, running the woman-oriented Gale-Shapley algorithm on $I$ and $I'$ returns the same woman-optimal stable matching $\mu_z$.   This means that the ``extra" agents added to complete the preference lists will not be a part of the symmetric shortlist of $S(I')$ and $S(I) = S(I')$.  By Corollary~\ref{cor:shortlists}, the result follows. 
\end{proof}

Both Corollary~\ref{cor:shortlists} and \ref{cor:shortlists-2} will be useful in Sections \ref{sec:generic-construction}--\ref{sec:k-range} where our goal is to construct an SM instance that realizes a poset $\calP$ but whose preference lists obey certain properties.  Our strategy is to first create a ``smaller'' SM instance $I$ that realizes $\calP$.  We then expand each agent's preference list into a complete list while making sure their symmetric shortlists stay the same. Thus, the new SM instance still realizes $\calP$.

\subsection{Restricted Preference Models}
\label{sec:restricted-preference-models}

%\textcolor{red}{Will,  can you fix the definition for the $(k_1, k_2)$-List Preferences?  You had notation for the case when the preference lists are incomplete and when they are complete.   Also, we should probably motivate these preference models like Bhatnagar et al.  BTW, should we allow for incomplete preference lists for all of them?  The k-range one is tricky because allowing for incomplete preference lists makes it trivial.. but it seems weird that we allow it for $(k_1, k_2)$-list preferences but not for others??}

Let $k$, $k_1$ and $k_2$ be positive integers.  We now describe the four models of restricted preference lists that we will study in the paper.  SM instances in the $k$-bounded model have incomplete preference lists (unless $k = n$) while those in the  $k$-attribute, $(k_1, k_2)$-list  and $k$-range  have complete preference lists. 

%We note that $k$-attribute, $(k_1, k_2)$-list  and $k$-range SM instances   have complete preference lists, while $k$-bounded SM instances have incomplete preference lists (unless $k = n$).

\paragraph{\texorpdfstring{$k$}{k}-Bounded Preferences}

An SM instance $I$ has \dft{$k$-bounded preferences} if each agent's preference list has length at most $k$. 
%contains at most $k$ members of the opposite gender. That is, $k$-bounded instances are those instances whose underlying graph of acceptable partners has degrees bounded by $k$. 
We denote the family of $k$-bounded SM instances with $n$ men and $n$ women by $\bound(k, n)$, and $\bound(k) = \bigcup_{n \in \Z^+} \bound(k, n)$.

\paragraph{\texorpdfstring{$k$}{k}-Attribute Preferences}

%% \textcolor{red}{Will, do you think it's better to use $\vec{x}$ to refer to the vector associated with $x$?  It'll help remind the reader that $x$ is not just a variable but a vector.  Chebolu et al. did this in their paper.   Also, I was hoping to use $a$ rather than $x$ and $b$ rather than $y$.} \\

An SM instance $I$ has \dft{$k$-attribute preferences} if  each agent $a \in M \cup W$ has an associated vector $\veca \in \R^k$ and linear function $\varphi_a : \R^k \to \R$. We refer to the pair $(\veca, \varphi_a)$ as $a$'s \dft{profile}. The preference list of $a$ is constructed from the profiles of the agents from the other group as follows:  $a$ has $P_a(1) = b_1, P_a(2) = b_2, \ldots, P_a(n) = b_n$  if and only if\footnote{We assume that $\varphi_a(\vecb) \neq \varphi_a(\vecb')$ for all $\vecb' \neq \vecb$, so that preferences are \emph{strict} (without ties).}
\[
\varphi_a(\vecb_1) > \varphi_a(\vecb_2) > \cdots > \varphi_a(\vecb_n).
\]
We denote the family of $k$-attribute SM instances with $n$ men and $n$ women by $\attr(k, n)$, and $\attr(k) = \bigcup_{n \in \Z^+} \attr(k, n)$. Figure~\ref{fig:k-attribute-eg} illustrates an example of $2$-attribute preferences.

\begin{figure}
  \begin{center}
    \includegraphics[scale=1]{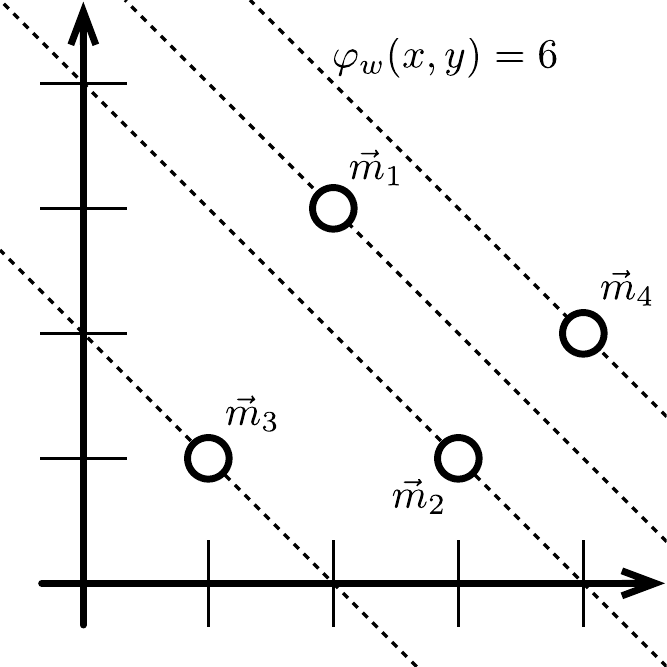}
  \end{center}
  \caption{An illustration of $2$-attribute preferences. In this example, the four men $m_1, m_2, m_3$, and $m_4$ are associated with vectors $\vecm_1 = (2, 3)$, $\vecm_2 = (3, 1)$, $\vecm_3 = (1, 1)$, and $\vecm_4 = (4, 2)$, respectively. If a woman $w$ is associated with the linear function $\varphi_w(x, y) = x + y$, then her corresponding preference list is $m_4, m_1, m_2, m_3$, because $\varphi_w(\vecm_i) = 6, 5, 4, 2$ for $i = 4, 1, 2, 3$, respectively. The dashed lines in the figure show the level sets for $\varphi_w(x, y) = 6, 5, 4, 2$.}
  \label{fig:k-attribute-eg}
\end{figure}

%% \wnote{TODO Make an illustrative example of $k$-attribute preferences, say, in $\R^2$. I will do this on my next pass.}

\paragraph{\texorpdfstring{$(k_1, k_2)$}{(k1, k2)}-List Preferences}

An SM instance $I$ has \dft{$(k_1, k_2)$-list preferences} if the men and women can be partitioned into at most $k_1$ and $k_2$ groups, respectively, such that all agents  within each group have the same preference lists.  If the women or men can be partitioned into an arbitrary number of groups, we say that the SM instance has \dft{$(k_1, \infty)$-list preferences} and \dft{$(\infty, k_2)$-list preferences} respectively.    We denote the set of SM instances with $(k_1, k_2)$-list preferences and $n$ men and $n$ women by $\klist(k_1, k_2, n)$, and $\klist(k_1, k_2) = \bigcup_{n \in \Z^+}  \klist(k_1, k_2, n)$.

We note that the $(k, k)$-list model is equivalent to the $k$-list model of Bhatnagar et al.~\cite{Bhatnagar2008-sampling} while the $(1, \infty)$-list and $(\infty, 1)$-list models correspond to  stable matchings with ``master lists'' that Irving et al investigated in~\cite{Irving2008-stable}. 

\paragraph{\texorpdfstring{$k$}{k}-Range Preferences}

Consider an instance $I$ with complete preferences. For each pair of agents $a, b$ of the opposite sex, let $P_a(b)$ denote $a$'s rank of $b$---i.e., the position in which $b$ appears in $a$'s preference list.  For a woman $w \in W$, her \dft{minrank} and \dft{maxrank} are, respectively, the minimum and maximum rank that she appears in any man's preference list:
\[
\min\rank(w) = \min_{m \in M} P_m(w) \quad\text{and}\quad \max\rank(w) = \max_{m \in M} P_m(w).
\]
The $\min\rank$ and $\max\rank$ of a man $m$ are defined analogously. The \dft{range} of $I$ is
\[
\range(I) = \max_{a \in M \cup W} (\max\rank(a) - \min\rank(a)) + 1.
\]
We say that $I$ has \dft{$k$-range preferences} if $\range(I) \leq k$.

Intuitively, $\range(I)$ gives a measure of how similar the agents' preferences are. When $k = 1$, all men and women (respectively) have the same preferences, while all instances $I$ of size $n$ have $\range(I) \leq n$. We denote the family of $k$-range SM instances with $n$ men and  $n$ women by $\range(k, n)$ and $\range(k) =  \bigcup_{n \in \Z^+} \range(k, n)$.

Note that every SM instance $I$ is in $\range(k)$ for $k = \range(I)$. Moreover, $\range(I)$ can be computed in linear time by computing the minrank and maxrank of each agent individually.

%% An SM instance $I$ has \dft{$k$-range preferences} if there is an \dft{objective rank} $\orank : M \cup W \to [n]$ such that for each $m \in M$ and $w \in W$, we have
%% \[
%% \orank(w) \leq P_m(w) \leq \orank(w) + k - 1 \quad\text{and}\quad \orank(m) \leq P_w(m) \leq \orank(m) + k - 1.
%% \]
%% That is, all (subjective) rankings by the agents are no better than and at most $k - 1$ positions worse than the objective rankings. 

%% \begin{rem}
%%   While the subjective ranking functions $P_a$, $a \in M \cup W$,   must be strict rankings (i.e., bijections between $W$ and $[n]$ or between $M$ and $[n]$ respectively), we do not require the objective rank $\orank : M \cup W \to [n]$ to be bijections when restricted to $M$ or $W$.  For example, it could be that $\orank(w_1) = \orank(w_2) = 1$. Thus, in  the $2$-range model, a man $m$ may have $P_m(w_1) = 1$ and $P_m(w_2) = 2$, or $P_m(w_1) = 2$ and $P_m(w_2) = 1$.
%% \end{rem}

%% We note that given an arbritrary (complete) SM instance $I$,  $I \in \range(k)$ for some $k$.  In Section~\ref{sec:k-range}, we describe how to find the smallest such $k$ in linear time.

\subsection{Pathwidth}
\label{sec:pathwidth}

Here we briefly review some fundamental results regarding pathwidth, and adapt them to our discussion of $k$-range preferences in Sections~\ref{sec:k-range} and \ref{sec:k-range-algo},
\begin{dfn}
  \label{dfn:path-decomp}
  A \dft{path decomposition} of graph $G = (V, E)$ is a sequence $(X_1, X_2, \ldots, X_r)$ of subsets of $V$  such that:
  \begin{enumerate}
  \item $\bigcup_{i = 1}^r X_i = V$,
  \item for each edge $\set{u, v} \in E$, there exists $i \in [r]$ such that $u, v \in X_i$,
  \item for all $i, j, k \in [r]$ with $i \leq j \leq k$, we have $X_i \cap X_k \subseteq X_j$.
  \end{enumerate}
  The \dft{width} of the path decomposition is $\width(\calX) = \max_{i} \abs{X_i} - 1$. The \dft{pathwidth} of $G$, denoted $\pw(G)$, is the minimum width over all path decompositions of $G$.
\end{dfn}

We extend the definition of pathwidths to directed graphs and posets. 

\begin{dfn}
  \label{dfn:pathwidth-p}
 Let $H$ be a directed graph.  The \dft{pathdwidth of $H$} is simply the pathwidth of the undirected version of $H$.   
  Let $\calP$ be a poset, and $H(\calP)$ its Hasse diagram.  The \dft{pathwidth of $\calP$}  is $\pw(\calP) = \pw(H(\calP))$. %where $\pw(H(\calP))$ is the \emph{undirected} pathwidth of $H(\calP)$. 
\end{dfn}

%\begin{rem}
 % \label{rem:directed-pathwidth}
 % Definition~\ref{dfn:path-decomp} can easily be generalized to directed graphs by simply ignoring the orientation of each edge. Thus, in the sequel, we will refer to pathwidth, etc., for directed graphs with this interpretation.
%\end{rem}

\begin{rem}
  \label{rem:pathwidth-interval}
  Suppose $(X_1, X_2, \ldots, X_r)$ is a path decomposition of $G$. Item~3 above implies that for each vertex $v$, there is an interval $I_v \subseteq [r]$ such that $v \in X_i$ if and only if $i \in I_v$. By item~2, if $\set{u, v} \in E$, then we must have $I_u \cap I_v \neq \varnothing$. Thus, $G$ is a subgraph of the interval graph\footnote{Recall that an \dft{interval graph} on a family $\calI$ of intervals is the graph $G = (\calI, E)$ where $\set{I, J} \in E$ if and only if $I \cap J \neq \varnothing$.} defined by the intervals $\set{I_v \sucht v \in V}$. For $I_v = [i_v, j_v]$, we say that $v$ is \dft{added} to the decomposition at index $i_v$, and \dft{removed} at index $j_v + 1$. 
  %\textcolor{red}{Note: previously, it was $j_v$ but I think it should be $j_v + 1$ to ensure that when $\chi$ is a nice decomposition, some node is added or removed at every index. }
\end{rem}

\begin{dfn}
  \label{dfn:nice-path} 
  Let $\calX = (X_1, X_2, \ldots, X_r)$ be a path decomposition of graph $G$. We say that $\calX$ is a \dft{nice path decomposition} if 
  $\abs{X_1} =  1$, $\abs{X_r} = 0$ and 
  for all $i \in [r - 1]$, we have $\abs{X_i\,\triangle\,X_{i+1}} = 1$. That is, when $\calX$ is nice, exactly one vertex is added or removed at each index. 
\end{dfn}

%% The (omitted) proof of the following lemma is straightforward.

\begin{lem}
  \label{lem:nice-path} Let $G$ be a graph with $n$ vertices. 
  Suppose $\calX = (X_1, X_2 \ldots, X_r)$ is a path decomposition of $G$ of width $k$. Then $G$ has a nice path decomposition $\calY = (Y_1, Y_2, \ldots, Y_{s})$ of width $k$ with $s = 2n$. Moreover, $\calY$ can be computed from $\calX$ in time $O(k n)$. 
\end{lem}
\begin{proof}[Proof sketch.]
  Let $\calX = (X_1, X_2, \ldots, X_r)$ be a path decomposition, and define $X_0 = \varnothing$. For each index $i = 0, 1, \ldots, r - 1$, let $\Delta_i = X_i \triangle X_{i+1}$ denote the symmetric difference of $X_i$ and $X_{i+1}$. If $\abs{\Delta_i} = 1$, then exactly one element was added or removed between $X_i$ and $X_{i+1}$, as desired. Otherwise, if $\abs{\Delta_i} > 1$, insert a sequence of $\abs{\Delta_i} - 1$ sets between $X_i$ and $X_{i+1}$ where each set is formed by removing a single element from $X_i \setminus X_{i+1}$ or adding a single element from $X_{i+1} \setminus X_i$. The resulting path decomposition $\calY$ has length $s = 2n$ because every vertex except the unique element in $Y_1$ is added exactly once, and every element is removed exactly once.
\end{proof}

The following seminal result of Bodlaender shows that computing the pathwidth and optimal path decompositions of a graph  is fixed parameter tractable.

\begin{lthm}[Bodlaender~\cite{Bodlaender1996-linear}]
  \label{thm:pathwidth}
  Let $G$ be a graph and let $k \in \N$ be a constant.  There is an algorithm that decides whether $\pw(G) \leq k$ in $O(f(k)\abs{G})$ time. If $\pw(G) \leq k$, then the algorithm outputs a path decomposition $\calX$ of $G$ of width $k$.
\end{lthm}

%\textcolor{red}{Will, should we go ahead and specify the running time of  the algorithm in the theorem as $O(f(k)\abs{G})$ too, like the corollary?}

The next corollary is immediate from Theorem~\ref{thm:pathwidth} and Lemma~\ref{lem:nice-path}.

\begin{cor}
  \label{cor:nice-path}
  For any graph $G$, a nice path decomposition of $G$ can be computed in time $O(f(k) \abs{G})$ where $k = \pw(G)$ and $f$ is some function depending only on $k$.
\end{cor}

\section{A Generic Construction}
\label{sec:generic-construction}

Let $\calP$ be a finite poset and let $H(\calP) = (V, E)$ be its Hasse diagram.   An \emph{edge coloring} $\phi$ of $H(\calP)$ assigns each edge $e \in E$ a color $\phi(e) \in \Z^+$.   In this section, we describe a simple algorithm, $\ConstructInstance$, that given the Hasse diagram $H(\calP)$ and an edge coloring $\phi$ returns an  SM instance $I = I(\calP, \phi)$ whose rotation poset is isomorphic to $\calP$. The number of agents in $I$  and the running time of the algorithm are linear in the size of $H(\calP)$.  In the later sections, we will demonstrate the remarkable versatility of $\ConstructInstance$.  By choosing appropriate edge colorings and tweaking  some of its parts, we show that the algorithm is capable of producing SM instances whose preference lists have all kinds of properties.

Let $V = [p] = \{1, 2, \hdots, p\}$.  Like  Irving and Leather~\cite{Irving1986-complexity},  our goal is to create rotations $\rho_1, \rho_2, \ldots, \rho_p$ so that the mapping $f(v) = \rho_v$ for $v \in [p]$ is an isomorphism from $\calP$ to $\calR(I)$. Here, we provide a high-level description of  $\ConstructInstance$ while pseudo-code is provided in Algorithm~\ref{alg:generic}. Figure~\ref{fig:generic-construction} shows an illustration of the algorithm's output for a poset $\calP$ that will serve as a running example for the remainder of the paper. We find it instructive to view an SM instance as being defined on the \dft{primal graph} $G = (M \cup W, F)$, where $\set{m, w} \in F$ if and only if $m$ and $w$ form an acceptable pair. The preferences are then defined by having each agent rank its incident \emph{edges} in the primal graph.

Let $H(\calP) = (V, E)$ be the Hasse diagram of a poset $\calP$, and let $\phi \colon E \to \Z^+$ be an edge coloring. The structure of $G$ corresponding to the SM instance constructed by $\ConstructInstance(H(\calP), \phi)$ is as follows. For each vertex $v \in V$, there is a corresponding cycle $\rho_v$ in $G$ of length at least $4$ (i.e., with at least two men and the same number of women). For each (directed) edge $(u, v) \in E$, there is a corresponding (undirected) edge between a woman in $\rho_u$ and a man in $\rho_v$. The coloring $\phi$ determines the size and the identities of the agents in each cycle, as well as the identities of the endpoints of edges between cycles (see Algorithm~\ref{alg:generic}). Specifically, let $C_v$ denote the set of colors of edges incident to $v$ in $H(\calP)$. For technical reasons, if there are fewer than two edge colors incident to $v$, we add additional colors to $C_v$ to ensure $\abs{C_v} \geq 2$; cf.\ Lines~\ref{ln:pad-colors-start}--\ref{ln:pad-colors-end}. Each cycle $\rho_v$ contains one man and one woman for each distinct color $c \in C_v$. For $(u, v) \in E$, an edge between $\rho_u$ and $\rho_v$ connects the woman in $\rho_u$ and man in $\rho_v$ corresponding to the color $\phi((u, v))$. This specifies the structure of the primal graph $G$.

The preferences on the edges of $G$ are defined as follows. Each cycle $\rho_v$ in $G$ supports two perfect matchings, one of which is preferred by all men in $\rho_v$ (cf.\ Lines~\ref{ln:man-opt-start}--\ref{ln:man-opt-end}), the other of which is preferred by all women in $\rho_v$ (cf.\ Lines~\ref{ln:woman-opt-start}--\ref{ln:woman-opt-end}). If an agent in $\rho_v$ has a neighbor in some cycle $\rho_u$ with $u \neq v$ (added in Lines~\ref{ln:edge-loop-start}--\ref{ln:edge-loop-end}),  the agent will prefer this neighbor ``in between'' their two neighbors in $\rho_v$. In the following lemmas, we show that if preferences are defined in this way, then (1) each cycle $\rho_v$ corresponds to a rotation in the SM instance, (2) every rotation in the SM instance corresponds to $\rho_v$ for some $v$, and (3) that $\rho_v$ is exposed in a stable matching if and only if all rotations $\rho_u$ for which there is an edge between a woman $w \in \rho_u$ and man $m \in \rho_v$ have been eliminated. Thus, the SM instance realizes $\calP$.

\medskip 
\begin{algorithm}
  \caption{$\ConstructInstance(H=(V,E), \phi)$ where $H=H(\calP)$ is the Hasse diagram of poset $\calP$ and $\phi$ is an edge coloring of $H$ that uses colors from $\Z^+$.  The algorithm constructs an SM instance $I(\calP, \phi)$ whose rotation poset is isomorphic to $\calP$.\label{alg:generic}}
   \begin{algorithmic}[1]
   \STATE Assume $V = \set{1, 2, \hdots, p}$. 
   \FORALL{$e=(u,v) \in E$}
       \STATE add $\phi(e)$ to the sets $C_u$ and $C_v$
   \ENDFOR
   \FOR{$v=1$ \TO $p$}\label{ln:pad-colors-start}
       \IF{$|C_v| < 2$}
         \STATE add colors $1$ and/or $2$ to $C_v$ so $|C_v| = 2$
        \ENDIF
    \ENDFOR\label{ln:pad-colors-end}
       \STATE $M \leftarrow   \set{m_{c,1} \sucht  c \in C_1} \cup \set{m_{c,2} \sucht  c \in C_2} \cup \hdots \cup \set{m_{c,p} \sucht  c \in C_p}$
    \STATE $W \leftarrow  \set{w_{c,1} \sucht  c \in C_1} \cup \set{w_{c,2} \sucht  c \in C_2} \cup \hdots \cup \set{w_{c,p} \sucht  c \in C_p}$
    \FORALL{$m_{c,v} \in M$} \label{ln:man-opt-start}
              \STATE  add $m_{c,v}$ and $w_{c,v}$ to each other's preference lists
    \ENDFOR \label{ln:man-opt-end}
    \FORALL{$e = (u, v) \in E$}\label{ln:edge-loop-start}
       \STATE $c \leftarrow  \phi((u,v))$ \label{ln:i-def}
       \STATE add $m_{c, v}$ to the front of $w_{c, u}$'s preference list \label{ln:add-man}
       \STATE add $w_{c, u}$ to the end of $m_{c, v}$'s preference list \label{ln:add-woman} 
    \ENDFOR\label{ln:edge-loop-end}
    \FOR{$v = 1$ \TO $p$} \label{ln:woman-opt-start}    
      \STATE let $\pi_v$ be some ordering of the colors  in $C_v$
      \STATE \COMMENT{We shall refer to the $i$th element in the ordering as $\pi_v(i)$ and use $b_v$ to denote $|C_v|$.}
      \STATE $\rho_v \leftarrow (m_{\pi_v(1), v},  w_{\pi_v(1), v}), (m_{\pi_v(2), v},  w_{\pi_v(2), v}), \hdots, (m_{\pi_v(b_v), v},  w_{\pi_v(b_v), v})$\label{ln:rho-def}
      \STATE \COMMENT{Note that $\rho_v$ is a circular list.}
      \FOR{$i = 1$ \TO $b_v$}
        \STATE add $w_{\pi_v(i+1), v}$ to the  end of $m_{\pi_v(i), v}$'s preference list
        \STATE add $m_{\pi_v(i), v}$ to the beginning of $w_{\pi_v(i+1), v}$'s  preference list\label{ln:add-next-man}
      \ENDFOR 
    \ENDFOR  \label{ln:woman-opt-end}   
  \end{algorithmic}
\end{algorithm}

\begin{figure}
    \begin{centering}
    \includegraphics[scale=1]{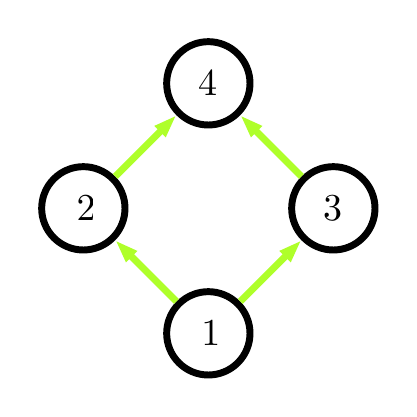}
    \hfill
    \includegraphics[scale=1]{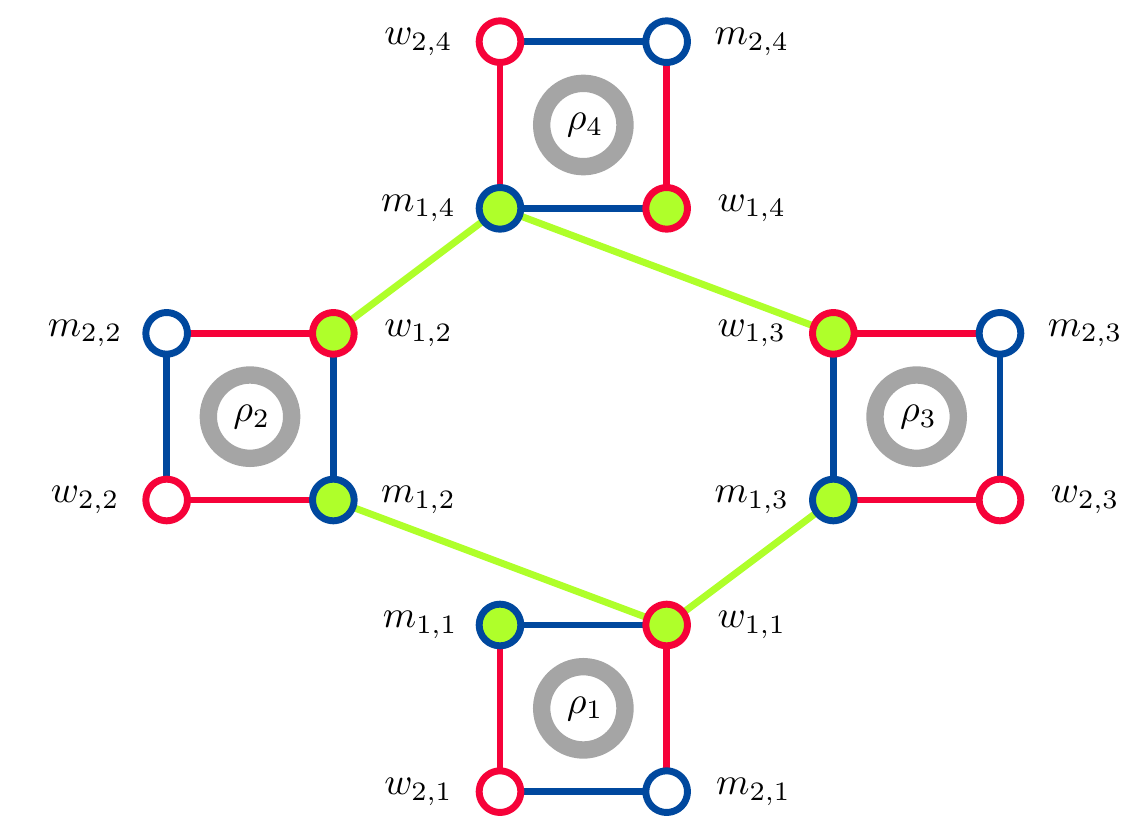}
    
    \bigskip

    \begin{minipage}{0.45\textwidth}
      \[
      \begin{array}{rllll}
        m_{1,1}: & w_{1,1} & w_{2,1}\\
        m_{2,1}: & w_{2,1} & w_{1,1} \\
        m_{1,2}: & w_{1,2} & w_{1,1} & w_{2,2}\\
        m_{2,2}: & w_{2,2} & w_{1,2}  \\
        m_{1,3}: & w_{1,3} & w_{1,1} & w_{2,3}\\
        m_{2,3}: & w_{2,3} & w_{1,3} \\
        m_{1,4}: & w_{1,4} & w_{1,2} & w_{1,3} & w_{2,4}\\
        m_{2,4}: & w_{2,4} & w_{1,4} \\
      \end{array}
      \]
    \end{minipage}
    \hfill
    \begin{minipage}{0.45\textwidth}
      \[
      \begin{array}{rllll}
        w_{1,1}: m_{2,1} &  m_{1,2} & m_{1,3} & m_{1,1} \\
        w_{2,1}: m_{1,1} & m_{2,1} \\
        w_{1,2}: m_{2,2} & m_{1,4} &  m_{1,2} \\
        w_{2,2}: m_{1,2} & m_{2,2} \\
        w_{1,3}: m_{2,3} &  m_{1,4} &  m_{1,3} \\
        w_{2,3}: m_{1,3} & m_{2,3} \\
        w_{1,4}: m_{2,4} & m_{1,4} \\
        w_{2,4}: m_{1,4} & m_{2,4} \\      
      \end{array}  
      \]
    \end{minipage}
    \end{centering}
    \caption{On the left is the Hasse diagram of a poset $\calP$  with four elements whose edges are all colored $1$ (indicated by green). Thus $C_1 = C_2 = C_3 = C_4 = \set{1, 2}$. On the right is the graph corresponding to SM instance created by $\ConstructInstance$, while the preference lists of the agents are shown below. The men are represented by nodes outlined in blue and the women by nodes outlined in red. The colors on the edges indicate preference. The blue and red edges correspond to acceptable partners added in Lines~\ref{ln:man-opt-start}--\ref{ln:man-opt-end}  and Lines~\ref{ln:woman-opt-start}--\ref{ln:woman-opt-end} respectively. The green edges correspond acceptable partners added in Lines~\ref{ln:edge-loop-start}--\ref{ln:edge-loop-end} and meant to enforce the (also green) edges of the Hasse diagram of $\calP$. The man-optimal matching consists of the blue edges, while the woman-optimal matching consists of the red edges. The rotations in the instance are alternating red-blue cycles. Notice that no green edge corresponds to a stable pair, but the green edges enforce the partial order of rotations in $\calR(I)$. For example, rotations $\rho_2$ and $\rho_3$ are only exposed in a stable matching $\mu$ in which $w_{1,1}$ is matched with $m_{2, 1}$ (her most preferred partner)---i.e., a matching in which $\rho_1$ has been eliminated.}
  \label{fig:generic-construction}
\end{figure}

Conceptually, our construction differs from that of Irving and Leather. In their algorithm, a man $m_{u,v}$ and a woman $w_{u,v}$ are created for every edge $(u,v) \in E$, and $m_{u,v}$ participates in the rotations $\rho_u$ and $\rho_v$ to enforce the precedence relation between the two rotations. That is, there is an edge from $\rho_u$ to $\rho_v$ in the rotation poset because of Rule 1 in Remark~\ref{rem:dag-structure}.    In our construction, however,   the edge is present in the rotation poset because of Rule 2 in Remark~\ref{rem:dag-structure} and the pair of agents  behind the rule is $(m_{c,v}, w_{c,u})$ which,  incidentally, is not a stable pair of the instance.

\begin{ntn}
  For a fixed poset $\calP$, element $v \in \calP$, and edge coloring $\phi$, $C_v \subseteq \Z^+$ denotes the set of colors associated with $v$. That is, $C_v$ contains the set of colors of edges incident to $v$.\footnote{In later sections, we include additional colors in $C_v$. This modification does not affect the correctness of the algorithm.} We denote $b_v = \abs{C_v}$, the number of such colors. We choose an arbitrary cyclic ordering of $C_v$, and denote this ordering by $\pi_v : [b_v] \to C_v$. That is, $\pi_v(1)$ is the ``first'' color in $C_v$, $\pi_v(2)$ is the second, and so on. Finally, given a color $c \in C_v$, we denote its ``next'' and ``previous'' colors in $\pi_v$ by $c^+$ and $c^-$, respectively.\footnote{While $c^+$ and $c^-$ depend both on $C_v$ and the permutation $\pi_v$, we suppress $\pi_v$ from the notation, as it will always be clear from context.} That is, if $c = \pi_v(i)$, then $c^+ = \pi_v(i+1)$ and $c^- = \pi_v(i-1)$. 
\end{ntn}
  
\begin{lem}
  \label{lem:generic-optimal-matchings}
  Let $I = I(\calP, \phi)$ be the SM instance created by $\ConstructInstance(H(\calP), \phi)$.   Let $\mu_0$ and $\mu_z$ denote the man-optimal and woman-optimal stable matchings of $I$ respectively.  Then
  \[
  \mu_0 = \set{\set{m_{c,v}, w_{c,v}} \sucht c \in C_v, v \in [p]}
  \]
  and 
  \[
  \mu_z = \set{\set{m_{c,v}, w_{c^+,v}} \sucht c \in C_v, v \in [p]} = \set{ \set{m_{c^-,v}, w_{c,v}} \sucht c \in C_v, v \in [p]}
  \]
  where $c^+$ and $c^-$ are respectively the colors to $c$'s next and previous colors in the (cyclic) ordering $\pi_v$ of $C_v$.  
\end{lem}
\begin{proof}
  In the preference lists created by $\ConstructInstance$,  the first choice of each man $m_{c,v}$ is $w_{c,v}$ (Lines~\ref{ln:man-opt-start}--\ref{ln:man-opt-end}). Thus the matching $\mu_0$ is a perfect matching assigning each man to his most preferred partner, which is clearly stable. On the other hand, the first choice of each woman $w_{c,v}$ is $m_{c^-,v}$ (Lines~\ref{ln:woman-opt-start}--\ref{ln:woman-opt-end}).    Again, no two women have the same first choice, so $\mu_z$ is the woman-optimal stable matching.
\end{proof}

\begin{rem}
  For our discussion below,  it is helpful to also express $\mu_0$ and $\mu_z$ in terms of the circular lists $\rho_v$ for $v \in [p]$. Recall that $b_v = \abs{C_v}$, and $\pi_v: [b_v] \rightarrow C_v$ is a cyclic ordering of colors in $C_v$ with $\pi_v(i)$ indicating the color in the $i$th position of the ordering.  Then we have
  \[
  \mu_0 =  \bigcup_{v=1}^{p} \set{ (m_{\pi_v(i), j},  w_{\pi_v(i), j}) \sucht i \in [b_v]}
  \]
  and 
  \begin{align*}
    \mu_z &=  \bigcup_{v=1}^p \set{ (m_{\pi_v(i), j},  w_{\pi_v(i+1), j}) \sucht i \in [b_v]}  
    = \bigcup_{v=1}^p \set{ (m_{\pi_v(i-1), j},  w_{\pi_v(i), j}) \sucht i \in [b_v]}. 
  \end{align*}
\end{rem}

\begin{lem}
  \label{lem:generic-exposed-matchings}
  Let $I = I(\calP, \phi)$ be the SM instance created by $\ConstructInstance(H(\calP), \phi)$.  Let $\rho_1, \rho_2, \hdots, \rho_p$ be the circular lists defined in Line~\ref{ln:rho-def} of the algorithm.  Without loss of generality, assume $1, 2, \hdots, p$ is a topological ordering of the vertices of $H(\calP)$.   Then for $v = 1, \hdots, p$, $\rho_j$ is a rotation exposed in the stable matching $\mu_{v-1}$ so that $\mu_v = \mu_{v-1}/\rho_v$  is also a stable matching of $I$.  Moreover,  $\mu_z = \mu_{p}$ so  the  set $\set{\rho_v \sucht v = 1, \ldots, p}$ contains all the rotations of $I$.

    % For $v = 1, \ldots, p$,  $\rho_v$ is  a rotation exposed in the stable matching $\mu_{v-1}$ so $\mu_v = \mu_{v-1}/ \rho_v$ is also a stable matching of $I$.    Moreover,  $\mu_z = \mu_{p}$ so  the  set $\set{\rho_v \sucht v = 1, \ldots, p}$ contains all the rotations of $I$. 
\end{lem}

\begin{proof}
  We prove the first part of the lemma by induction on $v$.  Notice that for every man $m_{c,v}$, his most preferred partner is his man-optimal stable partner, $w_{c,v}$,  while his least preferred is his woman-optimal stable partner, $w_{c^+, v}$.  Any women added by the algorithm in Lines~\ref{ln:edge-loop-start} to \ref{ln:edge-loop-end} are between these two women.  Consider vertex $1$ in $H$.  There are no edges entering $1$ because it is the first vertex in the topological ordering.  Thus, Lines~\ref{ln:edge-loop-start}--\ref{ln:edge-loop-end} do not add any women to $m_{c, 1}$'s list for any $c \in C_1$. That is, $m_{c,1}$'s list consists of just $w_{c,1}$ followed by $w_{c^+, 1}$. Therefore, 
  \[
  \rho_1 = (m_{\pi_1(1), 1},  w_{\pi_1(1), 1}), (m_{\pi_1(2), 1},  w_{\pi_1(2), 1}), \hdots, (m_{\pi_1(b_1), 1},  w_{\pi_1(b_1), 1})
  \]
  is a rotation exposed in $\mu_0$, hence $\mu_1 = \mu_0/ \rho_{1}$ is a stable matching of $I$.

  Now assume that for $v = 1$ to $j-1$, $\rho_{v}$ is a rotation exposed in the stable matching $\mu_{v-1}$.  Thus, $\mu_{j-1}$ is a stable matching of $I$. Let us now prove that $\rho_{j}$ is a rotation exposed in $\mu_{j-1}$. First, we note that
  \[
  \mu_{j-1} =  \bigcup_{v=1}^{j-1} \set{(m_{\pi_v(i),v}, w_{\pi_v(i+1),v}) \sucht i \in [b_v]} \cup \bigcup_{v=j}^{p} \set{(m_{\pi_v(i),v}, w_{\pi_v(i),v}) \sucht i \in [b_v]}.
  \]
  That is, every $m_{c,v}$ is matched to his woman-optimal stable partner if $v \leq j-1$ and to his man-optimal stable partner if $v \ge j$. 

  Next, consider $m_{c, j}$, $c \in C_{j}$. His preference list consists of $w_{c,j}$, followed by zero or more women added by Lines~\ref{ln:edge-loop-start}--\ref{ln:edge-loop-end}, and ends with $w_{c^+, j}$. The women added by Lines~\ref{ln:edge-loop-start}--\ref{ln:edge-loop-end}  are of the form $w_{c,u}$ where $(u, j)$ is an edge of $H$ and $c = \phi((u,j))$. But $u$ precedes $j$ in $H$ so it follows by the inductive hypothesis that $w_{c,u}$ is matched to $m_{c^{-}, u}$---her woman-optimal stable partner---in $\mu_{j-1}$. Thus, she does not prefer $m_{c,j}$ to her current partner in $\mu_{j-1}$. But $w_{c^+ j}$ does prefer $m_{c,j}$---her woman-optimal stable partner---to her current partner $m_{c^+,j}$ in $\mu_{j-1}$.  Hence, 
  \[
  \rho_{j} = (m_{\pi_{j}(1), j}, w_{\pi_{j}(1), j}), (m_{\pi_{j}(2), j}, w_{\pi_{j}(2), j}), \cdots,  (m_{\pi_{j}(b_{j}), j}, w_{\pi_{j}(b_{j}),j})
  \]
  is a rotation exposed in $\mu_{j-1}$ so $\mu_{j} =  \mu_{j-1}/ \rho_{j}$ is a stable matching of $I$.  

  By induction, the first part of the lemma is true.  It is easy to see that $\mu_z = \mu_p$,  which was obtained by eliminating all the rotations in the set $\set{\rho_v \sucht v = 1, \ldots, p}$.  Since $\mu_z$ can only be obtained by eliminating all rotations in $\calR(I)$, it follows that  $\calR(I) = \set{\rho_v \sucht v = 1, \ldots, p}$, as desired. %[ADD LEMMA]
\end{proof}
  
\begin{rem}
For each $v \in [p]$, let $A_v = \set{m_{c,v} \sucht c \in C_v } \cup \set{w_{c,v} \sucht c \in C_v }$.    Lemma~\ref{lem:generic-exposed-matchings}'s characterization of the rotations of $I$ means that every agent in $A_v$ is part of only one rotation: $\rho_v$.  Moreover, eliminating $\rho_v$ will shift each agent from their man-optimal stable partner to their woman-optimal stable partner.  

The above property highlights the simplicity of SM instance $I$.  In {\it every} stable matching of $I$,  {\it each} agent is matched to either their man-optimal or their woman-optimal stable partner.  But we note that it does not imply that a matching where every agent is paired to either their man-optimal or their woman-optimal stable partner is a stable matching of $I$.   For example, for any $u, v \in [p]$, if $u$ precedes $v$ in $\calP$, then there is {\it no} stable matching where the agents in $A_u$ are matched to their man-optimal stable partners while the agents in $A_v$ are matched to their woman-optimal stable partners.  
\end{rem}  
  
\begin{rem}
\label{rem:toporder}
We also note that Lemma~\ref{lem:generic-exposed-matchings} applies to every topological ordering $\tau$ of the vertices of $H$.  Let $\tau(i)$ denote the $i$th element in the ordering.  Let $\sigma_0 = \mu_0$.  Then for $i = 1, \hdots, p$,  $\rho_{\tau(i)}$ is exposed in the stable matching $\sigma_{i-1}$ and $\sigma_{i} = \sigma_{i-1} / \rho_{\tau(i)}$.  
\end{rem}

\begin{lem}
  \label{lem:generic-isomorphism}
  Let $\calP$ be a finite poset, $H(\calP) = (V, E)$  its Hasse diagram, and $\phi$ an arbitrary edge coloring of $H(\calP)$. Let $I = I(\calP, \phi)$ be the SM instance created by $\ConstructInstance(H(\calP), \phi)$. Then the relation $f(v) = \rho_v$ is an isomorphism between $\calP$ and $\calR(I)$.
\end{lem}

\begin{proof}
We have already established that $\calP$ and $\calR = \calR(I)$ have the same number of elements. Let $v', v$ be elements of $\calP$. We will now argue that $v'$ is a predecessor of $v$ in $\calP$ if and only if $\rho_{v'}$ is a predecessor of $\rho_{v}$ in  $\calR$.

Assume $v'$ is a predecessor of $v$.  Then there exists a sequence $u_0 = v'$, $u_1, u_2, \ldots, u_{r-1}, u_r$ with $u_r = v$, such that $u_j$ is an immediate predecessor of $u_{j+1}$ for $j = 0, \hdots, r-1$.  Let $c_j = \phi((u_j, u_{j+1})$, the color assigned to edge $(u_j, u_{j+1})$ in $H(\calP)$. Then, by construction, each $w_{c_j, u_j}$ is in the preference list of $m_{c_j, u_{j+1}}$, and $m_{c_j, u_{j+1}}$ prefers $w_{c_j, u_j}$ to his woman-optimal stable partner.  He is also in $w_{c_j, u_j}$'s preference list and she prefers him to her man-optimal stable partner. By Lemma~\ref{lem:generic-exposed-matchings} the only rotation that contains $w_{c_j, u_j}$ is $\rho_{u_j}$. Thus, in order for $\rho_{u_{j+1}}$ to be exposed in a stable matching $\mu$,  $\rho_{u_j}$ must be eliminated.  Since this is true for $j = 0, \hdots, r-1$, it follows that $\rho_{v'}$ has to be eliminated before $\rho_v$ is exposed. In other words, $\rho_{v'}$ precedes $\rho_v$ in $\calR$.  
By the same argument, if $v$ is a predecessor of $v'$ then $\rho_{v}$ precedes $\rho_{v'}$ in $\calR$. 

So the only case we have to consider is when $v'$ and $v$ are incomparable in $\calP$. In this case, there is a topological orderings $\tau_1$ and $\tau_2$ such that $v'$ occurs before $v$ in $\tau_1$, while $v$ occurs before $v$ in $\tau_2$.  From  Lemma~\ref{lem:generic-exposed-matchings} and Remark~\ref{rem:toporder},  there is a stable matching where $\rho_v'$ is eliminated but not $\rho_v$ and another stable matching where $\rho_v$ is eliminated but not $\rho_{v'}$.  Thus, $\rho_v$ and $\rho_{v'}$ are not comparable in $\calR(I)$. 
\end{proof}

%\wnote{Do we need to account for the number of colors in the runtime of our algorithm?}
\begin{lthm}
  \label{thm:generic-construct}
  Let $\calP$ be a finite poset and $H(\calP) = (V,E)$ be  its Hasse diagram with $p = |V|$, and $q = |E|$.  Then $I = I(\calP, \phi)$---the SM instance created by $\ConstructInstance(H(\calP), \phi)$---realizes $\calP$. It has $O(p+q)$ agents and it can be constructed in $O(p + q)$ time.   
  \end{lthm}
\begin{proof}
  By  Lemma~\ref{lem:generic-isomorphism}, we know that $I$ realizes $\calP$.  $\ConstructInstance$ creates at most $2 \times \max(2, \deg(v))$ agents for each $v \in [p]$.  Thus, the total number of agents is at most $2(2p + 2q) = O(p+q)$. Creating $C_v$ for $v \in [p]$ takes $O(\deg(v) + 1)$ time, so all lists are formed in time $O(\sum_{v \in V} (\deg(v)+1)) = O(p+ q)$. Creating the agents, and adding their man-optimal and woman-optimal stable partners take $O(p+q)$ time in total. Finally, adding their acceptable partners in Lines~\ref{ln:edge-loop-start}--\ref{ln:edge-loop-end} takes $O(q)$ time. Each iteration of the outer for loop in lines 20 to 29 takes $O(|C_v|)$ time so altogether the for loop runs in $O(\sum_v |C_v|) = O(p+q)$. Thus, the total running time of $\ConstructInstance$ is $O(p+q)$. 
\end{proof}

Some of  the agents in the SM instance created by $\ConstructInstance$ have incomplete preference lists.  If desired, each incomplete  preference list can be completed by appending its missing agents at the end of the list.   By Corollary~\ref{cor:shortlists-2}, the SM instance will realize the same poset.

%Notice that the man-optimal and woman-optimal stable matchings of the new instance are the same as the old one.   Since in the new instance the ``extra" agents were added at the end of each incomplete preference list,  the two instances also have the identical symmetric shortlists.  Thus, by Corollary~\ref{cor:shortlists}, the rotation posets of the two instances are isomorphic.

%By appending arbitrary acceptable partners to an instance $I$ created by $\ConstructInstance$, the new acceptable partners not preferred to the man- or woman-pessimal stable partners for either gender. Therefore, appending to the preference lists does not affect the set or structure of stable matchings for the instance. Thus, we obtain the following consequence.

%\begin{cor}
 % \label{cor:generic-complete}
 % Let $I = I(\calP, \phi)$ be the SM instance created by $\ConstructInstance(H(\calP), \phi)$.  Let $I'$ be an SM instance with complete preference lists obtained from $I$ by adding to the end of  each incomplete preference list its missing agents.   Then $I'$ also realizes $\calP$.
%\end{cor}

\begin{cor}
  \label{cor:generic-complete}
  Let $\calP$ be a finite poset and $H(\calP) = (V,E)$ be  its Hasse diagram with $p = |V|$, and $q = |E|$.  Then there exists an SM instance $I(\calP)$ of size $n = 2p$ that realizes $\calP$, and $I$ can be constructed in time $O(p+q)$ given $H(\calP)$.  
\end{cor}
\begin{proof}
Let $\phi$ assign each edge of $H(\calP)$ the color $1$.  Then $C_v = \set{1, 2}$ for each $v \in [p]$.  The instance $I(\calP, \phi)$ realizes $\calP$.  It has $n$ men and $n$ women where $n = 2p$ and can be constructed in time $O(p + q)$ time.  
\end{proof}

%% \wnote{Move up earlier?}
\begin{rem}
  The runtime of our construction in Corollary~\ref{cor:generic-complete} matches that of Irving and Leather's construction in~\cite{Irving1986-complexity}. Our construction also yields an instance $I$ whose size (i.e., number of agents) is dependent only on $p$ whereas Irving and Leather's is dependent on $q$ and the number of bottom and top nodes of $H(\calP)$.  Thus, our construction gives a quadratic improvement when $q = \Omega(p^2)$ (for example when $H(\calP)$ is a directed complete bipartite graph).  
\end{rem}

\section{\texorpdfstring{$k$}{k}-Bounded and $k$-Attribute Preferences}
\label{sec:k-bounded}

In Corollary~\ref{cor:generic-complete}, we showed that applying our generic construction, $\ConstructInstance$, with the with the trivial edge coloring $\phi(e) = 1$ for all $e \in E$ yields an SM instance with a small number of agents ($n = O(p)$). In this section, we show that by taking $\phi$ to be a \emph{proper} edge coloring\footnote{Recall that an edge coloring is \dft{proper} if for every vertex $v$, no two edges incident to $v$ have the same color.}, $\ConstructInstance$ computes an SM instance in which all preference lists have length at most $3$---i.e., a $3$-bounded instance. Thus we can apply Theorem~\ref{thm:generic-construct} to show that for any finite poset $\calP$, there is a $3$-bounded SM instance realizing $\calP$. A proper edge coloring of any graph $H$ can be computed in linear time, so the $3$-bounded SM instance realizing $\calP$ can be computed in linear time as well.

Next, we show that the $3$-bounded instances constructed by $\ConstructInstance$ as above can be efficiently transformed into $6$-attribute instances realizing the same rotation poset. Thus, $6$-attribute preferences realize all finite posets. More generally, our reduction shows that for any $k$-bounded SM instance, there is a $2k$-attribute instance in which all agents' rankings of their first $k$ acceptable partners are the same as the $k$-bounded instance.

\subsection{\texorpdfstring{$k$}{k}-Bounded Preferences}

Let $\phi$ be an edge coloring of a directed graph $H=(V,E)$.  We say that $\phi$ is a \dft{proper in-coloring} of $H$ if for every vertex $v \in V$ and any two edges $e$ and $e'$ entering $v$, $\phi(e) \neq \phi(e')$.  Similarly, $\phi$ is a \dft{proper out-coloring} of $H$ if for every vertex $v \in V$ and any two edges $e$ and $e'$ leaving $v$, $\phi(e) \neq \phi(e')$.   We then say that $\phi$ is a \dft{super coloring} of $H$ if it is both a proper in- and out-coloring. The following proposition shows that given a proper in-coloring (respectively, out-coloring), $\ConstructInstance$ produces an SM instance in which the men (respectively, women) have preference lists of length at most $3$.

\begin{prop}
\label{prop:in-n-out-coloring}
Let $H = H(\calP)$ be the Hasse diagram of poset $\calP$ and $\phi$ be an edge coloring of $H$. Let $I = I(\calP, \phi)$ be the SM instance created by $\ConstructInstance(H, \phi)$. The men's preference lists in $I$ have length at most $3$ if and only if $\phi$ is a proper in-coloring while the women's pereference lists in $I$ have length at most $3$ if and only if $\phi$ is a proper out-coloring.
\end{prop}
\begin{proof}
Consider an agent $m_{c,v}$ with $c  \in C_v, v \in [p]$.  By construction,  the length of $m_{c,v}$'s preference list is determined by the edges entering $v$ that $\phi$ has colored $c$.  In particular,  $m_{c,v}$'s preference list has length $\ell+2$ if and only if there are $\ell$ such edges.  Thus, $m_{c,v}$'s preference list has length $3$ if and only if $\phi$ assigned at most one edge entering $v$ the color $c$. This property holds for \emph{all} men if and only if $\phi$ is a proper in-coloring.  

Next, consider an agent $w_{c,v}$ with $c  \in C_v, v \in [p]$.  The length of $w_{c,v}$'s preference list  is determined by the edges leaving $v$ that $\phi$ colors $c$.  Using the same argument as above, we conclude that all women's preference list have length at most $3$ if and only if $\phi$ is a proper out-coloring.
\end{proof}

We can now prove  that  $3$-bounded SM instances can realize any finite poset. 

\begin{lthm}
 \label{thm:k-bound}
  Let $\calP$ be a finite poset.  Then there is an SM instance $I(\calP) \in \bound(3)$  that realizes $\calP$.  Moreover, given the Hasse diagram $H = H(\calP)$ with $p$ vertices and $q$ edges,  $I(\calP)$ has $O(p+q)$ agents and can be constructed in $O(p+q)$ time. Thus, {\it all} finite posets can be efficiently realized by instances in $\bound(k)$ for any $k \geq 3$.
\end{lthm} 
\begin{proof}
According to Proposition~\ref{prop:in-n-out-coloring}, it suffices to find $\phi$ that is a super coloring of $H$ so that $\ConstructInstance$ produces a $3$-bounded SM instance. So let $\phi$ assign edges of $H$ pair-wise distinct colors from $[q]$. It takes $O(p+q)$ time to create $\phi$. $\ConstructInstance$ will run in $O(p+q)$ time to create an SM instance $I(\calP, \phi)$ that is $3$-bounded and has $O(p+q)$ agents.  
\end{proof}

\begin{figure}
  \begin{centering}
    \includegraphics[scale=1]{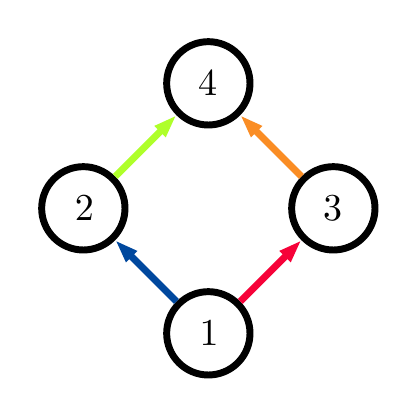}
    \hfill
    \includegraphics[scale=1]{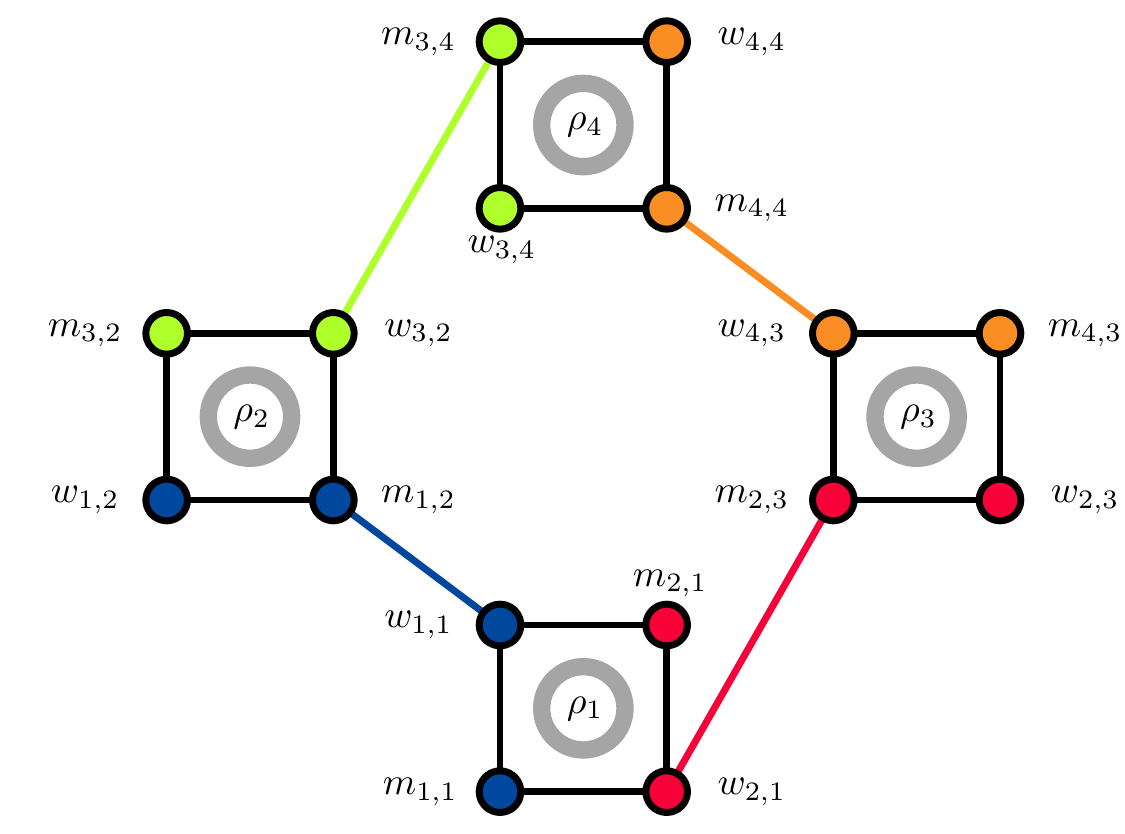}
        \bigskip

    \begin{minipage}{0.45\textwidth}
      \[
      \begin{array}{rllll}
        m_{1,1}: & w_{1,1} & w_{2,1}\\
        m_{2,1}: & w_{2,1} & w_{1,1} \\
        m_{1,2}: & w_{1,2} & w_{1,1} & w_{3,2}\\
        m_{3,2}: & w_{3,2} & w_{1,2}\\
        m_{2,3}: & w_{2,3} & w_{2,1} & w_{4,3}  \\
        m_{4,3}: & w_{4,3} & w_{2,3} \\
        m_{3,4}: & w_{3,4} & w_{3,2} & w_{4,4}\\
        m_{4,4}: & w_{4,4} & w_{4,3} & w_{3,4} \\
      \end{array}
      \]
    \end{minipage}
    \hfill
    \begin{minipage}{0.45\textwidth}
      \[
      \begin{array}{rllll}
        w_{1,1}: m_{2,1} &  m_{1,2} & m_{1,1} \\
        w_{2,1}: m_{1,1} & m_{2,3} & m_{2,1} \\
        w_{1,2}: m_{3,2} &  m_{1,2} \\
        w_{3,2}: m_{1,2} & m_{3,4} & m_{3,2} \\
        w_{2,3}: m_{4,3} & m_{2,3} \\
        w_{4,3}: m_{2,3} & m_{4,4} & m_{4,3}\\
        w_{3,4}: m_{4,4} & m_{3,4} \\
        w_{4,4}: m_{3,4} & m_{4,4} \\      
      \end{array}  
      \]
    \end{minipage}
  \end{centering}
  \caption{The $3$-bounded instance (right) formed by taking $\phi$ to be a proper edge coloing of the Hasse diagram (left) of a poset $\calP$. The simplicial complexes associated with the instance on the right are depicted in Figure~\ref{fig:complexes}}
  \label{fig:construct-bound}
\end{figure}

\subsection{\texorpdfstring{$k$}{k}-Attribute Preferences}
\label{sec:k-attribute}

Recall that for an SM instance in $\attr(k)$, every agent $a$ is assigned a profile $(\veca, \varphi_a)$ such that $\veca \in \R^k$ and $\varphi_a:  \R^k \rightarrow \R$.  Agent $a$'s preference list is then derived by applying $\varphi_a$ to all agents $b$ of the opposite gender and ranking them in descending order according to $\varphi_a(\vecb)$. Below, we show how an instance $I \in \bound(3)$ can be converted into an instance $I' \in \attr(6)$ such that $I$ and $I'$ have isomorphic rotation posets. Since we know from Theorem~\ref{thm:k-bound} that \emph{every} finite poset can be realized by some instance in $\bound(3)$, the same holds for $\attr(6)$. 

The idea behind our construction of the profiles is the following. Consider an embedding of the set $W$ into $\R^6$ such that $W$ is in convex position\footnote{We abuse notation and associate $W$ with the set of embedded points in $\R^6$.} (i.e., no $\vecw$ lies inside of the convex hull of $W \setminus\set{\vecw}$). Let $X_W \subseteq \R^6$ denote the convex hull of $W$. Suppose the points $\vecw_1, \vecw_2, \vecw_3$ form a $3$-face of the polytope $X_W$. Then there exists a linear function $\varphi : \R^6 \to \R$ such that $\varphi(\vecw_1) = \varphi(\vecw_2) = \varphi(\vecw_3) > \varphi(\vecw)$ for all $\vecw \in W$, $\vecw \neq \vecw_1, \vecw_2, \vecw_3$. By choosing a sufficiently small perturbation $\hat\varphi$ of $\varphi$, we can ensure that in fact
\[
\hat\varphi(\vecw_1) > \hat\varphi(\vecw_2) > \hat\varphi(\vecw_3) > \hat\varphi(\vecw).
\]
Thus, if a man $m \in M$ takes $\varphi_m = \hat\varphi$, he will rank $w_1$, $w_2$, $w_3$ as his top three choices. Our argument is to show that for any SM instance $I \in \bound(3)$, there exist embeddings of $W$ and $M$ in $\R^6$ such that \emph{every} triple of women (respectively, men) form a $3$-face as above. Thus, each man $m$ can choose $\varphi_m$ so that his first (at most) $3$ preferred partners according to $\varphi_m$ agree with his preference list in $I$. From this, it is easy to show that the constructed instance in $\attr(6)$ has precisely the same stable matchings as $I$.

\begin{rem}
  Our construction of $6$-attribute instances from $3$-bounded instances generalizes the technique used in Bhatnagar et al.'s construction of $3$-attribute preferences that realize star posets~\cite{Bhatnagar2008-sampling}. They further describe a modification showing that, in fact, the same rotation posets can be realized in the $2$-attribute model. We believe that a similar dimension reduction argument may allow one to reduce the dimension of our general embedding, thus showing that $k$-attribute instances realize all finite posets for some $k < 6$.
\end{rem}

\begin{figure}
  \begin{centering}
    \hfill
    \includegraphics[scale=1]{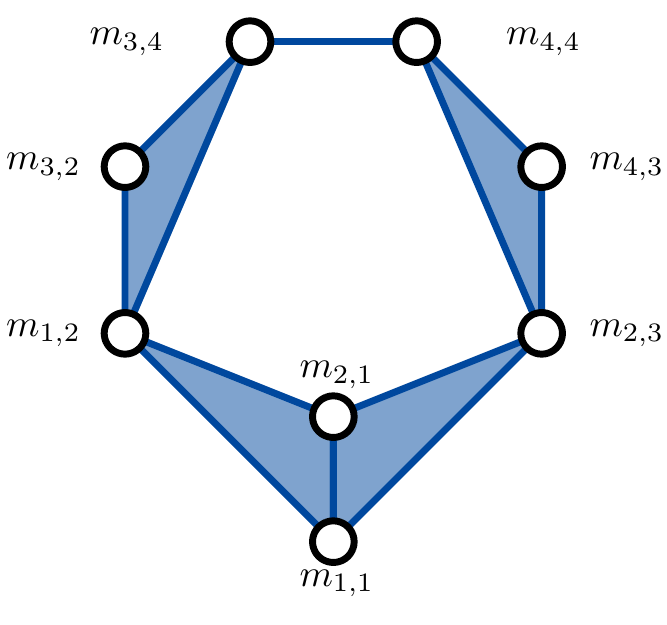}
    \hfill
    \includegraphics[scale=1]{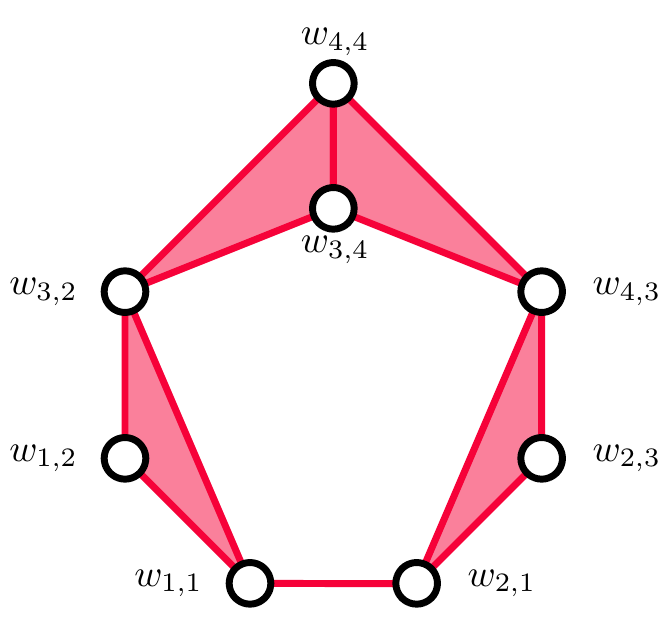}
    \hfill
  \end{centering}
  \caption{We can associate simplicial complexes $M$ (left) and $W$ (right) for the $3$-bounded SM instance depicted in Figure~\ref{fig:construct-bound}. Each face in the complexes corresponds to a set of agents appearing together in some agent's preference lists. For example, the men $\set{m_{1,1}, m_{2, 1}, m_{1,2}}$ form a face on the left, because they all appear on $w_{2,1}$'s preference list. For this example, both complexes are planar. Therefore, they can be embedded in convex position on the surface of the sphere in $\R^3$. Thus, the rotation poset can be realized in the $3$-attribute model. In general, Gale's construction implies that any $(k-1)$-dimensional complex (in particular, those associated with a $k$-bounded SM instance) can be embedded in convex position in $\R^{2k}$.}
  \label{fig:complexes}
\end{figure}

The main technical tool we require is a result of Gale~\cite{Gale1963-neighborly}, who showed that the \emph{cyclic polytopes} in $2 \ell$-dimensional space are \emph{$\ell$-neighborly}.

\begin{dfn}
  \label{dfn:cyclic-polytope}
  For any positive integers $\ell$ and $n$, consider the vertex set $V_{2 \ell, n} \subseteq \R^{2\ell}$ defined by
  \[
  V_{2 \ell, n} = \set{(i, i^2, i^3, \ldots, i^{2 \ell}) \sucht i \in [n]}.
  \]
  Then the \dft{$2\ell$-dimensional cylic polytope} of \dft{order $n$}, $X_{2 \ell, n}$, is defined to be the convex hull of $V_{2 \ell, n}$.\footnote{The proof of Corollary~\ref{cor:cyclic-polytope} implies that $V_{2 \ell, n}$ are in convex position, hence the set vertices of $X_{2 \ell, n}$ is precisely $V_{2 \ell, n}$}
\end{dfn}

Gale~\cite{Gale1963-neighborly} showed that the $2 \ell$-dimensional cyclic polytopes have the remarkable property of being \emph{$\ell$-neighborly}: every subset $V \subset V_{2 \ell, n}$ of size $\ell$ forms an $\ell$-face of $X_{2 \ell, n}$. More specifically, Gale proved the following lemma.

\begin{lem}[cf.~{\cite[Theorem~1]{Gale1963-neighborly}}]
  \label{lem:cyclic-polytope}
  Let $X_{2 \ell, n}$ be a cyclic polytope and let $V = \set{\vecv_1, \vecv_2, \ldots, \vecv_\ell}$ be any set of $\ell$ distinct vertices in $X_{2 \ell, n}$. Then there exists a linear function $\varphi \colon \R^{2 \ell} \to \R$ such that for all $\vecv_i, \vecv_j \in V$, (1) we have $\varphi(\vecv_i) = \varphi(\vecv_j)$, and (2) for all vertices $\vecv \notin V$, $\varphi(\vecv) < \varphi(\vecv_i)$. 
\end{lem}

Gale~\cite{Gale1963-neighborly} constructs the function $\varphi$ above as follows. For each $\vecv_j \in V$, write $\vecv_j = (i_j, i_j^2, \ldots, i_j^{2 \ell})$, and consider the function
\begin{equation}
  \label{eqn:poly}
  p(t) = \prod_{j = 1}^\ell (t - i_j)^2 = \beta_0 + \beta_1 t + \cdots + \beta_{2 \ell - 1} t^{2 \ell - 1} + t^{2 \ell}.  
\end{equation}
Setting $\vecbeta = (\beta_1, \beta_2, \ldots, \beta_{2 \ell - 1}, 1) \in \R^{2\ell}$, a straightforward argument shows that $\varphi(\vecx) = - \vecx \cdot \vecbeta$ gives $\varphi(\vecv_j) = \beta_0$ for all $\vecv_j$, and $\varphi(\vecv) < \beta_0$ for every vertex $\vecv \notin V$. In particular, computing the vector $\vecbeta$ can be performed in time $O(\ell^2)$ by simply performing the multiplication in~(\ref{eqn:poly}). By choosing a sufficiently small perturbation of the function $\varphi$ of the conclusion of Lemma~\ref{lem:cyclic-polytope}, we obtain the following corollary. We give a self-contained proof, in particular to demonstrate that the desired functions $\varphi$ can be efficiently constructed.

\begin{cor}
  \label{cor:cyclic-polytope}
  Let $X_{6, n}$ be a cyclic polytope and let $V = \set{\vecv_1, \vecv_2, \vecv_3}$ be any set of $3$ distinct vertices in $X_{6, n}$. Then there exists a linear function $\varphi \colon \R^{6} \to \R$ such that
  \[
  \varphi(\vecv_1) > \varphi(\vecv_2) > \varphi(\vecv_3) > \varphi(\vecv)
  \]
  for all vertices $\vecv \notin V$. The function $\varphi$ can be computed in constant time from the set $V$.
\end{cor}
\begin{proof}
  In order to give an explicit construction of a function $\varphi$ as in Corollary~\ref{cor:cyclic-polytope}, consider the function
  \begin{equation}
    \label{eqn:q}
    q(t) = (t - i_1)^2 (t - i_2 + \delta_2)^2 (t - i_3 + \delta_3)^2 = \alpha_0 + \alpha_1 t + \cdots + \alpha_5 t^5 + t^6,
  \end{equation}
  where $\delta_2, \delta_3 > 0$ are constants to be chosen later. Setting $\vecalpha = (\alpha_1, \alpha_2, \ldots, \alpha_5, 1)$ we obtain
  \[
  q(j) = \vecalpha \cdot (j, j^2, \ldots, j^5, 1) + \alpha_0.
  \]
  Thus, it suffices to compute suitable constants $\delta_2, \delta_3$ such that
  \begin{equation}
    \label{eqn:q-ineq}
    q(i_1) < q(i_2) < q(i_3) < q(i)
  \end{equation}
  for all $i \neq i_1, i_2, i_3$ (where as above, we take $\vecv_j = (i_j, i_j^2, \ldots, i_j^6)$). Indeed, then taking $\varphi(\vecv) = - \vecalpha \cdot \vecv$ gives the desired result. Observe that from the definition of $q$, we immediately obtain
  \begin{align*}
    q(i_1) &= 0\\
    q(i_2) &= (i_2 - i_1)^2 (\delta_2^2) (i_2 - i_3 - \delta_3)^2\\
    q(i_3) &= (i_3 - i_1)^2 (i_3 - i_2 - \delta_2)^2 (\delta_3^2)\\
    q(i) &= (i - i_1)^2 (i - i_2 + \delta_2)^2 (i - i_3 + \delta_3)^2.
  \end{align*}
  Taking $\delta_2, \delta_3 \leq 1/4$ (so that, for example, $(i - \delta_2)^2 > 1/2$), we obtain the following inequalities:
  \begin{align}
    0 &< q(i_2) \leq \delta_2^2 n^4 \label{eqn:d2}\\
    \frac 1 2 \delta_3^2 &\leq q(i_3) \leq n^4 \delta_3^2 \label{eqn:d3}\\
    \frac 1 4  &< q(i) \label{eqn:qi}.
  \end{align}
  By taking $\delta_2 = \frac{1}{4 n^4}$ and $\delta_3 = \frac{1}{2 n^2}$, we obtain
  \[
  \delta_2^2 n^4 = \frac{1}{16 n^4} < \frac{1}{8 n^4} = \frac 1 2 \delta_3^2
  \]
  and $n^4 \delta_3^2 = \frac 1 4$, hence Equation~(\ref{eqn:q-ineq}) is satisfied. Since~(\ref{eqn:q}) contains $O(1)$ terms, the coefficients $\alpha$ can be computed in time $O(1)$ from $V$. This gives the desired result.
\end{proof}

We are now ready to prove the main result of this section.

\begin{lthm}
  \label{thm:k-attr}
  Let $\calP$ be a poset whose Hasse diagram $H(\calP)$ has $p$ vertices and $q$ edges. There is an SM instance $I''(\calP) \in \attr(6)$ with $O(p+q)$ agents that realizes $\calP$.  Moreover, given $H(\calP)$, the profiles of the agents in $I''(\calP)$ can also be computed in $O(p+q)$ time. Thus, {\it all} finite posets can be efficiently realized by instances in $\attr(k)$ for any $k \geq 6$.
\end{lthm}
\begin{proof}
  Let  $I = I(\calP)$ be the $3$-bounded instance constructed by $\ConstructInstance(H(\calP), \phi)$, where $\phi$ is any proper edge coloring of $H$. We create an $\attr(6)$ instance $I'' = I''(\calP)$ with exactly the same agents as $I(\calP)$ as follows. Let $n$ be the number of men (and women) in $I$.  Arbitrarily label the men and women as $m_1, m_2, \hdots, m_n$ and $w_1, w_2, \hdots, w_n$ respectively. For each $m_i \in M$ and $w_i \in W$, we assign a profile vector in $\R^6$ corresponding to the $i\th$ vertex of the $6$-dimensional cyclic polytope $X_{6,n}$:
  \[
  \vecm_i, \vecw_i = (i, i^2, i^3, i^4, i^5, i^6).
  \]
  Now let $P_{m_i}$ and $P_{w_i}$ denote $m_i$'s and $w_i$'s preference lists in $I$. We assume that $\abs{P_{m_i}} = \abs{P_{w_i}} = 3$ by possibly adding another arbitrary person to the end of these preference lists. For $P_{m_i} =  w_{s_1}, w_{s_2}, w_{s_3}$ assign the profile function $\varphi_{m_i}$ to be the function asserted by Corollary~\ref{cor:cyclic-polytope} such that
  \[
  \varphi_{m_i}(\vecw_{s_1}) > \varphi_{m_i}(\vecw_{s_2}) > \varphi_{m_i}(\vecw_{s_3}) > \varphi_{m_i}(\vecw) \quad\text{for all } w \in W \setminus P_{m_i}.
  \]
  Thus, the preference list of $m_i$ in $I''$ consists of the same list in $I$ followed by other women. The women's preference lists are defined analogously.  Again, the preference list of each woman $w_i$ in $I''$ consists of the same list in $I$ followed by other men. By Corollary~\ref{cor:shortlists-2}, $I''$ realizes the same rotation poset as $I$, namely $\calP$.

  Finally computing $I$ takes time $O(p + q)$. Given $I$, each profile $(a, \varphi_a)$ can be computed in time $O(1)$ by Corollary~\ref{cor:cyclic-polytope}. Therefore, the overall runtime is $O(p + q)$, as desired.
\end{proof}

 %% \cnote{Will, is it too much work to take the example in Figure 2 and compute the profiles for each agent?} \wnote{We could definitely do this at some point (and write out the explicit functions we construct), but I think this is fairly low priority for now.}

\begin{rem}
  The construction we describe generalizes in the following way: given any SM instance $I \in \bound(k)$, we can construct an instance $I' \in \attr(2 k)$ such that each preference list in $I'$ is obtained by appending the missing agents to the end of the incomplete preference lists in $I$.  We only give full details for the $k = 3$ case, as this is sufficient to prove Theorem~\ref{thm:k-attr}.
\end{rem}

\section{\texorpdfstring{$k$}{k}-List Preferences}
\label{sec:k-list}

%For this section, we will we prove that $(2,2)-\ML$ SM instances realize every finite poset. We will then complete the preference lists of these instances to show that SM instances in $\bigcup_{n \in \Z^+} (2, n)$---that is, instances in which the men, say, have only two distinct (complete) preference lists---also realize all finite posets. Recall that if an SM instance has $(k_1, k_2)-\ML$ preferences, then (1) there exists $k_1$ complete lists of women so that each man's preference list is a sublist of one of them and (2) there exists $k_2$ complete lists of men so that each woman's preference list is a sublist of one of them.

%% \cnote{Changes were cosmetic in nature.  Fixed subscript issue.  Changed $M_1$ and $M_0$ to $GM_2$ and $GM_1$ respectively.  I'm saving the notation $M_i$ for the $k$-range discussion.  Lastly,  I elaborated on the shortlists argument at the end. That might be something you'd like to check for correctness.}

For this section, we will prove that $(2, \infty)$-list (or, equivalently, $(\infty, 2)$-list) SM instances realize every finite poset.  We will start by creating an SM instance using $\ConstructInstance$ where the incomplete preference lists of each man is a sublist of one of two distinct complete lists of women.  (Interestingly, the incomplete preference lists of the women are sublists of also two distinct complete lists of men.) We then assign the complete lists as preference lists for the men and append  the women's preference lists with missing men so they become complete.  Hence, the men can be divided into two groups and the men in each group have the same preference list while the women can be divided into any number of groups with the same property.  

To arrive at the right SM instance for poset $\calP$, we will have to label the vertices of $H = H(\calP)$---the Hasse diagram of $\calP$---in a particular way and tweak how $\ConstructInstance$ is implemented. We describe the specifications and their implication below. 
\begin{itemize}
\item \emph{Label the vertices of $H$  so that $(p, p-1, \hdots, 1)$ is a topological ordering.} This means every (directed) edge $(u,v)$ has $u > v$.
\item \emph{For each edge $(u,v)$ of $H$, let $\phi((u,v)) = u$.}  Thus, $\phi$ is a proper in-coloring since all edges entering a particular node are assigned different colors.
\item \emph{Add  the color $p+1$ to $C_v$ for every $v \in [p]$.}  Such an addition creates $2p$ more agents and expands the rotations, but does not affect the fact that if the input to the algorithm is $(H, \phi)$, the output $I(\calP, \phi)$ still realizes $\calP$.  
\item \emph{Process the edges in $E$ in Lines \ref{ln:edge-loop-start} to \ref{ln:edge-loop-end} in lexicographically decreasing order. }  Thus, if a woman's preference list contains men different from her woman- and man-optimal stable partners, the subscripts of the other men in her preferences are lexicographically increasing. 
\item {Finally, for each $v \in [p]$, let $\pi_v$ be an ordering of $C_v$ so that the colors are listed from smallest to largest.} Notice that $p+1$ is the last color on this list. 
\end{itemize}

\begin{figure}
    \begin{centering}
    \includegraphics[scale=1]{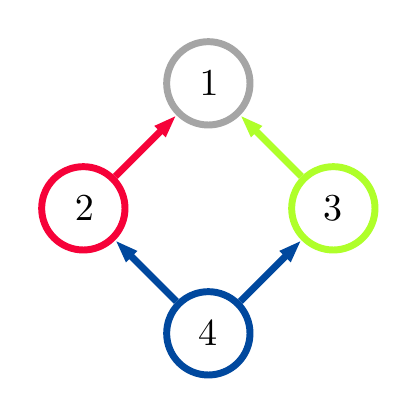}
    \hfill
    \includegraphics[scale=1]{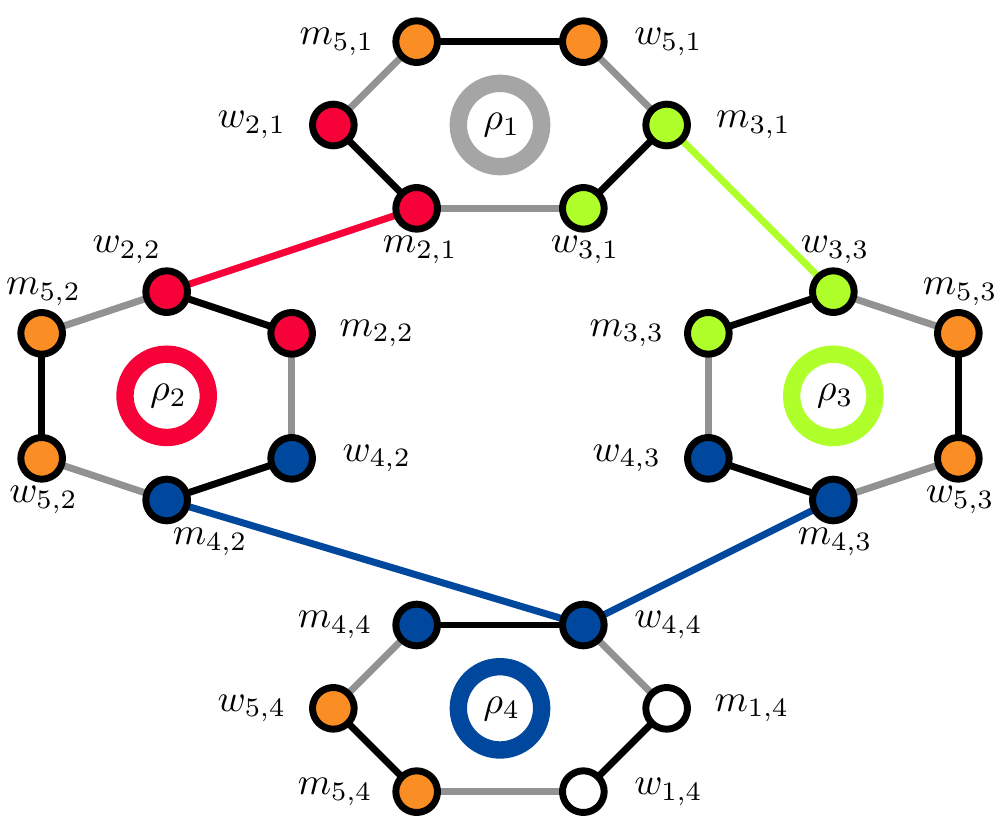}
    
    \bigskip
    
    \begin{minipage}{0.45\textwidth}
      \[
      \begin{split}
        GM_1 & \left\{
        \begin{array}{rllll}
          m_{4,4}: & w_{4,4} & w_{5,4} \\
          m_{4,2}: & w_{4,2} & w_{4,4} & w_{5,2}\\
          m_{2,2}: & w_{2,2} & w_{4,2} \\
          m_{4,3}: & w_{4,3} & w_{4,4} & w_{5,3} \\
          m_{3,3}: & w_{3,3} & w_{4,3} \\
          m_{3,1}: & w_{3,1} & w_{3,3} & w_{5,1} \\
          m_{2,1}: & w_{2,1} & w_{2,2} & w_{3,1}\\
          m_{1,4}: & w_{1,4} & w_{4,4}
        \end{array}\right.\\
        GM_2 & \left\{
        \begin{array}{rllll}
          m_{5,4}: & w_{5,4} & w_{1,4} \\
          m_{5,3}: & w_{5,3} & w_{3,3} \\
          m_{5,2}: & w_{5,2} & w_{4,2} \\
          m_{5,1}: & w_{5,1} & w_{2,1} \\
        \end{array}\right.
      \end{split}
      \]
    \end{minipage}
    \hfill
    \begin{minipage}{0.45\textwidth}
      \[
      \begin{split}
        GW_1 & \left\{
        \begin{array}{rllll}
          w_{4,4}: & m_{1,4} & m_{4,2} & m_{4,3} & m_{4,4}\\
          w_{5,4}: & m_{4,4} & m_{5,4} \\
          w_{4,3}: & m_{3,3} & m_{4,3} \\
          w_{5,3}: & m_{4,3} & m_{5,3} \\
          w_{4,2}: & m_{2,2} & m_{4,2} \\
          w_{5,2}: & m_{4,2} & m_{5,2} \\
          w_{3,1}: & m_{2,1} & m_{3,1} \\
          w_{5,1}: & m_{3,1} & m_{5,1} \\
        \end{array}\right.\\
        GW_2 & \left\{
        \begin{array}{rllll}
          w_{1,4}: & m_{5,4} & m_{1,4}\\
          w_{3,3}: & m_{5,3} & m_{3,1} & m_{3,3} \\
          w_{2,2}: & m_{5,2} & m_{2,1} & m_{2,2} \\
          w_{2,1}: & m_{5,1} & m_{2,1} \\
        \end{array}\right.
      \end{split}
      \]
    \end{minipage}
    \[
    \begin{array}{rllllllllllll}
      LM_1: & w_{1,4} & w_{2,1} & w_{2,2} & w_{3,1} & w_{3,3} & w_{4, 2} & w_{4,3} & w_{4,4} & w_{5,1} & w_{5,2} & w_{5,3} & w_{5,4}\\
      LM_2: & w_{5,1} & w_{5,2} & w_{5,3} & w_{5,4} & w_{1,4} & w_{2,1} & w_{2,2} & w_{3,1} & w_{3,3} & w_{4, 2} & w_{4,3} & w_{4,4} \\
      \hline
      LW_1: & m_{1,4} & m_{2,1} & m_{2,2} & m_{3,1} & m_{3,3} & m_{4, 2} & m_{4,3} & m_{4,4} & m_{5,1} & m_{5,2} & m_{5,3} & m_{5,4}\\
      LW_2: & m_{5,1} & m_{5,2} & m_{5,3} & m_{5,4} & m_{1,4} & m_{2,1} & m_{2,2} & m_{3,1} & m_{3,3} & m_{4, 2} & m_{4,3} & m_{4,4} \\
    \end{array}
    \]
  \end{centering}
  \caption{Illustration of the instance created by $\ConstructInstance$ for the edge-colored Hasse diagram depicted on the left. The vertices labeled $4, 3$, and $2$ correspond to the colors blue, red, and green, respectively. Orange is used to depict the ``extra'' color $5$. The white nodes in $\rho_4$ correspond to the color $1$ being added in Lines~\ref{ln:pad-colors-start}--\ref{ln:pad-colors-end} of $\ConstructInstance$. The incomplete preference lists of the agents are indicated below the figures. The ``master" preference lists for the men and women are indicated on the bottom.}
  %  The horizontal lines indicate the separation between groups with the same master preference lists. The master preference lists are indicated on the bottom.}
  \label{fig:list-instance}
\end{figure}

We now analyze the preference lists of the agents.  

%Let us call this particular implementation of the algorithm $\ListConstructInstance$.  

\begin{lem}
\label{lem:list-men-pref}
For each man $m_{c,v}$, $c \in C_v$ and $v \in [p]$, the preference list of $m_{c,v}$  is 
\[
m_{c,v}: \; w_{c, v} \; \; w_{c^+, v} \qquad \mbox{ or }  \qquad m_{c,v}: \; w_{c, v} \; \; w_{c, c} \;\; w_{c^+, v}.
\]
%where, as before, $i^+ = \pi_v(i+1)$. 
Furthermore, if $c \neq p+1$, the subscripts of the women in his list are lexicographically increasing.  If $c = p+1$, then $m_{c,v}$'s preference list has length two, and the subscripts of the women in his list are lexicographically decreasing
\end{lem}
\begin{proof}
By Proposition~\ref{prop:in-n-out-coloring}, $m_{c,v}$'s preference list has length at most $3$ since $\phi$ is a proper in-coloring.   If he has a woman on his list different from his man-optimal and woman-optimal stable partners, then there is an edge $(u,v)$ colored $c$ so that $w_{c,u}$ is his second choice.  But by the way $\phi$ is defined $c = u$. Thus,  the preference list of $m_{c,v}$ is of the form described in the proposition. 
%she will be of the form $w_{c,u}$ with $u = i$ because of how $\phi$ is defined.   
%$$ m_{c,v}: \; w_{c, v} \; \; w_{c^+, v} \qquad \mbox{ or }  \qquad m_{c,v}: \; w_{c, v} \; \; w_{c, i} \;\; w_{c^+, v}.$$

By our choice of $\pi_v$,  if $c \neq p+1$, then $c < c^+$.  Furthermore, if $w_{c,c}$ is $m_{c,v}$'s second choice, then  $v < c$ because $(u,v) = (c,v)$ is an edge of the graph.  Thus,  when $c \neq p+1$,  the subscripts of the women in $m_{c,v}$'s list are lexicographically increasing.   On the other hand,  $m_{p+1,v}$'s preference list has length $2$ since there is no edge assigned the color $p+1$.  Furthermore, $c^+ < p+1$ so the subscripts of the women in $m_{p+1, v}$'s list are lexicographically decreasing.
\end{proof}

\begin{lem}
\label{lem:list-women-pref}
For each woman $w_{c,v}$, $c \in C_v$ and $v \in [p]$, the preference list of $w_{c,v}$  is 
\[
w_{c,v}: \; m_{c^-, v} \; \;  m_{c,v} \;\; \mbox{ when $c \neq v$} \qquad \mbox{and } \qquad w_{c,c}: \; m_{c^-, c} \; \;  (\set{m_{c,x} \sucht (c,x) \in E})  \;\; m_{c,c} \;\; \mbox{ otherwise.}
\]
Here, $(\set{m_{c,x} \sucht (c,x) \in E})$ orders the men so that their subscripts are lexicographically increasing. %and $i^- = \pi_v(i-1)$. 
 If the subscripts of the men in $w_{c,v}$'s list are not lexicographically increasing then $c = \min C_v$ so that $c^- = p+1$. 
\end{lem}
\begin{proof}
From Lemma~\ref{lem:list-men-pref}, we know that only women of the form $w_{c,c}$ can be part of a man's preference list where she is neither his man-optimal nor woman-optimal stable partner.  Thus, for $w_{c,v}$ with $c \neq v$, her preference list consists of her woman-optimal stable partner, $m_{c^-,v}$,  followed by her man-optimal stable partner, $m_{c,v}$.   But for $w_{c,c}$, she can have other men in her preference list and they are of the form $m_{c,x}$ such that $(c,x) \in E$.  Since these men were added to the front of $w_{c,x}$'s list so that their subscripts are lexicographically decreasing, the sublist $(\set{m_{c,x} \sucht (c,x) \in E})$ has subscripts that are  lexicographically increasing.   Furthermore, the subscript of $m_{c,c}$ is lexicographically larger than all the men because if $(c,x) \in E$ then $c > x$.  

The last observation follows from the fact that  if $c = \min C_v$, then $c^- = p+1$ but when $c \neq \min C_v$, then $c^- < c$ because of how we chose $\pi_v$. 
\end{proof}

%\begin{lem}
%\label{lem:two-groups}
%Let $M_{1} = \set{m_{p+1, v} \sucht v \in [p]}$, and $M_0 = M \setminus M_1$; let $W_{1} = \set{w_{c, v} \sucht v \in V, c \text{ is min index in } C_v}$, and $W_0 = W_1 \setminus W_1$.  For $i = 0, 1$, the men in $M_i$ are sublists of a complete list of women; similarly, the women in $W_i$ are sublists of a complete list of men.
%\end{lem}

%\begin{proof}
% Set $LM_0$ to be complete list of women $W = \set{w_{c,v} \sucht c \in C_v, v \in [p]}$ in which their subscripts appear in lexicographically increasing order, and set $LM_1$ to be the complete list of the $W$ in which women of the form $w_{p+1, v}$ appear first (in lexicographical order), followed by all remaining women in lexicgraphic order (see Figure~\ref{fig:list-instance}). Let $LW_0$ and $LW_1$ be the lists of men constructed analogously to $LM_0$ and $LM_1$, respectively.
 
% By Lemma~\ref{lem:list-men-pref}, for any man $m \in M_0$ his preference list is a sublist of $LM_0$, and the preference list of every man $m' \in M_1$ is a sub-list of $LM_1$. Similarly, by Lemma~\ref{lem:list-women-pref}, every woman $w \in W_0$ and  every $w' \in W_1$ have preferences that are sublists of $LW_0$ and $LW_1$, respectively. Therefore, the instance $I(\calP)$ is an SM instance with $(2,2)$-$\MLI$ preferences, as desired.
%\end{proof}

\begin{lthm}
 \label{thm:k-list}
 Let $\calP$ be a finite poset.  There is an SM instance $I(\calP)$   that realizes $\calP$ such that every man's preference list is a sublist of two complete lists of women and every women's preference list is a sublist of two complete lists of men.   Moreover, given the Hasse diagram $H = H(\calP)$ with $p$ vertices and $q$ edges, $I(\calP)$ has $O(p+q)$ agents and can be constructed in $O(p+q)$ time.
 %   a $(2,\infty)$-list instance $I(\calP)$   that realizes $\calP$. Moreover, given the Hasse diagram $H = H(\calP)$ with $p$ vertices and $q$ edges, $I(\calP)$ has $O(p+q)$ agents and can be constructed in $O(p+q)$ time.
\end{lthm} 
 \begin{proof}
   Let $I(\calP)$ be the SM instance constructed by the specific implementation of the algorithm $\ConstructInstance(H, \phi)$  described above.  Let the first {\it group of men} be $GM_1 = \set{m_{c, v} \sucht c \neq p+1, v \in [p]}$ and the second group be $GM_2 = M \setminus GM_1$.  Similarly, denote the first {\it group of women} as $GW_1 = \set{w_{c, v} \sucht c \neq \text{ min color in } C_v, v \in [p]}$ and the second group as $GW_2 = W \setminus GW_1$.
   
    Set the first {\it list for men} $LM_1$ as the complete list of the women $W$ in which the subscripts of the women are lexicographically increasing.  Set the second list $LM_2$ as  the complete list of the women in which the women in $\set{w_{p+1,v} \sucht v \in [p]}$ appear first followed by the remaining women.  For each subgroup, the subscripts of the women are again lexicographically increasing.   (See Figure~\ref{fig:list-instance}).   Let the {\it lists for women} $LW_1$ and $LW_2$ be constructed  analogously as $LM_1$ and $LM_2$ respectively.
    
%    women of the form $w_{p+1, v}$ appear first (subscripts sorted in increasing lexicographical order), followed by all remaining women in lexicgraphic order (see Figure~\ref{fig:list-instance}).

   %Let $M_{1} = \set{m_{p+1, v} \sucht v \in [p]}$, amd $M_0 = M \setminus M_1$. Let $W_{1} = \set{w_{c, v} \sucht v \in V, i \text{is min index in } C_v}$, and $W_0 = W_1 \setminus W_1$.

 %  Set $LM_0$ to be complete list of the women $W = \set{w_{c, v}}$ in which the women appear in lexicographically increasing order, and set $LM_1$ to be the complete list of the women in which women of the form $w_{p+1, v}$ appear first (in lexicographical order), followed by all remaining women in lexicgraphic order (see Figure~\ref{fig:list-instance}). Let $LW_0$ and $LW_1$ be the lists of men constructed analogously to $LM_0$ and $LM_1$, respectively.

   By Lemma~\ref{lem:list-men-pref},  the preference list of every man in $GM_1$ is a sublist of $LM_1$ while those in $GM_2$ is a sublist of $LM_2$.   Similarly, by Lemma~\ref{lem:list-women-pref}, every woman in $GW_1$ and $GW_2$ have preference lists that are sublists of $LW_1$ and $LW_2$ respectively.
   
 %  $w \in LW_1$ and  every $w' \in W_1$ have preferences that are sublists of $LW_0$ and $LW_1$, respectively. 
   
 %  his preference list is a sublist of $LM_1$, and the preference list of every man $m' \in M_1$ is a sub-list of $LM_1$. Similarly, by Lemma~\ref{lem:list-women-pref}, every woman $w \in W_0$ and  every $w' \in W_1$ have preferences that are sublists of $LW_0$ and $LW_1$, respectively. Therefore, the instance $I(\calP)$ is an SM instance with $(2,2)$-$\MLI$ preferences, as desired.
   
 Since $I(\calP)$ is the output of $\ConstructInstance(H, \phi)$, it realizes $\calP$.  We also noted that adding $p+1$ to each $C_v$ increases the number of agents by $2p$ so $I(\calP)$ still has $O(p+q)$ agents.  The only detail we have to verify is the time it takes to implement the specifications on top of the $O(p+q)$ running time of $\ConstructInstance$.  Topologically sorting $H$ takes $O(p+q)$ time. Properly coloring each edge $(u,v)$ of $H$ so that $\phi((u,v)) = u$  takes $O(q)$ time.  Adding $p+1$ to each $C_v$ takes $O(p)$ time.   Processing the edges  $(u,v)$ in $E$ in Lines \ref{ln:edge-loop-start} to \ref{ln:edge-loop-end} in lexicographically decreasing order can be done by radix sort.  There are $q$ edges and $u, v \in [p]$ so the radix sort can be performed in $O(p + q)$ time.  Finally, we can simultaneously sort all $C_v$'s by creating pairs $(v, c)$ for each $c \in C_v$ and sorting them lexicographically using radix sort.    We note that $|C_v| \leq  \deg(v) + 2$ 
  %\deg^-(v) + \deg^+(v) + 2$ (where $\deg^-(v)$ and $\deg^+(v)$ denote the in-degree and out-degree of $v$, respectively) 
  so  $\sum_v |C_v| \leq 2q + 2p$.  It follows that there are $O(p+q)$ pairs to sort.  Since $v \in [p]$ and $c \in [p+1]$, radix sort will again take $O(p+q)$ time.   Thus, even with the five extra specifications,  $I(\calP)$ can be constructed in $O(p+q)$ time. 
 \end{proof}
  
  Let us now consider the  SM instance $I_L(\calP)$ obtained from the SM instance $I(\calP)$ by assigning each man in $GM_1$ the list $LM_1$ and each man in $GM_2$ the list $LM_2$.   For each woman, let her preference list be the one in $I(\calP)$ followed by an arbitrary ordering of the missing men.

%   the $LW_i$ list such that his preference list in $I(\calP)$ is a sublist of $LW_i$. In particular, for each man in $M_0$, let his preference list be $LM_0$; for those in $M_{1}$, let it be $LM_1$. For each woman, let her preference list be the one in $I(\calP)$ followed by an arbitrary ordering of the missing men. 
  
%  Let its agents be exactly the same as $I(\calP)$.  For each man in $M_{p+1}$, assign $LW_2$ as his preference list.  For each man in $M-  M_{p+1}$, let his preference list be $LW_1$.  For each woman, let her preference list be the one in $I(\calP)$ followed by an arbitrary ordering of the missing men. 
  
 % Let us now extend the preference lists of $I(\calP)$ to create an SM instance so that the men have only two kinds of complete preference lists.    Let the men in $M_{p+1}$ have the preference list $LW_2$ and the men in $M-  M_{p+1}$ have the preference list $LW_1$.  

\begin{lem}
\label{lem:same-man-optSM}
The man-optimal stable matching of $I_L(\calP)$ is $\mu^* = \set{(m_{c,v}, w_{c,v}) \sucht c \in C_v, v \in [p]}$.  That is, $I_L(\calP)$ and $I(\calP)$ have the same man-optimal stable matching. 
\end{lem}

\begin{proof}
Let $S = LW_1$, the sequence of men whose subscripts are lexicographically increasing. Let $S_{\ell}$ denote the $\ell\th$ man in the sequence. We run the man-oriented Gale-Shapley algorithm using $S$. In particular, we let each man $S_\ell$ make a series of proposals and get engaged before we move on to $S_{\ell + 1}$.  Of course it's possible that at the later stages of the algorithm $S_\ell$ becomes free because his fiancee rejects him.    We shall show though that this situation never happens.

Consider $S_1 = m_{c,v}$,  the man  with the lexicographically least subscript among all the men.  Since each $v \in [p]$ has $\abs{C_v} \ge 2$, we know $c < p+1$.  Thus, the preference list of $m_{c,v}$ is $LM_1$, and the first woman on $LM_1$ is $w_{c,v}$ because the subscripts of the women in $LM_1$ are also lexicographically increasing. No one has proposed to $w_{c,v}$ yet so  she will accept $m_{c,v}$'s proposal.  Assume that for $\ell' \leq \ell - 1$,  $S_{\ell'}$ is temporarily matched to his partner in $\mu^*$. We will show that the same is true for $S_{\ell}$.  

Suppose $S_{\ell} = m_{c,j}$ such that $c < p+1$.   Again, his preference list is $LM_1$.  He will propose first to women of the form $w_{s,t}$ such that  $(s,t) < (c,j)$.  By assumption, $w_{s,t}$ is matched to $m_{s,t}$.  According to Lemma~\ref{lem:list-women-pref}, the only men $w_{s,t}$ will prefer to $m_{s,t}$  will be $m_{p+1,t}$ or have subscripts lexicographically less than $m_{s,t}$'s.  In other words, $w_{s,t}$ does not prefer $m_{c,j}$ to $m_{s,t}$ so she will reject him.  Hence, $m_{c,j}$ will eventually propose to $w_{c,j}$ who will accept his proposal because she is free.

%belong to $M_{p+1}$ or have subscripts lexicographically less than $m_{c,d}$'s.  In other words, $w_{c,d}$ does not prefer $m_{c,j}$ to $w_{c,d}$ so she will reject him. Hence, $m_{c,j}$ will eventually propose to $w_{c,j}$ who will accept his proposal because she is free.

Suppose $S_{\ell} = m_{p+1, 1}$.  %Notice that  the man prior to $m_{p+1,1}$ in $S_{\ell}$ is of the form $m_{c,v}$ with $c < p+1$.  Furthermore, 
His preference list is $LM_2$ so the first woman on his list is $w_{p+1, 1}$.  He will of course propose to her and she will accept it because no one has proposed to her yet.

%\not \in M_{p+1}$ because  among the men in $M_{p+1}$, his subscript is lexicographically the least.  His preference list is $LM_1$ so the first woman on his list is $w_{p+1, 1}$.  He will of course propose to her and she will accept it because no one has proposed to her yet.  

Finally, let $S_{\ell} = m_{p+1, j}$ such that $j > 1$.  Again, his preference list is $LM_2$. He will first propose to women $w_{p+1, v}$ such that $v < j$.  By assumption, $w_{p+1, v}$ is matched to $m_{p+1,v}$.  According to Lemma~\ref{lem:list-women-pref}, the only man that $w_{p+1,v}$ prefers to $m_{p+1,v}$ is of the form $m_{c,v}$ with $c \neq p+1$.   Thus, $w_{p+1,v}$ will reject $m_{p+1, j}$'s proposal.  He will eventually propose to $w_{p+1, j}$ and she will accept the proposal because she is free. 

By induction, the man-oriented Gale-Shapley algorithm will match all men $m_{c,v}$ to $w_{c,v}$ so $\mu^*$ is the man-optimal stable matching of $I_L(\calP)$. By Lemma~\ref{lem:generic-optimal-matchings}, $\mu^*$ is also the man-optimal stable matching of $I(\calP)$. 
\end{proof}

\begin{lthm}
  \label{thm:ml}
 Let $\calP$ be a finite poset.  There is  an SM instance $I_L(\calP)$ with $(2, \infty)$-list preferences
 %in  $\bigcup_{n \in \Z^+} (2, n)$ 
 that realizes $\calP$. Moreover, given the Hasse diagram $H = H(\calP)$ with $p$ vertices and $q$ edges,  $I_L(\calP)$ has $O(p+q)$ agents and can be constructed in $O((p+q)^2)$ time. 
\end{lthm}

\begin{proof}
Consider $I_L(\calP)$,  which clearly has $(2, \infty)$-list preferences.  It was constructed from $I(\calP)$ which realizes $\calP$. In particular, the women's preference lists in $I_L(\calP)$ is exactly like that in $I(\calP)$ followed by an arbitrary ordering of the missing men.   We will now argue that the two instances have identical shortlists.  If so, by Corollary~\ref{cor:shortlists-2}, they have identical rotation posets and, consequently, $I_L(\calP)$ also realizes $\calP$. 

 We showed in Lemma~\ref{lem:same-man-optSM} that $I_L(\calP)$ and $I(\calP)$ have the same man-optimal stable matching.  They also have the same woman-optimal stable matching.  This is the case because we constructed $I(\calP)$ using $\ConstructInstance$ which lists every woman's woman-optimal stable partner as the first person in her preference list. This property holds in $I_L(\calP)$.  Thus, when we run the woman-oriented Gale-Shapley algorithm, the result will be the same for both instances. 

For ease of discussion, let $S_a$ and $S_a^L$ denote the symmetric shortlists of agent $a$ in $I(\calP)$ and $I_L(\calP)$ respectively.  
It is easy to verify that $S_a$ is exactly the preference list of $a$ in $I(\calP)$.   For a woman $w$, her preference list in $I_L(\calP)$ consists of her preference list in $I(\calP)$, which ends with her man-optimal stable partner, followed by the missing men.  Thus, $S_w^L$ does not have the missing men.  On the other hand, every man's preference list in $I_L(\calP)$ contains his preference list in $I(\calP)$ as a sublist.  So if a man $m$ is part of $S_w$,  he will remain so in $S_w^L$.   Thus, $S_w^L = S_w$.\footnote{Note that we don't have to worry about $S_w^L$ having more men than $S_w$ because {\it all} the men between the woman-optimal and man-optimal stable partners of $w$ are part of $S_w$.}

Next, consider a man $m$.  From our discussion above, we know that $S_m$ is a sublist of $S_m^L$.  If the latter contains a woman $w'$ that is not in $S_m$, then $S_{w'}^L$ should have $m$.  But we proved that $S_{w'}^L = S_{w'}$ so $m$ is in $S_{w'}$, which means $w'$ is in $S_m$.  A contradiction.  Therefore $S_m = S_m^L$.

Constructing $I(\calP)$ takes $O(p+q)$ time. Creating $LM_1$ and $LM_2$ takes $O(p+q)$ time.  But completing the women's preference lists take $O((p+q)^2)$ time so constructing $I_L(\calP)$ takes $O((p+q)^2)$ total time. 
\end{proof}

  Let $I_{L'}(\calP)$ be the SM instance obtained from $I(\calP)$ by assigning each woman in $GW_1$ the list $LW_1$ and each woman in $GW_2$ the list $LW_2$.  Then for each man, let his preference list be the one in $I(\calP)$ followed by an arbitrary ordering of the missing woman.  We leave it up to the reader to verify that $I_{L'}(\calP)$ and $I(\calP)$ will have the same woman-optimal stable matching.   Consequently, one can show that $I_{L'}(\calP)$ is an SM with $(\infty, 2)$-list preferences  that realizes $\calP$.  Like $I_L(\calP)$, it has $O(p+q)$ agents and can be constructed in $O((p+q)^2)$ time.

  %% \textcolor{red}{We can write the argument formally... but it's messier.}

\section{From Path Decompositions to \texorpdfstring{$k$}{k}-Range Preferences}
\label{sec:k-range}

%% Nixing this idea for now: (\textcolor{red}{IDEA:  Should we introduce a step in the algorithm where $C_v$ is initialized to some set?  Most of the time, it's initialized to the empty set.  But in the k-list section $C_v \leftarrow \set{p+1}$; in this section $C_{v} \leftarrow [a_v, b_v+1]$.} )

Recall that an SM instance $I$ is in $\range(k)$ if for all agents $a \in M \cup W$, we have $\max\rank(a) - \min\rank(a) \leq k - 1$ (see Section~\ref{sec:restricted-preference-models}). For rotational simplicity, we will denote $\orank(a) = \min\rank(a)$. Thus $I \in \range(k)$ if and only if for all $m \in M$ and $w \in W$ we have
\begin{equation}
  \label{eqn:objective-rank}
  \begin{split}
    \orank(w) \leq & P_m(w) \leq \orank(w) + k - 1, \quad\text{and}\\
    \orank(m) \leq & P_w(m) \leq \orank(m) + k - 1.
  \end{split}
\end{equation}

In this section, we describe an algorithm for constructing a $k$-range SM instance $I$ realizing a given poset $\calP$. Unlike the previous sections, the $k$-range construction is no longer just dependent only on $\calP$. Rather, the $k$-range construction additionally requires a nice path decomposition $\calX = (X_1, X_2, \ldots, X_{2p})$ of $H(\calP)$ as input, and the $k$ for which $I \in \range(k)$ depends on the width of $\calX$. Specifically, if $\calX$ has width $k$, then the instance $I = I_R(\calP)$ we construct will satisfy $I \in \range(O(k))$. The path decomposition $\calX$ determines the color set $C_v$ for each vertex $v$, and its width affects the number of agents created as well as the range of the agents' rankings. Thus, the smaller $k$ is, the more similar are the preference lists.

  The basic idea of the construction is as follows. Given a poset $\calP$ and path decomposition $\calX = (X_1, X_2, \ldots, X_{2p})$ of $H(\calP)$ of width $k$, we associate a set of $O(\abs{X_i}) = O(k)$ agents with each $X_i$. If an element $\rho \in \calP$ is contained in $X_i, X_2, \ldots, X_j$, then there is a corresponding rotation containing one man-woman pair from each $X_{\ell}$ with $i \leq \ell \leq j$. If $H(\calP)$ contains a directed edge $(\rho, \sigma)$ with $\ell$ the minimum index satisfying $\rho, \sigma \in X_\ell$, then the preferences of the agents corresponding to $\rho$ and $\sigma$ in $X_\ell$ enforce that $\rho$ must be eliminated before $\sigma$ is exposed.

  We first construct an instance, $I_1(\calP)$, realizing $\calP$ with incomplete preferences, such that each agent $a$ associated with a set $X_i$ only has acceptable partners associated with sets $X_j$ satisfying $\abs{i - j} \leq 2$.  We then form an instance $I_2(\calP)$ in which each $a$ appends all other agents associated with $X_{i-2}, X_{i-1}, \ldots, X_{i+2}$ to her preference lists. Finally, we complete each $a$'s preferences by inserting all agents associated with $X_1, \ldots, X_{i-3}$ to the front of her preference list, and appending all agents associated with $X_{i+3}, \ldots, X_{2p}$ to the end of her preference list. Intuitively, the resulting instance $I$ is in $\range(O(k))$ because two men (say) can only disagree on the relative ranks of two women $w$ and $w'$ if $w$ and $w'$ are associated with $X_i$ and $X_j$ ( respectively) with $\abs{i - j} \leq C$ for some constant $C$; for any $i, j$ satisfying $i < j - C$, all men prefer $w$ associated with $X_i$ to $w'$ associated with $W_j$. Thus, all rankings of all agents agree up to $O(C k) = O(k)$---i.e., $I \in \range(O(k))$.

%% We describe the construction in stages. We start by creating an SM instance $I_1(\calP)$ with incomplete preferences that realizes $\calP$.  We then extend the preference lists of the agents to create another SM isntance $I_2(\calP)$ that still realizes $\calP$ with longer---albeit still incomplete---preference lists. Finally, we will complete the preference list of the agents to arrive at $I_R(\calP) \in \range(O(k))$. \wnote{Should we make this more specific/intuitive?}

\subsection{Creating $I_1(\calP)$ and $I_2(\calP)$} 

Let $[a,b] = \{a, a+1, \hdots, b\}$ and $\ell = b- a + 1$.  The {\it bitonic sequence} corresponding to $[a,b]$ is a permutation $(s_1, s_2, \hdots, s_{\ell})$ of its elements such that $|s_i - s_{i+1}| \leq 2$ for $i = 1, \hdots, \ell $ (and $s_{\ell + 1} = s_1$).  That is, any two consecutive numbers in the (circular) sequence differ by at most $2$.   For example,  $(3, 5, 7, 6, 4)$ is a bitonic sequence for $[3,7]$.   
In general, we can construct a bitonic sequence for $[a,b]$ by starting at $a$ and incrementing each number by $2$ until we reach $b-1$ or $b$,   then adding the largest number that is not part of the sequence yet ($b$ or $b-1$),  and decrementing each number by $2$ until we reach $a+1$.   Bitonic sequences will play a role in our implementation of $\ConstructInstance$. 

We will once again make use of $\ConstructInstance$ to create the SM instance with the following specifications:
\begin{itemize}
\item  \emph{For each edge $(u,v)$, let $\phi((u,v)) = \min\set{i \sucht u, v \in X_i}$.}  We know that $\phi((u,v))$ is well-defined because $\calX$ is a path decomposition of $H$.  
\item \emph{For each $v \in [p]$, let $C_v = [a_v, b_v+1]$ where $a_v = \min \set{i \sucht v \in X_i}$ and $b_v = \max \set{i \sucht v \in X_i}$.}   Since there is some $X_i$ such that $v \in X_i$, notice that $|C_{v}| \ge 2$ for every $v \in [p]$.  Moreover,  every color assigned to some edge incident to $v$ is in $C_v$.  Thus, we are effectively skipping Lines 2 to 9 of Algorithm~\ref{alg:generic}.   (Note that for this section, we shall also use $i$ instead of $c$ to refer to the colors in $C_v$ as the colors are indices in the path decomposition $\calX$.)
\item \emph{For each $v \in [p]$, let $\pi_v$ be a bitonic sequence of $C_v$.} 
\end{itemize}
We refer to the instance created by the above specification of $\ConstructInstance(H(\calP), \phi)$ as $I_1(\calP)$. Figure~\ref{fig:k-range-instance} illustrates the construction for a poset with pathwidth 2.

\begin{prop}
  \label{prop:I1}
  The SM instance $I_1(\calP)$ realizes $\calP$.  It has $O(kp)$ agents and it can be constructed in $O(kp + q)$ time, where $k$ is the width of $\calX$---the nice path decomposition of $H(\calP)$---and $p$ and $q$ are the number of vertices and edges of $H(\calP)$ respectively.
\end{prop}
\begin{proof}
  Since $I_1(\calP)$ was constructed by $\ConstructInstance$,  the SM instance realizes $\calP$.    Let us now compute the instance's number of agents.  For each $v \in [p]$ and each  $i \in [a_v, b_v + 1]$, the algorithm creates two agents $m_{i, v}$ and $w_{i,v}$.  By definition, $a_v$ is the index when $v$ is added to the path decomposition $\calX$ while $b_v + 1$ is the index when $v$ is removed from $\calX$.   Now $\calX$ is a {\it nice} path decomposition so  at most one vertex is removed at index $i$, $i \in [2p]$.  It follows that, for a fixed $i$,  the number of women $w_{i,v}$ with $v \in [p]$ is at most $|X_i| + 1 \leq (k+1) + 1 = k+2$ because $\calX$ has width $k$.   The same bound holds for the number of men $m_{i,v}, v\in [p]$.    Thus, the total number of agents is at most $2p \times 2(k+2) = O(kp)$.     
  
  Finally, let us consider the running time of the algorithm with the specifications.  Creating the agents and adding their man-optimal stable partners into their preference lists takes $O(kp)$ time.  Adding their acceptable partners based on each edge of $H(\calP)$ takes $O(q)$ time.  Finally,  for each $v \in [p]$, creating $\rho_v$ based on the bitonic sequence for $[a_v, b_v+1]$ and adding the agents' woman-optimal stable partners into their preference lists  take  $O(|\rho_v|)$ time.   But $O(\sum_v |\rho_v|) = O(kp)$ since each man and each woman is part of exactly one rotation.   Thus, constructing $I_1(\calP)$ takes $O(kp + q)$ time. 
\end{proof}

\begin{prop}
  In the SM instance $I_1(\calP)$ the following are true:
  \begin{itemize}
  \item[(i)] For $m_{i,v}$, $i \in C_v, v \in [p]$,  every woman $w_{j,*}$ that appears in his preference list has $|j - i| \leq 2$.
  \item[(ii)] For  $w_{i,v}$, $i \in C_v, v \in [p]$,  every man $m_{j,*}$ that appears in her preference list has $|j - i| \leq 2$. 
  \end{itemize}  
\end{prop}
\begin{proof}
  The preference list of $m_{i,v}$ starts with $w_{i,v}$ followed by women of the form $w_{i,u}$ such that $(u,v) \in E$ and $\phi((u,v)) = i$ and then ending with $w_{i^+, v}$.   But $|i^+ - i| \leq 2$ because $\pi_v$ is a bitonic sequence of $[a_v, b_v +1]$.  
  
  For $w_{i,v}$, her preference list starts with $m_{i^-,v}$, followed by men of the form $m_{i, y}$ such that $(v, y) \in E$ and $\phi((v,y)) = i$ and then ending with $m_{i,v}$.  Again, $|i^- - i| \leq 2$ because $\pi_v$ is a bitonic sequence of $[a_v, b_v +1]$.
\end{proof}

For $i \in [2p]$, let $M_i = \set{ m_{i,v} \sucht i \in C_v}$.  Similarly, let $W_i = \set{ w_{i,v} \sucht i \in C_v}$.  The above proposition implies that for every man $m \in M_i$,  the women in his preference list is a subset of
\[
B_i = \bigcup \set{ W_j \sucht  j \in [2p],  |j - i| \leq 2 }.
\]
Similarly, for each woman $m \in W_i$, the men in her preference list is a subset of 
\[
D_i = \bigcup \set{ M_j \sucht  j \in [2p],  |j - i| \leq 2 }.
\]
We create $I_2(\calP)$ from $I_1(\calP)$ by appending to each $m_{i,v}$'s  preference list in $I_1(\calP)$ the missing women in $B_i$. Similarly, we append to each $w_{i,v}$'s preference list in $I_1(\calP)$ the missing men in $D_i$.

\begin{prop}
  In $I_2(\calP)$, for $i \in C_v$ and $v \in [p]$, $m_{i,v}$'s preference list contains all women in $B_i$ while $w_{i,v}$'s preference list contains all men in $D_i$.   Moreover, $I_2(\calP)$ realizes $\calP$. 
\end{prop}

\begin{proof}
  The first part of the proposition is true by construction.  Applying the same argument we used to prove Corollary~\ref{cor:shortlists-2}, $I_1(\calP)$ and $I_2(\calP)$ have identical shortlists.  Thus, like $I_1(\calP)$, the instance $I_2(\calP)$ also realizes $\calP$. 
\end{proof}

%\cnote{Need to add example here?} \wnote{Yes, I'll add it on the next pass.}

\begin{figure}
  \begin{minipage}{0.3\textwidth}
    \begin{center}
      \includegraphics{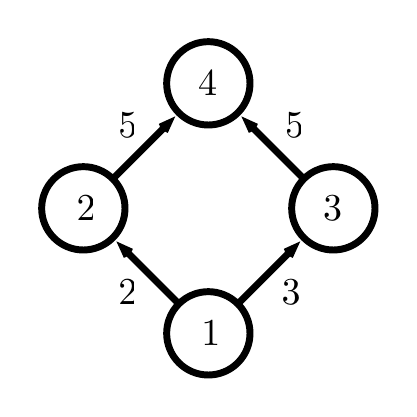}
    \end{center}
  \end{minipage}
  \hfill
  \begin{minipage}{0.65\textwidth}
    \[
    \calX = \set{1}, \set{1,2}, \set{1,2,3}, \set{2, 3}, \set{2, 3, 4}, \set{3, 4}, \set{4}, \varnothing
    \]
    \[
    \begin{array}{cccccccccccccccccc}
      C_4 : & & & & & & & &  & \{ & 5 & & 6 & & 7 & & 8 &\}\\
      C_3 : & &  & &  &  \{ & 3 & & 4 && 5 && 6 && 7 & \} & \\
      C_2 : &&   & \{ & 2 && 3 && 4 && 5 && 6 &\} &  & \\
      C_1 : & \{ & 1 && 2 && 3 & &4 &\}\\
    \end{array}
    \]

    \begin{align*}
      \pi_4 &= (5, 7, 8, 6)\\
      \pi_3 &= (3, 5, 7, 6, 4)\\
      \pi_2 &= (2, 4, 6, 5, 3)\\
      \pi_1 &= (1, 3, 4, 2)
    \end{align*}
  \end{minipage}

  \begin{center}
    \includegraphics{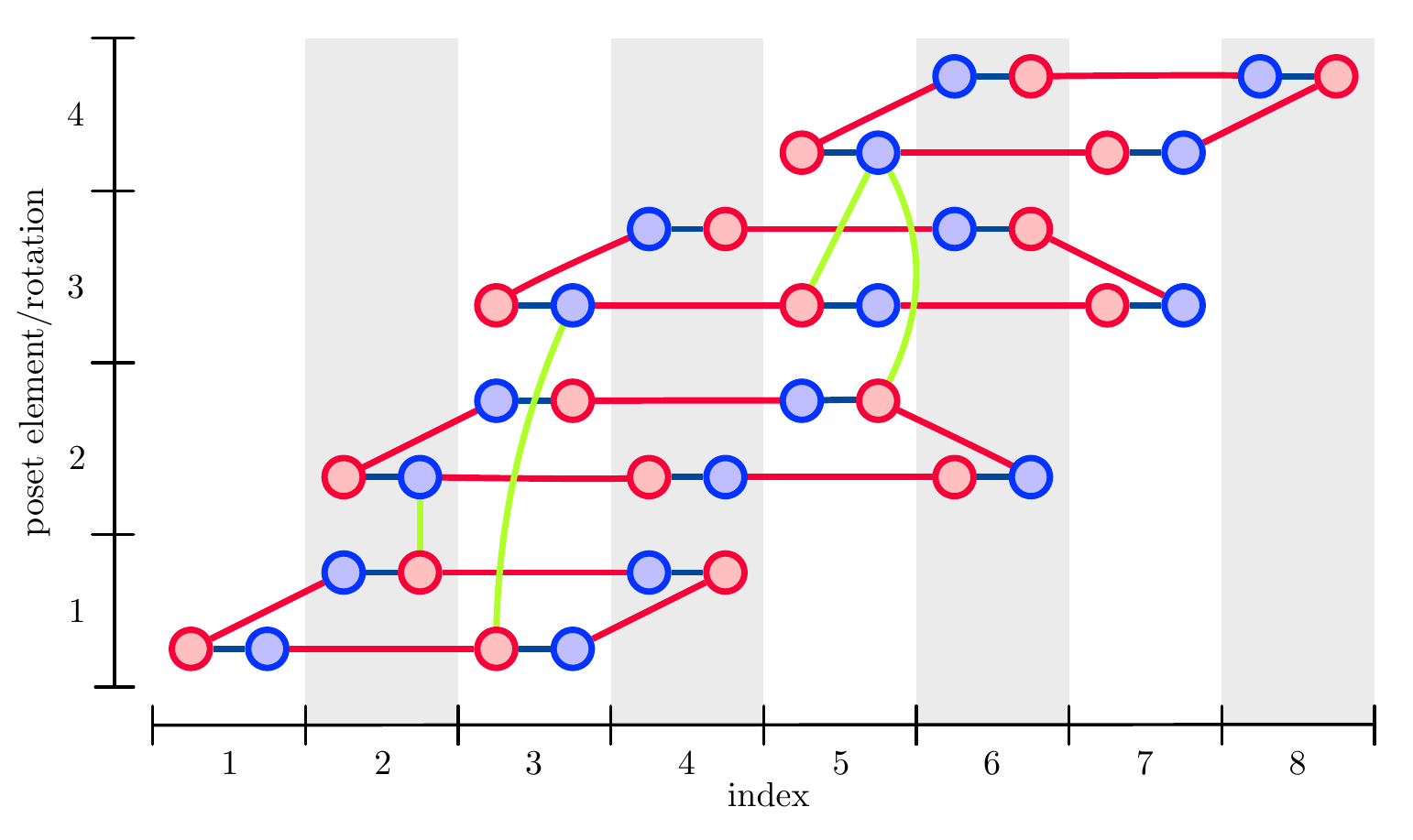}
  \end{center}
  \caption{Example of our $k$-range construction using the poset $\calP$ as our running example. The top diagram shows the Hasse diagram $H(\calP)$. The top right gives a path decomposition of $H(\calP)$ of width $2$. Each edge $(u, v)$ in $H(\calP)$ is labeled with the minimum index $i$ such that $u, v \in X_i$. We also show the bitonic sequence $\pi_v$ for each set $C_v$ of indices. The bottom image shows the incomplete preferences computed by $I_1(\calP)$. As before, men are depicted as blue nodes, while women are depicted as red. The men prefer blue to green to red edges, while the women prefer the edges in the opposite order.}
  \label{fig:k-range-instance}
\end{figure}

\subsection{Creating $I_R(\calP)$} 

We now derive $I_R(\calP)$ from $I_2(\calP)$ as follows:   for each $m_{i,v}$, we create a preference list  that  has the following structure:
\begin{equation*}
  \begin{array}{cccccccccc}
    m_{i,v} & : & W_1 & W_2 & \cdots & W_{i - 3} & \fbox{$m_{i,v}$'s preference list in $I_2(\calP)$}
    & W_{i+3} & \cdots & W_{2p}. 
  \end{array}
\end{equation*}
For each $j \in [1, i-3] \cup [i+3, 2p]$, the women in $W_j$ are arranged so that their subscripts are lexicographically increasing.  Similarly, for each woman $w_{i,v}$, her preference list has the following structure: 
\begin{equation*}
  \begin{array}{cccccccccc}
    w_{i,v} & : & M_1 & M_2 & \cdots & M_{i - 3} & \fbox{$w_{i,v}$'s preference list in $I_2(\calP)$} & M_{i+3} & \cdots & M_{2p}. 
  \end{array}
\end{equation*}
Again, the men in $M_j$, $j \in [1, i-3] \cup [i+3, 2p]$ are arranged so that their subscripts are lexicographically increasing.  Hence, agents from each group will have similar, but nonetheless distinct complete preference lists.

\begin{lem}
  \label{lem:block-critical}
  $I_R(\calP) \in \range(9 (k + 2))$.
\end{lem}
\begin{proof}
  Fix $w_{i,v} \in W_i$.   Recall that $P_m(w_{i,v})$ is the {\it rank} assigned by $m$ to $w_{i,v}$.  In what follows, we analyze $\min_{m \in M} P_m(w_{i,v})$ and $\max_{m \in M} P_m(w_{i,v})$ as their difference will determine the range of $I_R(\calP)$.  By convention, the smaller the value of $ P_m(w_{i,v})$, the more desirable she is to $m$.  Thus, when $P_m(w_{i,v})$ is small, we say that $m$ ranked her {\it high}; conversely, when $P_m(w_{i,v})$ is large, $m$ ranked her {\it low}.  
  
  The smallest index $j$ so that $w_{i,v} \in B_j$ is $i-2$.  Thus, for all men in $M_1 \cup M_2 \cup \hdots, M_{i-3}$,  the rank of $w_{i,v}$ in their preference lists is the same and  is equal to $|W_1| + |W_2| + \hdots + |W_{i-1}| + |W'_{i}| + 1$ where $W'_i$ contains the women in $W_i$ with subscripts lexicographically less than $w_{i,v}$.  
  
  However,  a man in $M_{i-2}$ can rank $w_{i,v}$ as high as $|W_1| + |W_2| + \hdots + |W_{i-5}| +1$ because $w_{i,v} \in B_{i-2}$.  It is easy to see that none of the men in $M_{i-1} \cup M_i \cup \hdots \cup M_{2p}$ can provide a higher rank for $w_{i,v}$.  
  
  Using the same analysis, a man in $M_{i+2}$ can rank $w_{i,v}$ as low as $|W_1| + |W_2| + \hdots + |W_{i+4}|$ and no man can give her a worse ranking.  Thus, in $I_R(\calP)$ the difference between the worst possible and the best possible rank of $w_{i,v}$ is 
  \begin{eqnarray*}
    |W_{i-4}| + |W_{i-3}| + \hdots + |W_{i+4}| - 1 &\leq & 9(k+2)
  \end{eqnarray*}
  because, as we already noted in the proof of Proposition~\ref{prop:I1}, each $|W_i| \leq k+2$.  By setting $\orank(w_{i,v})$ to the best rank it received from a man,  we have now shown that the range of the rankings of each woman is at most $9(k+2)$.  A similar analysis holds for the range of the rankings of each man in the women's preference lists. 
\end{proof}

\begin{lem}
  \label{lem:same-man-optSM-range}
  The man-optimal stable matching of $I_R(\calP)$  is $\mu^* = \set{ (m_{i,v}, w_{i,v}), i \in C_v, v\in [p]}$.  That is, $I_R(\calP)$ and $I_1(\calP)$ have the same man-optimal stable matching.
\end{lem}

\begin{proof}
  We shall prove the above lemma like Lemma~\ref{lem:same-man-optSM}.  Let $S$ be the sequence of men whose subscripts are lexicographically increasing.  We then run the male-oriented Gale-Shapley algorithm on $S$.  For our purposes, it is  useful to think of $S$ as consisting of men from $M_1$, followed by men from $M_2$, etc. and ending with men from $M_{2p}$.   By doing induction on $i \in [2p]$, we will prove that each man  $m_{i,v} \in M_i$ has $w_{i,v}$ as his man-optimal stable partner. 

  Our basis step involves $M_1 \cup M_2 \cup M_3$.  Let $m_{i,v}$ be a man from one of the sets; i.e., $i = 1, 2$ or $3$.  Notice that his preference list in $I_R(\calP)$ begins with his preference list from $I_2(\calP)$  followed by women from $W_{i+3} \cup W_{i+4} \cup \hdots \cup W_{2p}$. But his preference list from $I_2(\calP)$ begins with his preference list from $I_1(\calP)$ followed by some arbitrary ordering of the missing women from $B_i$.  Finally, $I_1(\calP)$ was constructed by $\ConstructInstance$ so his preference list begins with $w_{i,v}$.  Thus, $m_{i,v}$ will propose to $w_{i,v}$ first and she will accept because no one has proposed to her yet.  

  Assume that for $\ell' \leq \ell-1$ all men $m_{\ell', v} \in M_{\ell'}$ is temporarily matched to $w_{\ell', v}$.  Let us now consider a man $m_{\ell,v} \in M_{\ell}$.  Recall that his preference list in $I_R(\calP)$ has the following structure:
  \begin{equation*}
    \begin{array}{ccccccccc}
      m_{\ell,v} & : & W_1 &  \cdots & W_{\ell - 3} & \fbox{\fbox{$m_{\ell,v}$'s  list in $I_1(\calP)$} (missing women in $B_{\ell}$)}     & W_{\ell+3} & \cdots & W_{2p}. 
    \end{array}
  \end{equation*}

  When he proposes to some women $w_{i,v} \in W_1 \cup W_2 \cup \hdots, W_{\ell-3}$, she is already matched to $m_{i,v}$ by assumption.  But in $w_{i,v}$'s preference list in $I_R(\calP)$,  all the men in $M_{\ell}$ are ranked lower than $m_{i,v}$, which is the last man in her preference list in $I_1(\calP)$,  so she will reject $m_{\ell,v}$. If $w_{i,v} \in W_{\ell-3}$, for example, her preference list looks like the one below: 
  \begin{equation*}
    \begin{array}{ccccccccc}
      w_{i,v} & : & M_1 &  \cdots & M_{\ell - 6} & \fbox{\fbox{$w_{i,v}$'s  list in $I_1(\calP)$} (missing men in $D_{\ell-3}$)}     & M_{\ell} & \cdots & M_{2p}. 
    \end{array}
  \end{equation*}
  Thus, $m_{\ell,v}$ will propose to the first woman in his preference list in $I_1(\calP)$, which is $w_{\ell,v}$.  She is currently free and will accept his proposal.  So $m_{\ell,v}$ is temporarily matched to $w_{\ell,v}$.  We emphasize that this is true for {\it every} man in $M_{\ell}$. 

  By induction,  we have shown that $\mu^* = \set{ (m_{i,v}, w_{i,v}), i \in C_v, v\in [p]}$ is the man-optimal stable matching of $I_R(\calP)$. 
\end{proof}

\begin{lem}
  \label{lem:same-woman-optSM-range}
  $I_R(\calP)$ and $I_1(\calP)$ have the same woman-optimal stable matching.
\end{lem}

\begin{proof}
  The women's preference lists in $I_R(\calP)$ were constructed from $I_1(\calP)$ like the men's preference lists.  In particular, every woman $w_{i,v}$'s most preferred partner in $I_1(\calP)$ is her woman-optimal stable partner $m_{i^-,v}$ (just like $m_{i,v}$'s most preferred partner in $I_1(\calP)$ is his man-optimal stable partner $w_{i,v}$.)   Let $S'$ be the sequence of women whose subscripts are lexicographically increasing.    Applying the same analysis as the proof of Lemma~\ref{lem:same-man-optSM-range}, we can show that running the the woman-oriented Gale-Shapley algorithm on $S'$ will result in a stable matching that matches each woman $w_{i,v}$ to $m_{i^-,v}$. 
\end{proof}

\begin{lthm}
  \label{thm:k-range-construct}
  Let $H(\calP)$ be the Hasse diagram of a finite poset $\calP$ with $p$ vertices and $q$ edges.  Let  $\calX = (X_1, X_2, \hdots, X_{2p})$ be a nice path decomposition of $H(\calP)$ whose width is $k$.  Then there exists an SM instance $I_R(\calP) \in \range(O(k))$ that realizes $\calP$. $I_R(\calP)$ has $O(kp)$ agents and can be constructed in $O(k^2p^2)$ time. 
\end{lthm}
\begin{proof}
  In Lemma~\ref{lem:block-critical}, we established that $I_R(\calP) \in \range(9(k+2))$.   According to Lemmas~\ref{lem:same-man-optSM-range} and~\ref{lem:same-woman-optSM-range}, the SM instances $I_R(\calP)$ and $I_1(\calP)$ have the same man-optimal and woman-optimal stable matchings.   By construction, every agent $a$ has the same set of agents between their man-optimal and woman-optimal stable partners in $I_R(\calP)$ and $I_1(\calP)$.  It follows that the two instances have identical symmetric shortlists.  By Corollary~\ref{cor:shortlists}, they also have identical rotation posets.  From Proposition \ref{prop:I1}, $I_1(\calP)$ realizes $\calP$ and has $O(kp)$ agents  so $I_R(\calP)$ does too.   Finally, constructing $I_1(\calP)$ takes $O(kp + q)$ time but assigning each agent a complete preference list in $I_R(\calP)$ takes $O((kp)^2)$, which is the bottleneck.  Thus, the running time of $O(k^2p^2)$ follows.
\end{proof}

\begin{cor}
  Let $\calP$ be a finite poset on $p$ elements with pathwidth $\pw$. Then there is an SM instance $I$ with $I \in \range(O(\pw))$ that realizes $\calP$.  $I$ can be constructed in time $O(f(\pw) p^2 + \pw^2 p^2)$ for some function $f$ depending only on the pathwidth of $\calP$.
\end{cor}
\begin{proof}
  By Corollary~\ref{cor:nice-path}, a nice path decomposition $\calX$ of $H(\calP)$ can be computed in time $O(f(\pw) p^2)$. The desired result then follows from Theorem~\ref{thm:k-range-construct}.
\end{proof}

\section{From \texorpdfstring{$k$}{k}-Range Preferences to Path Decompositions}
\label{sec:k-range-algo}

%% \cnote{Time-permitting,  I suggest creating an example for this section too.} \wnote{I agree! I'll do it when I get a chance.}

In the previous section, we showed how starting with a finite poset $\calP$ with pathwidth $k$, we can construct an SM instance $I_R(\calP) \in \range(9k+2)$.   In this section, we show that the connection between $\calP$'s pathwidth and the range of the SM instance created is not incidental:  if $I \in \range(k)$ then $\calR(I)$ has a path decomposition of width $O(k^2)$ that can be constructed efficiently.  

Throughout the section, we assume the SM instances have complete preference lists. We begin by noting that \emph{every} SM instance  belongs to $\range(k, n)$ for some $k \in [1, n-1]$. Let $\range(I)$ denote the smallest such $k$ for which $I \in \range(k)$.  The parameter $\range(I)$ can be thought of as a similarity measure of the preference lists of the agents. The smaller $\range(I)$ is, the more alike the agents' preferences lists are to each other. As remarked in Section~\ref{sec:restricted-preference-models}, computing $\range(I)$ can be done efficiently. The following lemma is a formal statement of the discussion in Section~\ref{sec:restricted-preference-models}.

\begin{lem}
  \label{lem:compute-range}
  Given an SM instance $I$ of size $n$,  computing $k = \range(I)$ and the minrank function $\orank$ can be done in $O(n^2)$ time.
\end{lem}
%% \begin{proof}
%% For $w \in W$,  let $\orank(w) = \min_{m \in M} P_m(w)$ and  $\range(w) = \max_{m \in M} \set{P_m(w) - \orank(w)}$. Similarly, for  $m \in M$, let $\orank(m) = \min_{w \in W} P_w(m)$ and $\range(m) = \max_{w \in W} \set{P_w(m) - \orank(m)}$. From~(\ref{eqn:objective-rank}), $I \in \range(k,n)$ if and only if $k \ge \range(w)$ for each $w \in W$ and $k \ge \range(m)$ for each $m \in M$.   Thus, we set $k = \max_{a \in M\cup W} \range(a)$.

%% The objective rank function $\orank$ and the range $k$ can be obtained by going through the preference lists of $I$ which takes $O(n^2)$ time. 
%% \end{proof}

%% \subsection{Bounding the Pathwidth of the Rotation Poset}
%% \label{sec:bounding-pathwidth}

Let $I \in \range(k)$.  Our goal in this section is to create a path decomposition for $G(I)$, the rotation digraph of $I$, whose width is $O(k^2)$. Since the Hasse diagram of $\calR(I)$ is a subgraph of $G(I)$, it follows that the pathwidth of $\calR(I)$ is also $O(k^2)$. Towards bounding the pathwidth, we will assign each rotation $\rho \in \calR(I)$ an interval that is based on the minranks of the agents in the rotation and then show that the maximum number of intervals that contain a particular $i \in [1,n]$, where $n$ is the size of $I$, is $O(k^2)$.

We now present a series of structural results leading to our main result, Theorem~\ref{thm:k-range-upper}. To orient the reader, we give a brief overview of the results. Our first structural result, Proposition~\ref{prop:nonempty-support}, bounds the number of agents with minranks at most $i$ for instances $I \in \range(k)$. Using this result, we show that for all stable pairs $(m, w)$, the minranks of $m$ and $w$ cannot differ too much (Lemma~\ref{lem:k-range-match}).\footnote{Lemma~\ref{lem:k-range-match} generalizes the result that instances with master preferences (i.e., with $k = 1$) have a unique stable matching---the matching in which stable partners both have the same (min)rank.} As a consequence of Lemma~\ref{lem:k-range-match}, we argue that (1) adjacent men (say) $m_i, m_{i+1}$ in any rotation $\rho$ must have similar minranks (Corollary~\ref{cor:rotation-density}), and (2) no agent can have many stable partners. Hence each agent cannot be part many rotations (Corollary~\ref{cor:few-rotations}).

Combining these structural results, Theorem~\ref{thm:k-range-upper} follows by showing that not too many ($O(k^2)$) rotations can ``overlap'' in the sense of containing agents with similar minranks (Cf. Definition~\ref{dfn:extent} and Lemmas~\ref{lem:extent} and~ \ref{lem:limited-overlap}). In particular, the structural results imply that a simple and efficient procedure yields a path decomposition of $G(I)$ of width $O(k^2)$---a construction we leverage for our algorithmic results in the following section.

\begin{prop}
  \label{prop:nonempty-support}
 Let $I \in \range(k, n)$ with minrank function $\orank$.  Then  for all $i \in [n]$
 \begin{equation}
    \label{eqn:nonempty-support}
    \begin{split}
    i \leq  \abs{\set{w \in W \sucht \orank(w) \leq i}} &  \leq i + k - 1,  \quad\text{and}\\
     i \leq  \abs{\set{m \in M \sucht \orank(m) \leq i}} & \leq i + k - 1.
     \end{split}
  \end{equation}
\end{prop}
\begin{proof}
  We will show that the bounds for $\abs{\set{w \in W \sucht \orank(w) \leq i}}$ are correct.   The second set of inequalities follow by interchanging the roles of the men and women.

  Recall that every man $m$'s ranking function $P_m: W \rightarrow [n]$ is a bijection. Now, consider a woman $w$ with $\orank(w) \leq i$.  According to (\ref{eqn:objective-rank}), $P_m(w)  \leq i + k - 1$.  If $ \abs{\set{w \in W \sucht \orank(w) \leq i}}  > i + k-1$, $P_m$ will have to assign at least two women with $\orank(w) \leq i$ the same rank, contradicting the fact that $P_m$ is a bijection.  
 
  On the other hand,  suppose $ \abs{\set{w \in W \sucht \orank(w) \leq i}}  < i$ so that $ \abs{\set{w \in W \sucht \orank(w) > i}}  \ge n - i +1$.    Every woman $w$ with $\orank(w) > i$ will have $P_m(w) > i$ according to  (\ref{eqn:objective-rank}).   Thus, $P_m$ will have to assign at least $n-i+1$ women the ranks $i+1, \hdots, n$, which again contradicts the fact that $P_m$ is a bijection.  It follows that the  bounds for  $\abs{\set{w \in W \sucht \orank(w) \leq i}}$ hold.
\end{proof}

\begin{lem}
  \label{lem:k-range-match}  Let $I \in \range(k,n)$ with minrank function $\orank$. Then for any stable pair $(m,w)$ of $I$, 
  \begin{equation}
    \label{eqn:k-range-match}
    \abs{\orank(m) - \orank(w)} \leq 2 k - 2.
  \end{equation}
\end{lem}
\begin{proof}
  We will show that
  \begin{equation}
    \label{eqn:rel-rank}
    \orank(w) \leq \orank(m) + 2 k - 2.    
  \end{equation}
  The proof of the other inequality is identical, interchanging the roles of $m$ and $w$. 
  
  Let $\mu$ be a stable matching of $I$ such that $(m,w) \in \mu$.  To prove (\ref{eqn:rel-rank}),  we shall bound $P_m(w)$ below using $\orank(w)$ and above using $\orank(m)$.  The former is straightforward---by (\ref{eqn:objective-rank}),  $\orank(w) \leq P_m(w)$.   For the latter, we consider the number of women $w'$ that $m$ prefers to $w$, which we know is $P_m(w) - 1$.   Each such woman $w'$ must be matched to a man $m'$ that she prefers to $m$.   Now, by  (\ref{eqn:objective-rank}), $P_{w'}(m) \leq \orank(m) + k-1$ so $\orank(m') \leq P_{w'}(m') < \orank(m) + k-1$.   The number of men that will satisfy this property is bounded by 
  
   \begin{equation}
    \label{eqn:subj-rank}
    \abs{\set{m' \sucht \orank(m') \leq \orank(m) + k - 2}}  \leq \orank(m) + 2 k - 3.  
  \end{equation}
It follows that $m$ prefers at most $\orank(m) + 2k-3$ women to $w$ so 
   \[
  \orank(w)  \leq P_m(w)  \leq \orank(m) + 2 k - 2,
  \]
   whence~(\ref{eqn:rel-rank}) follows.
    \end{proof}

Proposition~\ref{prop:nonempty-support} and Lemma~\ref{lem:k-range-match} have the following consequences, which will be useful in our description of rotation posets arising from SM instances with $k$-range preferences.

\begin{cor}
  \label{cor:rotation-density}
  Let $I \in \range(k,n)$ with minrank function $\orank$. Suppose
  \[
  \rho = (m_0, w_0), (m_1, w_1), \ldots, (m_{\ell - 1}, w_{\ell-1})
  \]
  is a rotation of $I$.  Then for all $i = 0, 1, \ldots, \ell - 1$, we have
  \begin{equation}
    \begin{split}
      \abs{\orank(m_i) - \orank(m_{i+1})} &\leq 4 k - 4, \quad\text{and}\\
      \abs{\orank(w_i) - \orank(w_{i+1})} &\leq 4 k - 4.
    \end{split}
  \end{equation}
\end{cor}
\begin{proof}
Let $\rho$ be exposed in the stable matching $\mu$ of $I$.  Then $(m_{i+1}, w_{i+1}) \in \mu$ while $(m_i, w_{i+1}) \in \mu \setminus \rho$.  
  Applying Lemma~\ref{lem:k-range-match} twice (along with the triangle inequality) gives
  \begin{align*}
    \abs{\orank(m_i) - \orank(m_{i+1})} &= \abs{\orank(m_i) - \orank(w_{i+1}) + \orank(w_{i+1}) - \orank(m_{i+1})}\\
    &\leq \abs{\orank(m_i) - \orank(w_{i+1})} + \abs{\orank(w_{i+1}) - \orank(m_{i+1})}\\
    &\leq 2 (2 k - 2),
  \end{align*}
  as desired.
\end{proof}

\begin{cor}
  \label{cor:few-rotations} Let $I \in \range(k,n)$.  Then each agent in $I$ has at most $5 k - 4$ stable partners. In particular, each agent can appear in at most $5 k - 5$ rotations.
\end{cor}
\begin{proof} Assume the minrank function of $I$ is $\orank$. Let us prove the result for a man $m$.  Let $W_m$ denote his stable partners in $I$. By Lemma~\ref{lem:k-range-match}, all women $w \in W_m$ satisfy
  \begin{equation}
    \label{eqn:match-rank}
    \orank(m) - 2 k + 2 \leq \orank(w) \leq \orank(m) + 2 k - 2.
  \end{equation}
  By Proposition~\ref{prop:nonempty-support}, we have
  \begin{equation*}
    \abs{\set{w \in W \sucht \orank(w) \leq \orank(m) + 2 k - 2}} \leq \orank(m) + 3 k - 3.
  \end{equation*}
  Similarly, Proposition~\ref{prop:nonempty-support} implies that
  \begin{equation*}
    \orank(m) - 2 k + 1 \leq \abs{\set{w \in W \sucht \orank(w) < \orank(m) - 2 k + 2}}.
  \end{equation*}
  Combining the previous two expressions, the number of women $w$ satisfying~(\ref{eqn:match-rank}) is at most
  \[
  (\orank(m) + 3 k - 3) -  (\orank(m) - 2 k + 1) = 5 k - 4.
  \]
  Except for his partner in the women-optimal stable matching, $m$ appears with each of its stable partners in a rotation.  Thus, $m$ can be part of at most $5k-5$ rotations. 
\end{proof}

We now define the interval that will be assigned to each rotation of $I$.

\begin{dfn}
  \label{dfn:extent}
  Let $\rho = (m_0, w_0), (m_1, w_1), \ldots, (m_{\ell - 1}, w_{\ell-1})$ be a rotation of $I \in \range(k,n)$ with minrank function $\orank$.  Define the \dft{extent} of $\rho$, denoted $\ext(\rho)$,  as the interval
  \[
  \left[\rmin(\rho) - 2 k + 1, \rmax(\rho) + 2 k - 1\right]
  \]
  where
  \[
  \rmin(\rho) = \min_{0 \leq i \leq \ell - 1} \set{\orank(m_i), \orank(w_i)} \quad\text{and}\quad \rmax(\rho) = \max_{0 \leq i \leq \ell - 1} \set{\orank(m_i), \orank(w_i)}.
  \]
\end{dfn}

\begin{lem}
  \label{lem:extent}
  Let $I \in \range(k, n)$ and $G(I)$ be the rotation digraph of $I$.  
  Suppose $(\rho, \sigma)$ is a directed edge in $G(I)$. Then there exists $i \in [n]$ such that $i \in \ext(\rho) \cap \ext(\sigma)$.
\end{lem}
\begin{proof}
 Let $\orank$ be the minrank function. The DAG $G(I)$ contains an edge $(\rho, \sigma)$ only when Rule~1 or~2 of Remark~\ref{rem:dag-structure} is satisfied. We consider the cases corresponding to these rules below.
  \begin{description}
  \item[Case 1.] If Rule~1 applies, then there exists $m \in M$ and $w, w' \in W$ such that $(m, w) \in \rho$ and $(m, w') \in \sigma$. Since $m$ appears in both $\rho$ and $\sigma$, we have $\orank(m) \in \ext(\rho) \cap \ext(\sigma)$, as desired.
  \item[Case 2.] Suppose Rule~2 applies. Then there exists $m$ and $w$ such that $w$ is part of the rotation $\rho$ while $m = m_{i}$ is  in the rotation $\sigma =  (m_1, w_1), (m_2, w_2), \ldots, (m_\ell, w_\ell)$ and
    \begin{equation}
      \label{eqn:relative-ranks}
      P_{m}(w_i) < P_{m}(w) < P_m(w_{i+1}).
    \end{equation}
    We know that $\orank(w_i) \leq P_m(w_i)$.  Since $w_i$ is a stable partner of $m$, by  Lemma~\ref{lem:k-range-match}
      \begin{equation}
      \label{eqn:extreme-ranks}
      \orank(m) - 2 k + 2 \leq P_m(w_i).
      \end{equation}
    On the other hand, $w_{i+1}$ is also a stable partner of $m$. From the proof of Lemma~\ref{lem:k-range-match} we also know that 
\begin{equation}
      \label{eqn:extreme-ranks-2}
      P_m(w_{i+1}) \leq \orank(m) + 2 k - 2.
    \end{equation}
    
    Combining~(\ref{eqn:relative-ranks}), (\ref{eqn:extreme-ranks}) and (\ref{eqn:extreme-ranks-2})  gives
    \begin{equation}
      \label{eqn:rank-bounds}
      \abs{P_m(w) - \orank(m)} \leq 2 k - 2.
    \end{equation}
    Finally, applying the triangle inequality and the definition of $k$-range preferences, we obtain
    \begin{align*}
      \abs{\orank(w) - \orank(m)} &\leq \abs{\orank(w) - P_m(w)} + \abs{P_m(w) - \orank(m)}\\
      &\leq (k - 1) + (2 k - 2) \\
      & = 3k-3.
    \end{align*}
    Since $[\orank(w) - 2k+1, \orank(w) + 2k-1] \subseteq \ext(\rho)$ while $[\orank(m) - 2k+1, \orank(m) + 2k-1] \subseteq \ext(\sigma)$,  $\abs{\orank(w) - \orank(m)} \leq 3 k - 3$ implies that  $\ext(\rho) \cap \ext(\sigma) \neq \varnothing$.
  \end{description}
  Thus, for all edges $(\rho, \sigma)$ in $G$, we have $\ext(\rho) \cap \ext(\sigma) \neq \varnothing$, as desired.  Now we note that the non-empty intersection can happen in three ways: 
  \begin{itemize}
  \item $\left[\rmin(\rho), \rmax(\rho) \right] \cap \left[\rmin(\sigma), \rmax(\sigma)\right] \neq \varnothing$ 
  \item  $[\rmax(\rho), \rmax(\rho) + 2k -1] \cap [\rmin(\sigma) - 2k + 1, \rmin(\sigma)] \neq \varnothing$ or 
  \item  $[\rmax(\sigma), \rmax(\sigma) + 2k -1] \cap [\rmin(\rho) - 2k + 1, \rmin(\rho)] \neq \varnothing$. 
 \end{itemize}
  In each case, some $i \in [n]$ has to be part of this intersection since the minranks of all agents are in $[n]$. 
\end{proof} 

\begin{lem}
  \label{lem:limited-overlap}
  Let $I \in \range(k, n)$ and integer $i \in [n]$. Then there are at most $50 k^2$ rotations $\rho$ such that $i \in \ext{\rho}$.
\end{lem}
\begin{proof}
  Let $\orank$ be the minrank function. Consider the set $\{\rho_1, \rho_2, \ldots, \rho_\ell\}$ of rotations whose extents contain $i$. That is, for all $j$ we have $i \in \ext{\rho_j}$. By the definition of $\ext$, for each $\rho_j$, there exists an agent $a_j$  in $\rho_j$ such that   
  \begin{equation}
    \label{eqn:rank-range}
    \abs{\orank(a_j) - i} \leq 2 k - 1    
  \end{equation}
  Following the proof of Corollary~\ref{cor:few-rotations}, the number of women that satisfy~(\ref{eqn:rank-range}) is at most $5k-4$.  Similarly, the number of men that satisfies the same inequality is at most $5k-4$. Thus, at most $10k-8$ agents satisfy the inequality. But each such agent can be part of at most $5k-5$ rotations.  Therefore, $\ell \leq (10k-8)(5k-5)  = 50k^2 -90k + 40 \leq 50 k^2$ when $k \ge 1$.    
\end{proof}

We now present the  main result of this section.  

\begin{lthm}
  \label{thm:k-range-upper}
  Suppose $I \in \range(k,n)$ and $\calR = \calR(I)$ is its rotation poset. Then $\pw(\calR) \leq 50 k^2$. 
\end{lthm}
\begin{proof}
  Let $I \in \range(k)$ have minrank function $\orank$.  Let $G = G(I) = (V, E)$. Consider the sequence $\calX = (X_1, X_2, \ldots, X_n)$ defined by
  \[
  X_i = \set{\rho \in V \sucht i \in \ext(\rho)}.
  \]
  We claim that $\calX$ is a path decomposition of $G$. We argue that each of the three conditions in Definition~\ref{dfn:path-decomp} are satisfied:
  \begin{enumerate}
  \item $\bigcup_{i = 1}^n X_i = V$ holds because each rotation $\rho \in V$ has a non-empty extent, and intersects the interval $\set{1, \ldots, n}$.
  \item For each edge $\set{\rho, \sigma} \in E$, there exists $i \in [n]$ such that $\rho, \sigma \in X_i$. This holds by Lemma~\ref{lem:extent}.
  \item For all $i, j, k \in [n]$ with $i \leq j \leq k$, we have $X_i \cap X_k \subseteq X_j$. To see this, suppose $\rho \in X_i \cap X_k$. Then $i, k \in \ext(\rho)$ by the definitions of $X_i$ and $X_k$. Therefore, $\rho \in X_j$, because $\ext(\rho)$ is an interval, as desired.
  \end{enumerate}
  By Lemma~\ref{lem:limited-overlap}, $\calX$ has width at most $50 k^2$. Thus, $G$  has pathwidth at most $50 k^2$.  And since the Hasse diagram of $\calR$, $H(\calR)$, is a subgraph of $G$, it follows that $H(\calR)$ has a path decomposition whose width is at most $50k^2$.  Thus,  $\calR$ also has pathwidth at most $50 k^2$, the desired bound.
\end{proof}

\begin{cor}
 \label{cor:no-constant}
There is no constant $c$ such that every poset can be realized by an instance in $\range(c)$. 
\end{cor}
\begin{proof}
For any positive integer $n$, there is a tree $T_n$ whose pathwidth is at least $n$~\cite{Scheffler1990-linear}. Take any vertex $z$ of $T_n$ and orient all edges away from $z$ and call the new directed acyclic graph $\vec{T}_n$. Notice that the transitive reduction\footnote{Recall that the \emph{transitive reduction} of a DAG $G = (V, E)$ is the DAG $G' = (V, E')$ formed by removing all edges $(u, v)$ from $E$ such that there exists a path of length at least $2$ from $u$ to $v$ in $G$.} of $\vec{T}_n$ is itself since $T_n$ contains no cycles. Let $\calP_n$ be a poset whose Hasse diagram is $\vec{T}_n$.  By definition, the pathwidth of $\calP_n$ is equal to the pathwidth of $\vec{T}_n$, which is at least $n$.   Let $c$ be a fixed constant.  By Theorem~\ref{thm:k-range-upper}, every SM instance $I$ in $\range(c)$ has a rotation poset whose pathwidth is at most $50 c^2$.  Thus, if $n > 50 c^2$, no instance in $\range(c)$ can realize $\calP_n$. 
\end{proof}

Given an SM instance $I$, Algorithm $\ConstructPathDecomposition$ incorporates the steps from the discussion above to create a nice path decomposition for $G(I)$, the rotation digraph of $I$, whose width is $O(k^2)$ where $k = \range(I)$.

\begin{algorithm}
  \caption{$\ConstructPathDecomposition(I=(M,W,P))$. Construct rotation digraph $G(I)$ and nice path decomposition $\calX = (X_1, X_2, \hdots, X_n)$ of width $O(k)$ where $k = \range(I)$.}
  \begin{algorithmic}[1]
    \STATE  Compute $G(I)$ using Gusfield's algorithm (Theorem~\ref{thm:stable-closed}).
    \STATE  Compute $k = \range(I)$ and the minrank function $\orank$ (Lemma~\ref{lem:compute-range}).
    \FORALL{rotations $\rho \in G(I)$}
    \STATE $\orank_{\min}(\rho)  \leftarrow \min \set{\orank(a) \sucht a \in \rho}$
    \STATE $\orank_{\max}(\rho)  \leftarrow \max \set{\orank(a) \sucht a \in \rho}$
    \STATE $\ext(\rho) \leftarrow [\orank_{\min}(\rho) - 2k+1, \orank_{\max}(\rho) + 2k - 1]$
    \ENDFOR
    \STATE  Initialize all the sets in $(X_1, X_2, \hdots, X_n)$ to empty sets.
    \FORALL{rotations $\rho \in G(I)$}
    \FORALL{$i \in \ext(\rho) \cap [1,n]$}
    \STATE Add $\rho$ to $X_i$
    \ENDFOR
    \ENDFOR
    \STATE Convert $(X_1, X_2, \hdots, X_n)$ into a nice path decomposition (Lemma~\ref{lem:nice-path}).
  \end{algorithmic}
\end{algorithm}

\begin{cor}
  \label{cor:k-range-decomp}
  Given an SM instance $I$ of size $n$, there 
  exists an $O(k^2 n + n^2)$-time algorithm that computes the rotation digraph $G(I)$ and a nice path decomposition $\calX$ of $G(I)$ with width at most $50k^2$, where $k = \range(I)$.  
\end{cor}
\begin{proof}
Gusfield's algorithm takes $O(n^2)$ time to compute $G(I)$ (Theorem~\ref{thm:stable-closed}).  Computing $k$ and objective function $\orank$ also takes $O(n^2)$ time (Lemma~\ref{lem:compute-range}). Determining $\ext(\rho)$ takes $O(|\rho|)$ time so doing this for all rotations takes $O(\sum_{\rho} |\rho|) = O(n^2)$ because each man-woman pair can be in at most one rotation.  Computing the path decomposition $(X_1, X_2, \hdots, X_n)$ takes $O(k^2n)$ since each $|X_i| \leq 50k^2$ according to Theorem~\ref{thm:k-range-upper}. Finally converting $(X_1, X_2, \hdots, X_n)$ to a nice path decomposition takes $O(k^2 n)$ time (Lemma~\ref{lem:nice-path}).  Thus, the total running time of $\ConstructPathDecomposition$ is $O(k^2n + n^2)$. 
\end{proof}

\section{Algorithmic Implications}
\label{sec:k-range-consequences}

We now describe FPT algorithms for several computationally hard stable matching problems parameterized by $\range(I)$. The results rely on Theorem~\ref{thm:k-range-upper} and Corollary~\ref{cor:k-range-decomp}, as well as FPT algorithms in~\cite{Cheng2020-simple, Gupta2017-treewidth}.

\subsection{Counting and Sampling Stable Matchings}

In a companion paper~\cite{Cheng2020-simple}, we prove the following results:

\begin{lthm}[\cite{Cheng2020-simple}]
  \label{thm:downsets}
  Let $G$ be a directed acyclic graph with $n$ vertices.  Given $G$ and a simple path decomposition $\calX$ of $G$ of width $r$, there is an algorithm that computes the number of closed subsets of $G$ in $O(2^{r} r n )$ time. 
\end{lthm}

\begin{lthm}[\cite{Cheng2020-simple}]
  \label{thm:sampling-downsets}
  Let $G$ be a directed acyclic graph with $n$ vertices.  Given $G$ and a simple path decomposition $\calX$ of $G$ of width $r$, there is an algorithm that returns a downset of $G$ sampled uniformly at random in time $2^{O(r)} n^{O(1)}$.
\end{lthm}

\begin{cor}
  \label{cor:k-range-count}
  Let $I$ be an SM instance of size $n$ and  $k = \range(I)$.  The number of stable matchings of $I$ can be computed in  $O(2^{50k^2}k^2n + n^2)$ time. 
  \end{cor}
\begin{proof}
  By Corollary~\ref{cor:k-range-decomp}, we can compute the rotation digraph $G(I)$ and a simple path decomposition of $G(I)$ of width at most $50 k^2$ in $O(k^2 n + n^2)$ time.  By Theorem~\ref{thm:stable-closed}, the number of stable matchings for $I$ is $\down(G(I))$, which we can compute in time $O(2^{50k^2}k^2 n)$  by Theorem~\ref{thm:downsets}.   Thus, computing the number of stable matchings of $I$ can be done in $O(2^{50k^2}k^2n + n^2)$ time.
\end{proof}

\begin{cor}
  \label{cor:k-range-sample}  Let $I$ be an SM instance of size $n$ and  $k = \range(I)$. There is an algorithm that samples the stable matchings of $I$ uniformly at random in $2^{O(k^2)} n^{O(1)}$ time.  
\end{cor}
\begin{proof}
Again, we compute $G(I)$ and a simple path decomposition of $G(I)$ of width at most $50k^2$ in $O(k^2 n + n^2)$ time. 
  By Theorem~\ref{thm:sampling-downsets}, we can sample a downset $Z$ of $G(I)$ uniformly at random in time $2^{O(k^2)} n^{O(1)}$.  Since there is a  one-to-one correspondence between the downsets of $G(I)$ and the stable matchings of $I$, it follows that the stable matching $\mu$ formed by eliminating all the rotations in $Z$ from the man-optimal stable matching $\mu_M$ is a stable matching of $I$ chosen uniformly at random.  Computing $\mu$ from $\mu_M$ and $Z$ takes $O(n^2)$ time.
\end{proof}

\subsubsection{Finding Median Stable Matchings}

For SM instance $I$, let $\calM(I)$ contain all the stable matchings of $I$ and $N = | \calM(I)|$. For each stable matching $\mu \in \calM(I)$ and for each agent $a$, let $\mu(a)$ denote $a$'s partner in $\mu$.  Sort the multiset of agents $\{ \mu(a), \mu \in \calM(I)\}$  from $a$'s most preferred partner to $a$'s least preferred partner. The agent(s) in the middle of this sorted list is referred to as the \dft{(lower or upper) median stable partner of $a$}.\footnote{If the number of agents is odd, there is a unique median stable partner, whereas if the number of agents is even, there are two median stable partners: the upper and lower stable partner.} A \dft{median stable matching} is a stable matching where every agent is paired with a median stable partner.

Teo and Sethuraman~\cite{Teo1998-geometry} were the first to recognize that every SM instance $I$ has a median stable matching.  Cheng~\cite{Cheng2008-generalized, Cheng2010-understanding} showed that median stable matchings of $I$ are also remarkable in that they are exactly the median elements in the lattice of stable matchings of $I$.  That is,  their average distance to all the stable matchings of $I$ is the least.\footnote{In a poset $\calP$, the distance between two elements $x$ and $y$, $d(x,y)$, is defined as the length of the shortest path between $x$ and $y$ in the undirected version of $H(\calP)$.   Thus, the total distance of $x$ to all elements of $\calP$ is $\sum_y d(x,y)$ and the average distance of $x$ to these elements are computed similarly.} Thus, median stable matchings are fair in a very strong sense. She also characterized these stable matchings in terms of the rotations of $I$ and used it to show that computing a median stable matching is $\shP$-hard. 

\begin{lthm}
[Cheng~\cite{Cheng2008-generalized, Cheng2010-understanding}]
 \label{thm:medianSM}
Suppose an SM instance $I$ has $N$ stable matchings.  For each rotation $\rho$ of $I$, let $n_\rho$ denote the number of downsets of $G(I)$ that contain $\rho$.    When $N$ is odd, $I$ has only one  median stable matching and it corresponds to the downset 
\[\set{ \rho \sucht n_\rho \ge (N+1)/2}.\]   
On the other hand, when $N$ is even,  every median stable matching of $I$ corresponds to the downset of the form
\[
\set{\rho \sucht n_\rho \ge  (N/2) + 1} \cup S \text{ where } S \subseteq \set{\rho \sucht n_\rho = N/2}.
\]
\end{lthm}

 \begin{cor}
 Given SM instance $I$ of size $n$,  a median stable matching of $I$ can be computed in $2^{O(k^2)} n^{O(1)}$, where $k = \range(I)$. 
 \end{cor}  
\begin{proof}
Our starting point is once again the rotation digraph $G(I)$ and a nice path decomposition $\calX$ of $G(I)$ whose width is $O(k^2)$.   These structures can be constructed in $O(k^2n + n^2)$ time.  By Corollary~\ref{cor:k-range-count}, we can then compute $N$, the number of stable matchings of $I$ in $2^{O(k^2)} n^{O(1)}$ time.

Next, for each rotation $\rho$, let $\Anc(\rho)$ consist of $\rho$ and its ancestors in $G(I)$.   Notice that every downset of $G(I)$ that contains $\rho$ is of the form $\Anc(\rho) \cup T$ where $T$ is a downset of $G(I) \setminus \Anc(\rho)$.  Thus, $n_\rho = \down(G_\rho)$, where    $G_\rho = G(I) \setminus \Anc(\rho)$.   To compute $n_\rho$, we first construct $G_\rho$ and a nice path decomposition $\calX_\rho$ for $G_\rho$ from $G(I)$ and $\calX$.  This will take $O(n^2 + kn)$ time.  Then we use them to compute $\down(G_\rho)$ in $2^{O(k^2)} n^{O(1)}$.  So the total time for computing $n_\rho$ is $2^{O(k^2)} n^{O(1)}$.   

We now have the ingredients for computing a median stable matching of $I$.  We start by finding $\mu_M$, the man optimal stable matching.  Then we determine the set $Z= \set{\rho \sucht n_\rho \ge (N+1)/2}$.  By Theorem~\ref{thm:medianSM},  the stable matching obtained by eliminating $Z$ from $\mu_M$ is a median stable matching.  These steps take $O(n^2)$ time.  

The bottleneck for our procedure is the computation of $N$ and $n_\rho$, $\rho \in G(I)$.  It follows that the running time is $2^{O(k^2)} n^{O(1)}$. \end{proof}  

\subsubsection{Finding Sex-Equal and Balanced Stable Matchings}

Given a stable matching $\mu$ of $I$, the satisfaction of an agent $a$ with $\mu$ can be represented by $P_a(\mu(a))$, the rank of $a$'s partner in $\mu$.  The \dft{total satisfaction of the men and women with $\mu$}  then are $S_M(\mu) = \sum_{m \in M} P_m(\mu(m))$ and $S_W(\mu) = \sum_{w \in W} P_w(\mu(w))$ respectively.   A natural strategy for computing a fair stable matching of $I$ is to create a measure for $\mu$ that captures how fair $\mu$ is to the agents of $I$.

In the \dft{sex-equal stable marriage problem (SESM)},  the fairness measure of  $\mu$ is
\[
\delta(\mu) = |S_M(\mu) - S_W(\mu)|
\]
and the goal is to find a stable matching of $I$ whose SESM fairness measure is as small as possible.  Kato~\cite{Kato1993-complexity} showed that SESM is NP-hard. 

In the \dft{balanced stable marriage problem (BSM)}, the fairness measure of $\mu$ is 
\[
\beta(\mu) = \max\set{ S_M(\mu), S_W(\mu)}.
\]
The goal this time is to find a stable matching of $I$ whose BSM fairness measure is as small as possible. Feder~\cite{Feder1995-stable} showed that BSM is NP-hard. 

Gupta et al.~\cite{Gupta2017-treewidth} recently considered the parametrized complexity of several hard variants of the SMP parametrized by the treewidth of the rotation poset.  Specifically, they proved the following result.

\begin{lthm}[Gupta et al.~{\cite[Theorem~3]{Gupta2017-treewidth}}]
  \label{thm:gupta-treewidth}
  The sex equal and balanced stable matching problems are both solvable in time $O(2^{\mathrm{tw}} n^6)$ where $\mathrm{tw}$ is the treewidth of the rotation poset of the SM instance of size $n$. 
\end{lthm}

Since the pathwidth of a graph is an upper bound to its treewidth, Theorem~\ref{thm:gupta-treewidth} together with Theorem~\ref{thm:k-range-upper} immediately gives the following corollary.

\begin{cor}
  \label{cor:sex-equal-balanced}
   Let $I$ be an SM instance of size $n$ and  $k = \range(I)$. Then computing a sex-equal and balanced stable matching for $I$ can be done in time $2^{O(k^2)} n^6$.
\end{cor}

 \section{Discussion and Questions}
\label{sec:conclusions}

\paragraph{Counting stable matchings in graphs}

Our $k$-bounded model is equivalent to SM instances defined on (bipartite) graphs with maximum degree at most $k$. Thus Theorem~\ref{thm:k-bound} implies that $\shSM$ remains $\shBIS$-complete even when restricted to graphs with maximum degree $3$. Since graphs with maximum degree $2$ (i.e., $2$-bounded instances) are disjoint unions of paths and cycles, $\shSM$ can be solved on $2$-bounded instances in polynomial time. (Paths have unique stable matchings, and each cycle can support $1$ or $2$ stable matchings.) Thus $\shSM$ has a sharp hardness threshold in the $k$-bounded model at $k = 3$. Interestingly, $\shBIS$ shows a similar threshold around $k = 6$: counting independent sets in bipartite graphs of maximum degree $k \geq 6$ is $\shBIS$-complete~\cite{Cai2016-bis}, while instances of maximum degree $k \leq 5$ admit an FPRAS, even if vertices on only one side of the bipartition obey the degree bound~\cite{Liu2015-fptas}. Recently, Curticapean et al.~\cite{Curticapean2019-fixed} considered other parameterized variants of $\shBIS$, and showed that the variants admit efficient fixed-parameter algorithms in bounded-degree graphs.

It would be interesting to see if $\shSM$ becomes easy in any ``natural'' restricted graph class.

\begin{que}
  For what families of (bipartite) graphs can $\shSM$ be computed or approximated efficiently? Can $\shSM$ be efficiently computed in planar graphs? Graphs of bounded genus? What rotation posets can be realized by these restricted instances?
\end{que}

It seems likely that planar SM instances (i.e., instances where the graph $G$ of acceptable partners is planar) realize a restricted family of rotation posets. To see why, note that every rotation corresponds to a cycle in $G$, and a rotation $\rho$ can be immediate predecessor of $\rho'$ only if either the corresponding cycles intersect, or there is an edge connecting a node in one cycle to a node in another. This restriction seems to limit what rotation posets are realizable by planar instances. For example, we believe that a poset $\calP$ whose Hasse diagram is a sufficiently large complete bipartite graph minus a perfect matching cannot be realized by any planar instance. Does the family of rotation posets realized by planar SM instances admit a clean description? Can such a characterization be used to solve structural SM problems more efficiently than the general case?

\paragraph{Incomplete preferences and ties}

The $k$-attribute, $(k_1, k_2)$-list, and $k$-range preference models can also be used to describe \dft{SMI instances}---that is, SM instances with incomplete preferences. Specifically, an SMI instance arising from one of these models is an instance for which every preference list is a sub-list of an instance in the model with complete preferences. Our construction for the $(k_1, k_2)$-list model in Section~\ref{sec:k-list} has the property that there are two master preferences for each gender such that all (incomplete) preference lists are sub-lists of one of the two master lists. Thus, \emph{every} finite poset is realizable in the SMI variant of the $(2, 2)$-list model. Since complete preferences are a special case of SMI instances, our $k$-attribute construction also implies that SMI $k$-attribute preferences realize all finite posets.

\begin{que}
  What rotation posets can be realized in the SMI variant of the $k$-range model?
\end{que}

These restricted preference models can also be generalized to allow for preferences with ties: that is, an agent may be indifferent to a choice of several partners. It is well-known that many stable marriage problems become intractible in instances with both incomplete preferences and ties (see, e.g.,~\cite{Manlove2002-hard}).

\begin{que}
  Consider SM instances with ties and incomplete preferences. What stable matching problems remain NP-hard when restricted to instances arising from the $k$-attribute, $(k_1, k_2)$-list, or $k$-range models?
\end{que}

\paragraph{\texorpdfstring{$k$}{k}-attribute and Euclidean Preferences}

Theorem~\ref{thm:k-attr} characterizes the rotation posets realized in the $k$-attribute model for $k \geq 6$, but does not say anything about $\attr(k)$ for $k \leq 5$.

%\cnote{Remember, Chebolu et al characterized rotation posets for $k = 1$.}

\begin{que}
  What rotation posets are realized by $\attr(k)$ for $k = 2, 3, \ldots, 5$?
\end{que}

We conjecture that $\attr(k)$ realizes arbitrary rotation posets for any $k \geq 3$. This conjecture is consistent with the results of Chebolu et al.~\cite{Chebolu2012-complexity}, who show that $\shSM$ is $\shBIS$-complete in the $k$-attributed model for any $k \geq 3$. Our use of cyclic polytopes in Section~\ref{sec:k-attribute} to establish Theorem~\ref{thm:k-attr} is wasteful in the sense that cyclic polytopes give much stronger guarantees than our result actually requires. Indeed, the construction allows us to embed \emph{any} $2$-dimensional simplicial complex into $\R^6$ in convex position. However, our argument only requires that the $2$-dimensional faces of the complexes corresponding to preference lists output by $\ConstructInstance$ be embedded in convex position. Thus it seems that a more careful analysis of the output of $\ConstructInstance$ could establish Theorem~\ref{thm:k-attr} for smaller values of $k$.

Chebolu et al.~\cite{Chebolu2012-complexity} and K\"unnemann et al.~\cite{Kunnemann2019-subquadratic} also consider the $k$-Euclidean preference model defined by Bogomolnaia and Laslier~\cite{Bogomolnaia2007-euclidean}. In this model, each agent is also associated with a point in $\R^k$. Preferences are determined by the Euclidean distance between the points: $a$ prefers agents in order $b_1, b_2, \ldots, b_n$ where $b_1$ is $a$'s nearest agent, $b_2$ the second nearest, and so on. Chebolu et al.~\cite{Chebolu2012-complexity} showed that $\shSM$ is $\shBIS$ complete in the $k$-Euclidean model for $k \geq 2$.

\begin{que}
  What rotation posets are realized in the $k$-Euclidean model?
\end{que}

We suspect that $k$-Euclidean preferences realize arbitrary posets for a relatively small fixed constant $k$.

K\"unnemann et al.~\cite{Kunnemann2019-subquadratic} also consider an asymmetric version of the $k$-Euclidean model in which each agent $a$ has two associated points in $R^k$: a \emph{position} $\vecx_a$ and an \emph{opinion} $\vecy_a$. The agent $a$ ranks agents in order of their position's distance to $a$'s \emph{opinion}. That is $a$ ranks $b_1$ with $\vecy_{b_1}$ closest to $\vecx_a$ first, and so on.

Our $6$-attribute construction can be modified to realize arbitrary rotation posets in the asymmetric $6$-Euclidean model. In the modified construction, the positions of each $m \in M$ and $w \in W$ are the same as the embedding in Section~\ref{sec:k-attribute} (i.e., each agent corresponds to a vertex of a cyclic polytope). If $m_1, m_2, m_3$ appear on $w$'s preference list in the $3$-bounded instance constructed by $\ConstructInstance$, then she can choose her opinion $\vecx_w$ so that (the positions of) $\vecx_{m_1}$, $\vecx_{m_2}$, and $\vecx_{m_3}$ are the three closest positions to $\vecx_w$ as follows. Let $\vecc_w$ denote the circumcenter of the triangle in the plane determined by $\vecx_{m_1}$, $\vecx_{m_2}$, and $\vecx_{m_3}$, and let $\vecbeta_w \in \R^6$ be the vector defined immediately after~(\ref{eqn:poly}), so that $-\vecbeta_w$ is an outward facing normal vector for the plane containing $\vecx_{m_1}, m_2$, and $m_3$. Thus, for any point $p \in \R^6$ of the form $\vecy = \vecc_w - t \vecbeta_w$ with $t \in \R$, the points $m_1, m_2$ and $m_3$ are all equidistant from $\vecy$. By choosing $t$ sufficiently large, we can ensure that, in fact, $\vecx_{m_1}$, $\vecx_{m_2}$ and $\vecx_{m_3}$ are the three closest points in $M$ to $\vecy$. By taking $\vecy_w$ to be a slight perturbation of such a $\vecy$, we can ensure that for all $m \neq m_1, m_2, m_3$, we have
\[
\abs{\vecx_{m_1} - \vecy_w} < \abs{\vecx_{m_2} - \vecy_w} < \abs{\vecx_{m_3} - \vecy_w} < \abs{\vecx_m - \vecy_w}.
\]
That is, $w$ ranks $m_1$, $m_2$, and $m_3$ (in that order) ahead of all other men. As with the $6$-attribute case, repeating the procedure for all men and women is sufficient to prove that the corresponding asymmetric $6$-Euclidean instance realizes the same rotation poset as the underlying $3$-bounded instance.

\paragraph{\texorpdfstring{$k$}{k}-list Preferences}

While Theorem~\ref{thm:k-list} characterizes the rotation posets realizable by $(k_1, \infty)$-list SM instances, it does not say anything about $(k_1, k_2)$-list preferences for constant $k_2$, nor the original $k$-list preference structure of Bhatnagar et al.~\cite{Bhatnagar2008-sampling}.

\begin{que}
  What rotation posets can be realized in the $(k_1, k_2)$-list model for $k_1, k_2 < n$?
\end{que}

We suspect that for $k_1 = k_2 = 2$, the rotation poset is restricted, but we are unsure of how to prove even this seemingly simple fact. We note that K\"unnemann et al.~\cite{Kunnemann2019-subquadratic} had similar difficulty with the $k$-range model, and posed as an open question whether or not $2$-list preferences admit an $o(n^2)$-time algorithm for finding a stable matching.

\paragraph{\texorpdfstring{$k$}{k}-range Preferences}

Our construction in the proof of Theorem~\ref{thm:k-range-upper} shows that every $k$-range instance has a rotation poset whose pathwidth is at most $50 k^2$. Can this bound be improved?

\begin{que}
  What is the maximum pathwidth of $\calR(I)$ for $I \in \range(k)$? Is it $\Omega(k^2)$? Or can our $O(k^2)$ upper bound be imporved to, say, $O(k)$?
\end{que}

\paragraph{Distributed Stable Matchings}

The stable marriage problem has a natural interpretation as a distributed problem, where each agent is represented by a processor and agents communicate via point-to-point communication. In fact, the Gale-Shapley algorithm~\cite{Gale1962-college} has a natural interpretation in such a computational model. It is straightforward to show that finding a stable matching requires a number of rounds proportional to the network diameter, even when there are no constraints on local computation and communication (i.e., in the LOCAL model described in~\cite{Peleg2000-distributed}). If bathdwidth is restricted to $O(\log n)$ bits per edge per communication round (i.e., the CONGEST model~\cite{Peleg2000-distributed}), finding a stable matching still requires $\Omega(\sqrt{n})$ rounds, even in networks of diameter $O(\log n)$~\cite{Kipnis2009-note}.  Can stable matchings be found faster in distributed models of computation if the preferences are restricted?

\begin{que}
  Consider the distributed stable marriage problem in which each agent is represented by a processor. How many communication rounds are needed to find a stable matching if the preferences are assigned according to the $k$-attribute, $(k_1, k_2)$-list, or $k$-range models?
\end{que}

We note that ``almost stable'' matchings can be computed in $O(1)$ rounds in bounded-degree networks~\cite{Floreen2010-almost} and $\log^{O(1)}(n)$ rounds in general networks~\cite{Ostrovsky2015-fast}. The work of Khanchandani and Wattenhofer~\cite{Khanchandani2017-distributed} shows that stable matchings in the $k$-range model can be computed with less total communication than the Gale-Shapley algorithm.

%% \section*{Acknowledgments}

%% The second author thanks D\'aniel Marx and S\'andor Kisfaludi-Bak for valuable discussions related pathwidth-bounded DAGs.

\urlstyle{same}
\bibliographystyle{plainnat}
\bibliography{restricted-stable-matchings}

%\appendix

%\input{downsets}

\end{document}